\documentclass{article}
\usepackage{fullpage}

\usepackage{epsf}
\usepackage{fancyhdr}
\usepackage{graphics}
\usepackage{graphicx}
\usepackage{psfrag}
\usepackage[dvipsnames]{xcolor}
\usepackage[colorlinks=true, allcolors=blue]{hyperref}
\usepackage{footnote}
\makesavenoteenv{tabular}

\usepackage{amsthm}
\usepackage{amsfonts}
\usepackage{amsmath}
\usepackage{amssymb}
\usepackage{bbm}

\usepackage{caption}
\usepackage{subcaption}
\usepackage[numbers]{natbib}

\usepackage{tikz}
\usetikzlibrary{shapes.geometric, arrows}

\usepackage{algorithm}
\usepackage{algorithmic}
\makeatletter
\def\plist@algorithm{Alg.\space}
\makeatother

\usepackage{dsfont}
\usepackage{xspace}
\usepackage{pifont}

\newcommand{\pos}{\mathrm{Pos}\xspace}
\newcommand{\me}{\mathbb{E}}
\newcommand{\one}{\mathds{1}\xspace}
\newcommand{\pei}{$\mathrm{I}^3$\xspace}
\newcommand{\peiZero}{Crossfit-$\mathrm{I}^3$\xspace}
\newcommand{\peiNeg}{MaY-$\mathrm{I}^3$\xspace}
\newcommand{\peiZeroKnowP}{Crossfit-$\mathrm{I}^3_{\pi}$\xspace}
\newcommand{\peiNegKnowP}{MaY-$\mathrm{I}^3_{\pi}$\xspace}
\newcommand{\peiZeroUnknowP}{Crossfit-$\mathrm{I}^3_{\widehat \pi}$\xspace}
\newcommand{\peiNegUnknowP}{MaY-$\mathrm{I}^3_{\widehat \pi}$\xspace}

\newcommand\independent{\protect\mathpalette{\protect\independenT}{\perp}}
\def\independenT#1#2{\mathrel{\rlap{$#1#2$}\mkern2mu{#1#2}}}
\DeclareMathOperator*{\argmin}{argmin}

\newtheorem{theorem}{Theorem}
\newtheorem{lemma}{Lemma}
\newtheorem{corollary}{Corollary}

\newtheorem{remark}{Remark}

\begin{document}
 
\begin{center}

{\bf{\LARGE{Interactive identification of individuals\\ 
  with positive treatment effect\\ 
  \vspace{4pt}
  while controlling false discoveries}}}

\vspace*{.2in}

{\large{
\begin{tabular}{cccc}
Boyan Duan$^1$ & Larry Wasserman$^2$ & Aaditya Ramdas$^3$\\
\end{tabular}\\
$^1$\texttt{boyand@alumni.cmu.edu} 
$^2$\texttt{larry@stat.cmu.edu} 
$^3$\texttt{aramdas@stat.cmu.edu}
}}

\vspace*{.2in}

\begin{tabular}{c}
$^{1, 2, 3}$ Department of Statistics and Data Science \footnote{BD now works at Google LLC.},\\
Carnegie Mellon University,  Pittsburgh, PA  15213 
\end{tabular}

\vspace*{.2in}

\today

\vspace*{.2in}

\begin{abstract}
{Out of the participants in a randomized experiment with anticipated heterogeneous treatment effects, is it possible to identify which subjects have a positive treatment effect?
While subgroup analysis has received attention, claims about individual participants are much more challenging.
We frame the problem in terms of multiple hypothesis testing: each individual has a null hypothesis (stating that the potential outcomes are equal, for example) and we aim to identify those for whom the null is false (the treatment potential outcome stochastically dominates the control one, for example). 
We develop a novel algorithm that identifies such a subset, with nonasymptotic control of the false discovery rate (FDR). Our algorithm allows for interaction --- a human data scientist (or a computer program) may adaptively guide the algorithm in a data-dependent manner to gain power. We show how to extend the methods to observational settings and achieve a type of doubly-robust FDR control. We also propose several extensions: (a) relaxing the null to nonpositive effects, (b) moving from unpaired to paired samples, and (c) subgroup identification. We demonstrate via numerical experiments and theoretical analysis that the proposed method has valid FDR control in finite samples and reasonably high identification power.}
\end{abstract}
\end{center}

\section{Introduction}
Subgroup identification --- or identifying subgroups of the population that have some positive response to a treatment --- has been a major topic in the clinical trial community and the causal literature (see \citet{lipkovich2017tutorial,powers2018some,loh2019subgroup} and references therein). Typically, the treatment effect in the investigated population varies with the subject's gender, age, and other covariates. Identifying subjects with positive effects can help guide follow-up research and provide medical guidance. However, most existing methods do not have an error control guarantee at the level of the individual --- it is possible that most subjects in the identified subgroup do not have positive effects. For example, an identified subgroup could be ``female subjects younger than 40'', which may typically mean that  the average treatment effect is positive in this subgroup, but it is possible that only 10\% of them with age between 18 and 20 may truly have a positive treatment effect. 

This paper considers the identification of positive effects at an individual level: among the participants in a trial, which ones have the treatment potential outcome larger than control potential outcome (we call this the subject's treatment effect)? This is a hard question to answer in general, because we only observe one or the other potential outcome. But we will give a nontrivial answer that may be powerless in the worst case, but powerful in cases where the covariates are informative about the treatment effects, and never making too many false claims. 

For each subject, the available data we have includes its treatment assignment, covariates, and observed outcome. We will identify a set of subjects as having positive effects, with an error control on false identification, and without any modeling assumption on how the (treated/control) outcomes vary with covariates. Below is an example to visualize the problem, and prepare for our solution.


\subsection{An illustrative example} \label{sec:intro_example}
Consider an experiment involving 500 subjects, each is assigned to treatment or control independently with probability $1/2$. Suppose each subject is associated with two simple covariates, which are independently and uniformly distributed in $[0,1]$. All the data is shown in Figure~\ref{fig:illustrate_obs}, where the observed outcomes are in Figure~\ref{fig:illustrate_obs_outcome} with darker color indicating a higher outcome. Readers may have some intuitive guess on our interested question: ``which subjects could have positive treatment effect if treated?'', such as subjects on the top right corner. Yet we note that the underlying ground truth can be counter-intuitive while mostly correctly captured by our proposed algorithm (results in Section~\ref{sec:intro_I_cube}). Before showing the results, we formalize our question of interest in the next section.

\begin{figure}[h!]
\centering
\hspace{1cm}
    \begin{subfigure}[t]{0.3\textwidth}
        \centering
        \includegraphics[width=1\linewidth]{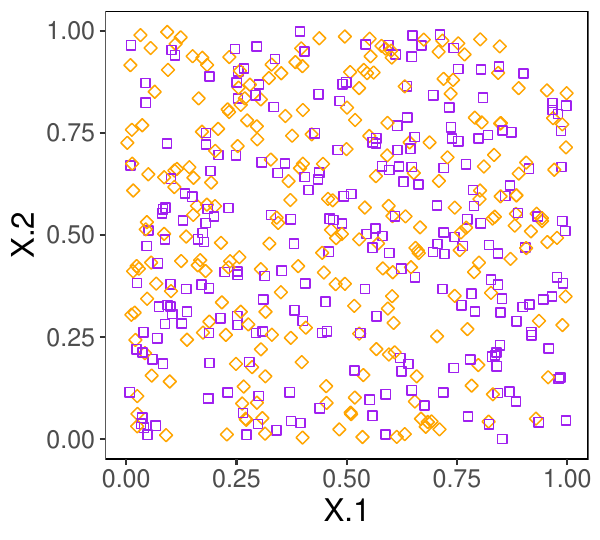}
        \caption{Subjects in treated group and control group separated by two shapes and colors.}
    \end{subfigure}
    \hfill
    \begin{subfigure}[t]{0.3\textwidth}
        \centering
        \includegraphics[width=1\linewidth]{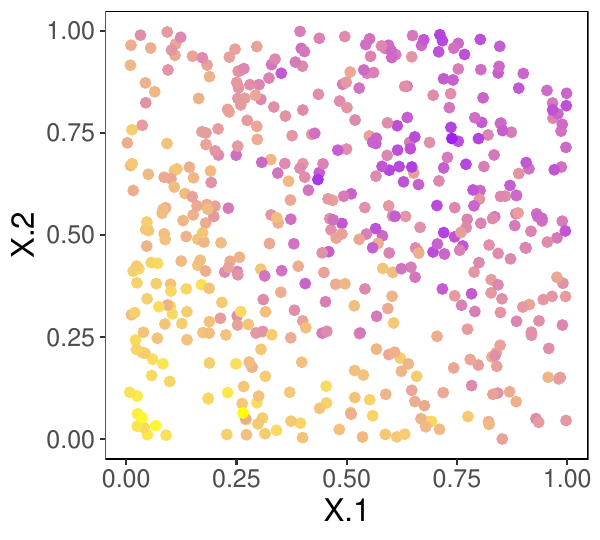}
        \caption{Higher outcomes are indicated by darker color.}
        \label{fig:illustrate_obs_outcome}
    \end{subfigure}
    \hspace{0.1cm}
    \begin{subfigure}[t]{0.08\textwidth}
        \centering
        \includegraphics[width=1\linewidth]{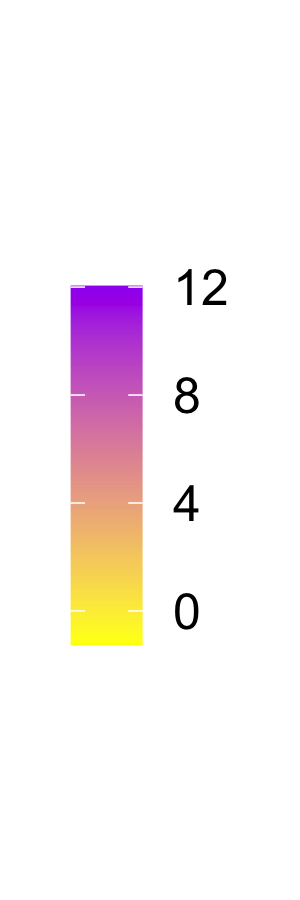}
    \end{subfigure}
\hspace{1cm}
    \caption{An illustrative example with 500 subjects, each has two recorded covariates.  
    Every point is a subject.
    The treatment assignments are in the left plot (squares are treated, diamonds are not), and the observed outcomes are in the right plot. We hope to answer the question: which individuals have a positive treatment effect?}
    \label{fig:illustrate_obs}
\end{figure}

\subsection{Problem setup}
Suppose we have $n$ subjects in the dataset. Each subject $i$ has potential control outcome $Y_i^C$, potential treated outcome $Y_i^T$, and the treatment indicator~$A_{i}$ for $i \in [n] \equiv \{1,2,\dots,n\}$. Our results allow the potential outcomes to either be viewed as random variables or fixed. The treatment effect of subject $i$ is defined as $Y_i^T - Y_i^C$ and the observed outcome is $Y_i = Y_i^C(1 - A_i) + Y_i^T A_i$ under the standard causal assumption of consistency ($Y_i = Y_i^T$ when $A_i = 1$ and $Y_i = Y_i^C$ when $A_i = 0$). Person $i$'s covariate is denoted as $X_i$. 
We first focus on Bernoulli randomized experiments without interference:
\begin{enumerate}
    \item[(i)] conditional on covariates, treatment assignments are independent coin flips:
    \begin{align} \label{eq:randomize1}
        \mathbb{P}[(A_1, \ldots, A_n) = (a_1, \ldots, a_n) \mid X_1, \ldots, X_n] = \prod_{i=1}^n \mathbb{P}(A_i = a_i) = (1/2)^n,
    \end{align}
    for any $(a_1, \ldots, a_n) \in \{0,1\}^n$. The above setting is later extended to observational studies where the probabilities of receiving treatments can be heterogeneous and possibly unknown.
    
    \item[(ii)] conditional on covariates, the outcome of one subject $Y_{i}$ is independent of the assignment~$A_{j}$ of another subject, for any $i \neq j$:
    \begin{align} \label{eq:randomize2}
         Y_{i} \perp A_{j} \mid \{X_1,\dots,X_n\} \text{ for } i \neq j,
    \end{align}
    which is implied by~\eqref{eq:randomize1} when the potential outcomes are viewed as fixed values.
\end{enumerate}
We do not assume the observed data $(Y_i, A_i, X_i)$ are identically distributed. We consider heterogeneous effects in the sense that the distribution of $Y_i^T - Y_i^C$ varies,
and aim at identifying those individuals with a positive treatment effect. (If the covariates~$X_i$ are not informative about the heterogeneity in $Y_i^T - Y_i^C$, our identification power could be low, and this is to be expected.) We choose to formalize and frame the problem in terms of multiple hypothesis testing, by first defining the null hypothesis for subject $i$ as having zero treatment effect: 
\begin{align} \label{eq:hp_twoside}
    H_{0i}^\text{zero}: (Y_i^T \mid X_i) \overset{d}{=} (Y_i^C \mid X_i), 
\end{align}
or equivalently, $H_{0i}^\text{zero}: (Y_i \mid A_i = 1, X_i) \overset{d}{=} (Y_i \mid A_i = 0, X_i)$. \footnote{Alternatively, we can treat the potential outcomes and covariates as fixed, and frame the null hypothesis as
    \(H_{0i}^\text{zero}: Y_i^T = Y_i^C.\)
A last, hybrid, version (e.g., \citet{howard2019uniform}) is to treat the two potential outcomes as random with joint distribution $(Y_i^T,Y_i^C) \mid X_i \sim P_i$, and the null posits 
    \(H_{0i}^\text{zero}: Y_i^T = Y_i^C \text{ almost surely-}P_i,\)
meaning that $P_i$ is supported on $\{(x,y): x=y\}$. All our theoretical results work with any interpretation, but we stick to~\eqref{eq:hp_twoside} by default.}

An extension is introduced in Appendix~\ref{sec:nonpositive-effect}, where we relax the null as those with a nonpositive effect, defined by stochastic dominance $(Y_i^T \mid X_i) \preceq (Y_i^C \mid X_i)$, meaning that $\mathbb{P}(Y_i^T \leq y \mid X_i) \leq \mathbb{P}(Y_i^C \leq y \mid X_i)$, or simply $Y_i^T \leq Y_i^C$ if the potential outcomes are fixed. 

Our algorithms control the error of falsely identifying subjects whose null hypothesis is true (i.e., having zero effect), and aim at correctly identifying subjects with positive effects. Let $\succ$ denote stochastic dominance, as above. We say a subject has a \textit{positive effect} if
\begin{align} 
    (Y_i^T \mid X_i) \succ (Y_i^C \mid X_i).
\end{align}
When treating the potential outcomes and covariates as fixed, we simply write $Y_i^T > Y_i^C$.

The output of our proposed algorithms is a set of identified subjects, denoted as $\mathcal{R}$, with a guarantee that the expected proportion of falsely identified subjects is upper bounded. Specifically, denote the set of subjects that are true nulls as $\mathcal{H}_0 := \{i\in [n]: H_{0i}^\text{zero} \text{ is true}\}$. 
Then the number of false identifications is~$|\mathcal{R} \cap \mathcal{H}_0|$. The expected proportion of false identifications is  a standard error metric, the false discovery rate (FDR):
\begin{align} \label{eq:fdr_def}
    \mathrm{FDR} := \mathbb{E}\left[\frac{|\mathcal{R} \cap \mathcal{H}_0|}{\max\{|\mathcal{R}|,1\}}\right].
\end{align}
Given $\alpha \in (0,1)$, we propose algorithms that guarantee FDR $\leq \alpha$, and have reasonably high \emph{power}, which is defined as the expected proportion of correctly identified subjects:
\[
\text{power} := \mathbb{E}\left[\frac{|\mathcal{R} \cap \pos|}{\max\{|\pos|,1\}}\right],
\]
where $\pos := \left\{i: (Y_i^T \mid X_i) \succ (Y_i^C \mid X_i)\right\}$ or $\pos := \{ i: Y_i^T > Y_i^C\}$ are subjects with positive effects.

\subsection{Related problem: error control in subgroup identification} \label{sec:review}

We note that our problem setup is not exactly the same as most work in subgroup identification, such as \citet{foster2011subgroup,zhao2012estimating,imai2013estimating}. The identified subgroups are usually defined by functions of covariates, rather than a subset of the investigated subjects as in our paper. While defining the subgroup by a function of covariates makes it easy to generalize the finding in the investigated sample to a larger population, it does not seem straightforward to nonasymptotically control the error of false identifications using the former definition, which is a major distinction between previous studies and our work. Most existing work does not have an error control guarantee (see an overview in \citet{lipkovich2017tutorial}, Table XV), except a few discussing error control on the level of subgroups as opposed to the level of individuals in our paper. The difference between FDR control at a subgroup level and at an individual level is detailed below.

\subsubsection{Subgroup FDR control}
\citet{karmakar2018false,gu2018oracle,xie2018false} discuss FDR control at a subgroup level, where the latter two have little discussion on incorporating continuous covariates and require parametric assumptions on the outcomes. Thus, we follow the setup in \citet{karmakar2018false} to compare the FDR control at a subgroup level (in their paper) and individual level (in our paper). Let the subgroups be non-overlapping sets $\{\mathcal{G}_1, \ldots, \mathcal{G}_G\}$. The null hypothesis for a subgroup $\mathcal{G}_g$ is defined as:
\begin{align*} 
    \mathcal{H}_{0g}: H_{0i}^\text{zero} \text{ is true for all } i \in \mathcal{G}_g,
\end{align*}
or equivalently, $\mathcal{H}_{0g}: \mathcal{G}_g \subseteq \mathcal{H}_0$ (recall $\mathcal{H}_0$ is the set of subjects with zero effect). Let $D_g$ be the $0/1$-valued indicator function for whether $\mathcal{H}_{0g}$ is identified or not. The FDR at a subgroup level is defined as the expected proportion of falsely identified subgroups:
\begin{align} \label{eq:fdr_group}
    \mathrm{FDR}^\text{subgroup} := \mathbb{E}\left[ \frac{|\{g \in [G]: \mathcal{G}_g \subseteq \mathcal{H}_0, D_g = 1|}{\max\{|\{g \in [G]: D_g = 1|,1\}}\right],
\end{align}
which collapses to the FDR at an individual level as defined in~\eqref{eq:fdr_def} when each subgroup has exactly one subject. Although our interactive procedure is designed for FDR control at an individual level, we propose extensions to FDR control at a subgroup level in Section~\ref{apd:subgroup}.
As a brief summary, \citet{karmakar2018false} propose to control $\mathrm{FDR}^\text{subgroup}$ by constructing a $p$-value for each subgroup and apply the classical BH method \citep{benjamini1995controlling}. However, it is not trivially applicable to control FDR at an individual level, because their $p$-values would only take value $1/2$ or $1$ when each subgroup has exactly one subject, leading to zero identification power following either the classic BH procedure or the more recent AdaPT framework in \citet{lei2018adapt}. In other cases where subgroups have more than one subject, the above error control does not imply whether subjects within a rejected subgroup are mostly non-nulls, or if many are nulls with zero effect. Such error control can be too weak to effectively tell apart most subgroups. Our paper appears to be the first to propose methods for identifying subjects having positive effects with (finite sample) FDR control. Individual level inferences are more fine-grained, from which (a) researchers can proceed with a follow-up analysis on identified individuals whom they strongly believe benefit from the treatment; (b) the identified individuals have  more trust in the treatment, and as an example, in industry one can recommend a new product to identified customers with much higher confidence). Importantly,  the identified individuals need not have been originally treated, as seen in our representative example (Figures~\ref{fig:illustrate_obs}, \ref{fig:illustrate_output}).

 We end by noting that we do propose an extension to our method for subgroup identification as well, and compare it to~\cite{karmakar2018false} in Section~\ref{apd:subgroup}.

\subsubsection{Other related error control at a subgroup level}
\citet{cai2011analysis} and \citet{athey2016recursive} develop confidence intervals for the averaged treatment effect within subgroups, where the former assumes the size of each subgroup to be large, and the latter requires a separate sample for inference. These intervals can potentially be used to generate a $p$-value for each subgroup and control FDR at a subgroup level via standard multiple testing procedures, but no explicit discussion is provided. \citet{lipkovich2011subgroup}, \citet{lipkovich2014strategies}, \citet{sivaganesan2011bayesian} and \citet{berger2014bayesian} propose methods with control on a different error metric: the global type-I error, which is the probability of identifying any subgroup when no subject has nonzero treatment effect (i.e., $H_{0i}^\text{zero}$ is true for all subjects). Our FDR control guarantee implies valid global type-I error, and FDR control is more informative on the correctness of the identified subgroups/subjects when there exist subjects having nonzero effects.

\subsection{An overview of our procedure} \label{sec:intro_I_cube}
As discussed, it appears to be new and practically interesting to provide FDR control guarantees at an individual level. Another merit of our proposed method is that it allows a human analyst and an algorithm to interact, in order to better accomplish the goal. 

Interactive testing is a recent idea that emerged in response to the growing practical needs of allowing human interaction in the process of data analysis. In practice, analysts tend to try several methods or models on the same dataset until the results are satisfying, but this violates the validity of standard testing methods (e.g., invalid FDR control). In our context of identifying positive effects, the appealing advantages of an interactive test include that (a)~an analyst is allowed to use (partial) data, together with prior knowledge, to design a strategy of selecting subjects potentially having positive effects, and (b)~it is a multi-step iterative procedure during which the analyst can monitor performance of the current strategy and make adjustments on the selection strategy at any step (at the cost of not altering earlier steps). Despite the flexibility of an analyst to design and alter the algorithm using (partial) data, our proposed procedure always maintains valid FDR control. We name our proposed algorithm \pei (I-cube), for interactive identification of individual treatment effects.

\vspace{20pt}
\begin{figure}[ht]
    \centering
\tikzstyle{op} = [rectangle, dashed, minimum width=1.5cm, minimum height=0.6cm, text centered, draw=black, fill=none]
\tikzstyle{rec} = [rectangle, minimum width=1.5cm, minimum height=0.6cm, text centered, draw=black, fill=none]
\tikzstyle{nrec} = [rectangle, minimum width=1.5cm, minimum height=0.6cm, text centered, draw=none, fill=none]
\tikzstyle{rrec} = [rectangle, rounded corners, minimum width=1.5cm, minimum height=0.6cm, text centered, draw=purple]
\tikzstyle{arrow} = [thick,->,>=stealth]

\begin{tikzpicture}[node distance=2cm]
\label{flow:exclude}

\node (hp) [rec]{$\{A_i\}$};
\node (gp) [rec, above of=hp, xshift = 0cm, yshift = 0.5cm]{$\{Y_i, X_i\}$};
\node (unmask) [nrec, left of=hp, text width=1.2cm]{Unmask $A_{i_{t-1}^*}$};
\node (side) [rec, above of= gp, xshift = 1.5cm, yshift = -1cm]{Prior information};
\node (rejset) [rec, right of=gp, xshift = 3.5cm]{Rejection set $\mathcal{R}_t$};
\node (error) [rec, right of=hp, xshift = 3.5cm]{Estimate $\widehat{\mathrm{FDR}}(\mathcal{R}_t)$};
\node (reject) [right of=error, xshift = 4.3cm]{Report $\mathcal{R}_t$};
\node (start) [right of=rejset, xshift = 4.1cm, text width=2.5cm]{Start $t = 0$, $\mathcal{R}_0 = [n]$, $i_0^* = \emptyset$};

\node (select) [below of=rejset, xshift = -3cm, yshift = 1.5cm,  color = cyan]{Selection};
\node (control) [below of=error, xshift = -3cm, yshift = 2.8cm, color = red]{Error control};

\draw [arrow] (start) -- (rejset);
\draw [arrow, dashed] (hp) -- (error);
\draw [arrow, dashed] (gp) -- (hp);
\draw [arrow] (gp) -- node[anchor= south, pos = 0.7] {Exclude $i_t^*$} (rejset);
\draw [arrow] (rejset) -- (error);
\draw [arrow] (error) -- node[anchor= south, pos = 0.5] {If $\widehat{\mathrm{FDR}}(\mathcal{R}_t) \leq \alpha$} (reject);
\draw [arrow] (side) -- (1.5,2.6);
\draw [arrow] (error) |- node[anchor= south, pos = 0.7] {If $\widehat{\mathrm{FDR}}(\mathcal{R}_t) > \alpha$, then $t \leftarrow t + 1$} (-2, -1.8) -- (unmask);
\draw [arrow] (unmask) |- (0, 3) -| (1.3, 2.6);
\draw [arrow] (hp) -- (-1.5, 0);

\draw[rounded corners=15pt, color=cyan, thick]
  (-1,1.7) rectangle ++(8.5,2.2);
\draw[rounded corners=15pt, color=red, thick]
  (-1,1) rectangle ++(8.5,-1.8);
\end{tikzpicture}
    \caption{A schematic of the \pei algorithm. All treatment assignments are initially kept hidden: only $(Y_i, X_i)_{i \in [n]}$ are revealed to the analyst, while all $\{A_i\}$ remain `masked'. The initial candidate rejection set is $\mathcal{R}_0=[n]$ (thus no subject is excluded initially and $i_0^* = \emptyset$). The false discovery proportion $\widehat{\mathrm{FDR}}$ of the current candidate set $\mathcal{R}_t$ is estimated by the algorithm (dashed lines), and reported to the analyst.
    If $\widehat{\mathrm{FDR}}(\mathcal{R}_t) > \alpha$, the analyst chooses a subject $i_t^*$ to remove it from the proposed rejection set $\mathcal R_t = \mathcal R_{t-1}\backslash \{i_t^*\}$, whose assignment~$A_{t^*_t}$ is then `unmasked' (revealed). Importantly, using any available prior information, covariates and working model, the analyst can choose subject $i_t^*$ and shrink $\mathcal R_t$ in any manner. This process continues until $\widehat{\mathrm{FDR}}(\mathcal{R}_t) \leq \alpha$ (or $\mathcal R_t = \emptyset$).}
    \label{fig:flow_interact}
\end{figure} 

\vspace{20pt}
The core idea that enables human interaction is to separate the information used for selecting subjects with positive effects and that for error control, via ``masking and unmasking'' (Figure~\ref{fig:flow_interact}). In short, masking means we hide the treatment assignment $\{A_i\}_{i=1}^n$ from the analyst. The algorithm alternates between two steps --- selection and error control --- until a simple stopping criterion introduced later is reached. Below is a intuitive description of the framework and we give the full method in Section~\ref{sec:zero-effect}. 
\begin{enumerate}
   \item \textbf{Selection}. Consider a set of candidate subjects to be identified as having a positive effect (whose null to be rejected), denoted as rejection set $\mathcal{R}_t$ for iteration~$t$. We start with all the subjects included, $\mathcal{R}_0 = [n]$. At each iteration, the analyst excludes possible nulls (i.e., subjects that are unlikely to have positive effects) from the previous $\mathcal{R}_{t-1}$, using all the available information (outcomes $Y_i$ and covariates $X_i$ for all subjects $i \in [n]$, and progressively unmasked $A_i$ from the step of error control, and possible prior information). Note that our method does not automatically use prior information and the revealed data. The analyst is free to use any black-box prediction algorithm that uses the available information, and evaluates the subjects possibly using an estimated probability of having a positive treatment effect. This step is where a human is allowed to incorporate her subjective choices.
   
   \item \textbf{Error control (and unmasking)}. The algorithm uses the complete data $\{Y_i, A_i, X_i\}$ to estimate FDR of the current candidate rejection set $\widehat{\mathrm{FDR}}(\mathcal{R}_t)$, as feedback to the analyst. 
   If the estimated FDR is above the target level $\widehat{\mathrm{FDR}}(\mathcal{R}_t) > \alpha$, the analyst goes back to the step of selection, along with additional information: the excluded subjects ($i \notin \mathcal{R}_t$) have their $A_i$ unmasked (revealed), which could improve her understanding of the data and guide her choices in the next selection step.
\end{enumerate}
The algorithms we propose in the main paper build on and modify the above procedure to achieve reasonably high power and develop various extensions.

Recall our illustrative example in Section~\ref{sec:intro_example}, the underlying groundtruth and the identifications made by the \peiZero (our central algorithm) are in Figure~\ref{fig:illustrate_result}. Although the observed outcomes tend to be higher for subjects in the top right corner, the subjects with true positive effects are in the center. Such discrepancy is because the \textit{control} outcome varies with covariates (could often happen in practice) and is designed to be higher in the top right corner, such that their observed outcomes could be high regardless of whether the treatment has any effect. Nonetheless, our proposed algorithm can tell the difference of high observed outcomes caused by high control outcomes versus those caused by positive treatment effects, and correctly identify most true positive effects.

\begin{figure}[h!]
\centering
\hspace{2cm}
    \begin{subfigure}[t]{0.3\textwidth}
        \centering
        \includegraphics[width=1\linewidth]{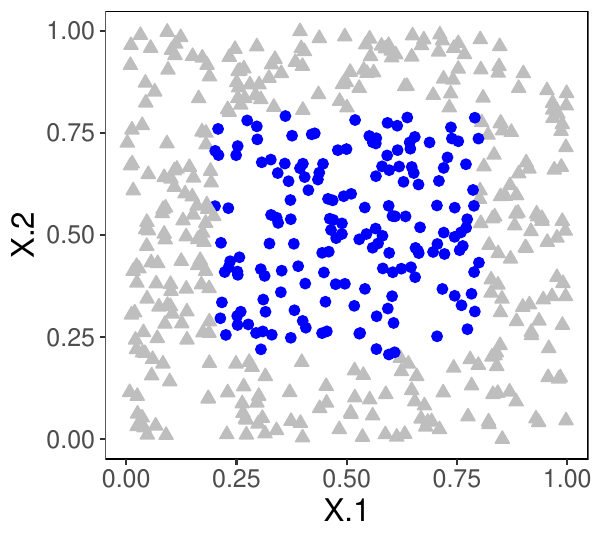}
        \caption{Blue round dots represent subjects with true positive effect (unknown ground truth).}
        \label{fig:illustrate_true}
    \end{subfigure}
    \hfill
    \begin{subfigure}[t]{0.3\textwidth}
        \centering
        \includegraphics[width=1\linewidth]{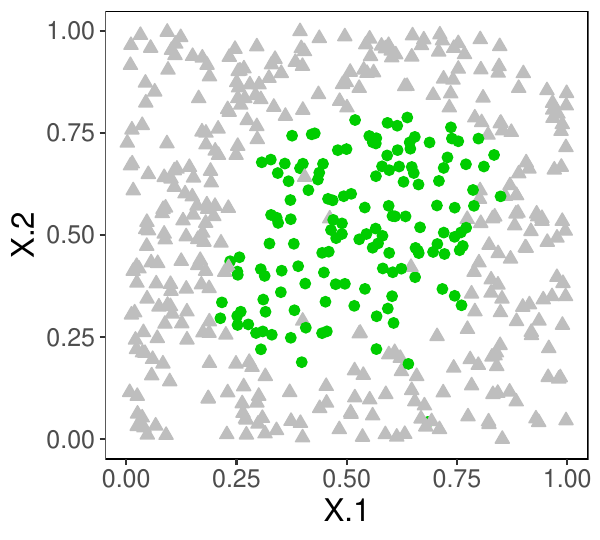}
        \caption{Green round dots represent subjects identified by the \peiZero.}
        \label{fig:illustrate_output}
    \end{subfigure}
\hspace{2cm}  
    \caption{An illustrative example with 500 subjects, each has two recorded covariates. The \peiZero identifies most subjects with positive effects, although about half of them did not receive treatment and their potential treated outcomes were not observed.}
    \label{fig:illustrate_result}
\end{figure}

\subsubsection{Related work in testing}
Testing procedures that allow human interaction are first proposed by \citet{lei2018adapt} and \citet{lei2017star} for the problem of FDR control in multiple testing, followed by several works for other error metrics, such as \citet{duan2019interactive} and \citet{duan2020familywise}. These papers focus on generic multiple testing problems, which operate on the $p$-values and ignore the process of generating $p$-values from data. In contrast, this paper applies the idea of interactive testing to observed data and propose tests in the context of causal inference for treatment effect. To our knowledge, we are not aware of $p$-values for individual identification in causal inference that are not binary and impose no model assumption. An example of binary $p$-value is $\mathbb{I}(\widehat \Delta_i \leq 0)$, where $\widehat \Delta_i$ can be viewed as a treatment effect estimator and defined in~\eqref{eq:est_effect}, and the $p$-value satisfies the mirror-conservative property in \citet{lei2018adapt}. 
 However, the masking framework then results in trivial guesses because each masked $p$-value is $\{0, 1\}$ for all individuals.
 We view our contribution as (1)~framing formally the question of selecting individuals with positive effect as a multiple testing problem with FDR control; and (2)~overcoming the challenge of designing $p$-value of individual zero effect hypothesis without any model assumption, by connecting the key property that enables masking framework in \citet{lei2018adapt}, independence between two functions of a single numeric $p$-value, and the key observation from our paper --- independence between treatment assignment $A_i$ and the vector of outcome and covariates~$\{Y_i, X_i\}$. The interactive tests first stem from the knockoff method for regression by \citet{barber2015controlling}, which relies on the symmetry of the carefully-designed test statistics. Here, we construct symmetric statistics in the framework of causal inference.


\subsubsection{Paper Outline}
The rest of the paper is organized as follows. In Section~\ref{sec:zero-effect}, we describe an interactive algorithm wrapped by a cross-fitting framework, which identifies subjects with positive effects with FDR control. We evaluate our proposed algorithm numerically in Section~\ref{sec:sim}, and provide theoretical power analysis in simple settings in Section~\ref{sec:power_analysis}. Section~\ref{sec:propensity_unknown} presents an extension to the setting of observational studies, and we implement the proposed algorithm to real datasets in Section~\ref{sec:imp}. More extensions can be found in Appendix~\ref{apd:extension}.  Section~\ref{sec:discuss} concludes the paper with a discussion on the potential of our proposed interactive procedures. 

\section{An interactive algorithm with FDR control} \label{sec:zero-effect}

To enable us to effectively infer the treatment effect, we use the following \textit{working model}:
\begin{align} \label{eq:work_model}
    Y_i^C = f(X_i) + U_i \text{ and }
    Y_i^T = \Delta(X_i) + f(X_i) + U_i,
\end{align}
where $U_i$ is zero-mean noise (unexplained variance) that is independent of $A_i$. When working with such a model, we effectively want to identify subjects with a positive treatment effect $\Delta(X_i)$. Importantly, model~\eqref{eq:work_model} needs not be correctly specified or accurately reflect reality in order for the algorithms in this paper to have a valid FDR control. However, the more ill-specified or inaccurate the model is, the more power may be hurt; related numerical experiments are included in Appendix~\ref{apd:working_model}.

To identify subjects with positive effects, we first introduce an estimator of the treatment effect~$\Delta(X_i)$ following the working model~\eqref{eq:work_model}. Denote the expected outcome given the covariates as $m(X_i) := \me(Y_i \mid X_i)$, and let $\widehat m(X_i)$ be an arbitrary estimator of $m(X_i)$ using the outcomes and covariates $\{Y_i, X_i\}_{i=1}^n$. Define the \textit{residual} as $E_i := Y_i  - \widehat m(X_i)$, and an estimator of $\Delta(X_i)$ is
\begin{align} \label{eq:est_effect}
    \widehat \Delta_i := 4 (A_i - 1/2) \cdot E_i,
\end{align}
which, under randomized experiments, is equivalent to the nonparametric estimator of the conditional treatment average effect $\mathbb{E}(Y_i^T \mid X_i) - \mathbb{E}(Y_i^C \mid X_i)$ in several recent papers \citep{nie2020quasi,kennedy2020optimal}, and can be traced back to the semiparametrics literature with~\citet{robinson1988root}.
A critical property of $\widehat \Delta_i$ that 
later leads to
FDR control is that\footnote{Note that property~\eqref{eq:sign_property_zero} uses the fact that the outcome estimator $\widehat m(X_i)$ is independent of $A_i$, so it is important that the estimation of $\widehat m$ does not use the assignments $\{A_i\}_{i=1}^n$; however, it should not affect the estimation much because $m(X_i) \equiv \me(Y_i \mid X_i)$ is not a function of $A_i$.}
\begin{align} \label{eq:sign_property_zero}
    \mathbb{P}(\widehat \Delta_i > 0 \mid \{Y_j, X_j, E_j\}_{j=1}^n) \leq 1/2,
\end{align}
under $H_{0i}^\text{zero}$ in~\eqref{eq:hp_twoside}, because $H_{0i}^\text{zero}$ implies $A_i \independent \{Y_i, X_i\}$ and $\mathbb{P}(A_i - 1/2 > 0) = 1/2$. With the rigorous argument for validity in Appendix~\ref{apd:proof-naive}, we briefly describe the reasoning here: the condition in~\eqref{eq:sign_property_zero} corresponds to the information used for selecting subjects (recall in Figure~\ref{fig:flow_interact}), because the treatment assignments $A_i$ are hidden for candidates in $\mathcal{R}_t$ and assignments are mutually independent.
The above property indicates that the estimated effect $\widehat \Delta_i$ is no more likely to be positive than negative if the selected subject has zero effect, regardless of how the analyst decides which subject to select. 
Therefore, the sign of $\widehat \Delta_i$ can be used to estimate the number of false identifications and achieve FDR control, which we elaborate next.

\subsection{The \pei algorithm} 
This section presents the \pei algorithm and proves that it controls FDR. We introduce a modification based on cross-fitting that improves identification power in the next section.

The  \pei proceeds as progressively shrinking a candidate rejection set $\mathcal{R}_t$ at iteration $t$, 
\[
[n] = \mathcal{R}_0 \supseteq \mathcal{R}_1 \supseteq \ldots \supseteq \mathcal{R}_n = \emptyset,
\]
where recall $[n]$ denotes the set of all subjects. We assume without loss of generality that one subject is excluded in each step. Denote the subject excluded at iteration $t$ as $i_t^*$. The choice of $i_t^*$ can use the information available to the analyst before iteration $t$, where we follow \citet{lei2018adapt} to describe the information available to the analysis formally as a filtration (sequence of nested $\sigma$-fields):
\begin{align} \label{eq:sigma}
    \mathcal{F}_{t-1}= \sigma\left(\{Y_j, X_j\}_{j \in \mathcal{R}_{t-1}}, \{Y_j, A_j,  X_j\}_{j \notin \mathcal{R}_{t-1}}, \sum_{j \in \mathcal{R}_{t-1}} \one\{\widehat \Delta_j > 0\}\right),
\end{align}
or equivalently,
\begin{align} 
    \mathcal{F}_{t-1} = \sigma\left(\{Y_j, X_j\}_{j \in [n]}, \{ A_j\}_{j \notin \mathcal{R}_{t-1}}, \sum_{j \in \mathcal{R}_{t-1}} \one\{\widehat \Delta_j > 0\}\right),
\end{align}
where we unmask (reveal) the treatment assignments $A_j$ for subjects excluded from $\mathcal{R}_{t-1}$, and the sum $\sum_{i \in \mathcal{R}_{t-1}} \one\{\widehat \Delta_i > 0\}$ is mainly used for FDR estimation as we describe later.  The above available information include arbitrary functions of the revealed data,
such as the residuals~$\{E_j\}_{j=1}^n$ defined above equation~\eqref{eq:est_effect}. Similar to property~\eqref{eq:sign_property_zero}, for each candidate subject $i \in \mathcal{R}_{t-1}$, we have
\begin{align} \label{eq:sign_property_full}
     \mathbb{P}(\widehat \Delta_i > 0 \mid \{Y_j, X_j\}_{j \in \mathcal{R}_{t-1}}, \{Y_j, A_j,  X_j\}_{j \notin \mathcal{R}_{t-1}}) \leq 1/2,
\end{align}
which ensures the FDR control as we explain next.

To control FDR, the number of false identifications is estimated by \eqref{eq:sign_property_full}. The idea is to partition the candidate rejection set $\mathcal{R}_t$ into $\mathcal{R}_t^+$ and $\mathcal{R}_t^-$ by the sign of $\widehat \Delta_i$: 
\begin{align*}
    \mathcal{R}_t^- := \{i \in \mathcal{R}_t: \widehat \Delta_i \leq 0\}, \quad
    \mathcal{R}_t^+ := \{i \in \mathcal{R}_t: \widehat \Delta_i > 0\}.
\end{align*}
Notice that our proposed procedure only identifies the subjects whose estimated effect is positive, i.e., those in $\mathcal{R}_t^+$. Thus, the FDR is $\mathbb{E}\left[\frac{|\mathcal{R}_t^+ \cap \mathcal{H}_0|}{\max\{|\mathcal{R}_t^+|,1\}}\right]$ by definition, where recall $\mathcal{H}_0$ is the set of true nulls. Intuitively, the number of false identifications ${|\mathcal{R}_t^+ \cap \mathcal{H}_0|}$ can be approximately upper bounded by ${|\mathcal{R}_t^- \cap \mathcal{H}_0|}$, since the number of positive signs should be no larger than the number of negative signs for the falsely identified nulls, according to property~\eqref{eq:sign_property_zero}. Note that the set of true nulls $\mathcal{H}_0$ is unknown, so we use $|\mathcal{R}_t^-|$ to upper bound $|\mathcal{R}_t^- \cap \mathcal{H}_0|$, and propose an estimator of FDR for the candidate rejection set $\mathcal{R}_t$:
\begin{align} \label{eq:fdr_hat}
    \widehat{\mathrm{FDR}}(\mathcal{R}_t) = \frac{|\mathcal{R}_t^-| + 1}{\max\{|\mathcal{R}_t^+|,1\}}.
\end{align}
Such FDR estimators using the count of positive or negative individuals stem from \citet{barber2015controlling}, and have been extensively used in the literature of interactive testing such as \citet{lei2018adapt,lei2017star,duan2019interactive,duan2020familywise}.
Overall, the \pei shrinks $\mathcal{R}_t$ until time $\tau:= \inf\{t:\widehat{\mathrm{FDR}}(\mathcal{R}_t) \leq \alpha\}$ and identifies only the subjects in $\mathcal{R}_\tau^+$, as summarized in Algorithm~\ref{alg:pei}. We state the FDR control of \pei in Theorem~\ref{thm:naive-ITE} and the proof can be found in Appendix~\ref{apd:proof-naive}.
\vspace{20pt}
\begin{theorem} \label{thm:naive-ITE}
In a randomized experiment with assumptions~\eqref{eq:randomize1} and~~\eqref{eq:randomize2}, and for any analyst that updates their working model(s) at any iteration $t$ using the information in $\mathcal{F}_{t-1}$, the set $R_\tau^+$ rejected by the \pei algorithm has FDR controlled at level~$\alpha$, meaning that
\[
\mathbb{E}\left[\frac{|\mathcal{R}_\tau^+ \cap \mathcal{H}_0|}{\max\{|\mathcal{R}_\tau^+|,1\}}\right] \leq \alpha,
\]
for the null hypothesis~\eqref{eq:hp_twoside}.
\end{theorem}
Consider a simple case where model~\eqref{eq:work_model} is accurate for every subject with a constant treatment effect $\Delta(X_i) = \delta > 0$. If $\delta$ is larger than the maximum noise, we have $\mathcal{R}_0^+ = [n]$, and the algorithm can stop at the very first step identifying all subjects. At the other extreme, if the effect~$\delta$ is too small, the algorithm may also return an empty set, and this makes sense because while small \textit{average} treatment effects can be learned using a large population, larger treatment effects are needed for \textit{individual-level} identification.

\begin{algorithm}[h]
   \caption{The \pei (interactive identification of individual treatment effect) procedure.}
   \label{alg:pei}
\begin{algorithmic}
   \STATE {\bfseries Initial state:} Statistician (S) knows covariates and outcomes $\{X_i,Y_i\}_{i = 1}^n$.
   \STATE Computer (C) knows the treatment assignments
   $\{A_i\}_{i=1}^n$.
   \STATE Target FDR level $\alpha$ is public knowledge.
   \STATE {\bfseries Initial exchange:} Both players initialize $\mathcal{R}_0 = [n]$ and set $t=1$.
   \STATE 1.~S builds a prediction model $\widehat m$ from $X_i$ to $Y_i$.
   \STATE 2.~S informs C about residuals $E_i \equiv Y_i - \widehat m(X_i)$.
   \STATE 3. C estimates the treatment effect as $\widehat \Delta_i \equiv 4 (A_i - 1/2) E_i$. 
   \STATE 4. C then divides $\mathcal{R}_t$ into $\mathcal{R}_t^- := \{i \in \mathcal{R}_t: \widehat \Delta_i \leq 0\}$ and $\mathcal{R}_t^+ := \{i \in \mathcal{R}_t: \widehat \Delta_i > 0\}$.
   \STATE 5. C reveals only  $|\mathcal{R}_t^+|$ to S (who infers $|\mathcal{R}_t^-|$).
   \STATE {\bfseries Repeated interaction:} 6. S checks if {$\widehat{\mathrm{FDR}}(\mathcal{R}_t) \equiv \frac{|\mathcal{R}_t^-| + 1}{\max\{|\mathcal{R}_t^+|,1\}} \leq \alpha$}.
   \STATE 7.  If yes, S sets $\tau=t$, reports $\mathcal{R}_\tau^+$ and exits.
   \STATE 8.~Else, S picks any $i_t^* \in \mathcal{R}_{t-1}$ using everything S currently knows.
   \STATE (S tries to pick an $i^*_t$ that S thinks is null, i.e. S hopes that $\widehat \Delta_{i^*_t} \leq 0$.)
   \STATE 9.~C reveals $A_{i_t^*}$ to S, who also infers $\widehat \Delta_{i_t^*}$ and its sign.
   \item 10. S updates ${\mathcal{R}_{t+1} = \mathcal{R}_{t}\backslash \{i_t^*\}}$, and also $|\mathcal{R}^+_{t+1}|$ and $|\mathcal{R}^-_{t+1}|$.
   \STATE 11. Increment $t$ and go back to Step 6.
\end{algorithmic}
\end{algorithm}
We end the section with a remark.
In step~8 of Algorithm~\ref{alg:pei}, we hope to exclude subjects that are unlikely to have positive effects, based on the revealed data in~$\mathcal{F}_{t-1}$. In other words, we should guess the sign of treatment effect~$\widehat \Delta_i$, which depends on both the revealed data $\{Y_i, X_i\}$ and the hidden assignment $A_i$.  
However, notice that at the first iteration, we may learn/guess the opposite signs for all the subjects; when all assignments~$\{A_i\}_{i=1}^n$ are hidden at $t = 1$, the likelihood of $\{A_i\}_{i=1}^n$ being the true values (leading to all correct signs for~$\widehat \Delta_i$) is the same as the likelihood of all opposite values (leading to all opposite signs for~$\widehat \Delta_i$), no matter what working model we use. 
Consequently, the subjects with large positive effects could be guessed as having large negative effects, causing them to be excluded from the rejection set. 
To improve power, we propose to wrap around the \pei by a cross-fitting framework as described in the next section. 


\subsection{Improving stability and power with \peiZero}
Cross-fitting refers to the idea of splitting the samples into two halves. We perform the \pei on each half separately, so that for each half, the complete data (including the assignments) of the other half is revealed to the analyst to help infer the sign of treatment effect, addressing the issue of learning the opposite signs and improving the identification power. 

Specifically, split the subjects randomly into two sets of equal size, denoted as $\mathcal{I}$ and~$\mathcal{II}$ where $\mathcal{I} \cup \mathcal{II} = [n]$. The \pei (Algorithm~\ref{alg:pei}) is implemented on each set separately: at the start of \pei on set~$\mathcal{I}$, the analyst has access to the complete data for all subjects in set $\mathcal{II}$, and tries to identify subjects with positive effects in set $\mathcal{I}$ with FDR control at level $\alpha/2$; similar is the \pei on set $\mathcal{II}$. Mathematically, let the candidate rejection set of implementing the \pei on set $\mathcal{I}$ be $\mathcal{R}_t(\mathcal{I})$, where the initial set is $\mathcal{R}_0(\mathcal{I}) = \mathcal{I}$. The available information at iteration $t$ is defined as:
\begin{align} \label{eq:sigma_A}
    \mathcal{F}_{t-1}(\mathcal{I})= \sigma\left(\{Y_i, X_i\}_{i\in \mathcal{R}_{t-1}(\mathcal{I})},  \{Y_j, A_j, X_j\}_{j \notin \mathcal{R}_{t-1}(\mathcal{I})}, \sum_{i \in \mathcal{R}_{t-1}(\mathcal{I})} \one\{\widehat \Delta_i > 0\}\right),
\end{align}
which includes the complete data $\{Y_j, A_j, X_j\}$ for $j \in \mathcal{II}$ at any iteration $t \geq 0$ \footnote{For notational clarity, we use $i$ to denote candidate subjects $i \in \mathcal{R}_t(\mathcal{I})$, and $j$ for non-candidate subjects $j \notin \mathcal{R}_t(\mathcal{I})$, while $k$ is used to index all subjects $k \in [n]$.}. Similarly, we define 
$\mathcal{R}_t(\mathcal{II})$ and 
$\mathcal{F}_{t-1}(\mathcal{II})$ for the \pei implemented on set $\mathcal{II}$. The final rejection set is the union of rejections in $\mathcal{I}$ and~$\mathcal{II}$ (see Algorithm~\ref{alg:para}). We call this algorithm the \peiZero. 

\begin{algorithm}[h]
   \caption{The \peiZero.}
   \label{alg:para}
\begin{algorithmic}
   \STATE {\bfseries Input:} Covariates, outcomes, treatment assignments $\{Y_i, A_i, X_i\}_{i = 1}^n$, target level~$\alpha$;
   \STATE {\bfseries Procedure:} 
   \STATE 1.~Randomly split the sample into two subsets of equal size, denoted as $\mathcal{I}$ and $\mathcal{II}$;
   \STATE 2.~Implement Algorithm~\ref{alg:pei} at level $\alpha/2$, where E initially knows $\{Y_k, X_k\}_{k=1}^n \cup \{A_j\}_{j\in \mathcal{II}}$ and sets $\mathcal{R}_0(\mathcal{I}) = \mathcal{I}$, 
   getting a rejection set $\mathcal{R}_\tau^+(\mathcal{I}) \subseteq \mathcal{I}$;
   \STATE 3.~Implement Algorithm~\ref{alg:pei} at level $\alpha/2$, where E initially knows $\{Y_k, X_k\}_{k=1}^n \cup \{A_j\}_{j\in \mathcal{I}}$ and sets $\mathcal{R}_0(\mathcal{II}) = \mathcal{II}$, 
   getting a rejection set $\mathcal{R}_\tau^+(\mathcal{II}) \subseteq \mathcal{II}$;
   \STATE 4.~Combine two rejection sets as the final rejection set, $\mathcal{R}_\tau^+ = \mathcal{R}_\tau^+(\mathcal{I}) \cup \mathcal{R}_\tau^+(\mathcal{II})$.
\end{algorithmic}
\end{algorithm}
As long as the \pei on two sets do not exchange information, Algorithm~\ref{alg:para} has a valid FDR control (see the proof in Appendix~\ref{apd:crossfit}). 
\begin{theorem} 
\label{thm:LNO-ITE}
Under assumption~\eqref{eq:randomize1} and~\eqref{eq:randomize2} of randomized experiments, $\mathcal{R}_\tau^+$ rejected by the \peiZero has FDR controlled at level $\alpha$ for the null hypothesis~\eqref{eq:hp_twoside}.
\end{theorem}

In addition to addressing the issue of learning the opposite $\widehat \Delta_i$ in the original \pei, another benefit of using the crossing-fitting framework is that with the complete data revealed for at least half of the sample, the analyst does not have to deal with the problem of inferring missing data (the assignment $A_i$), which probably needs some parametric probabilistic modeling and the EM algorithm. Instead, because the assignments are revealed for subjects not in the candidate rejection set (at least half of the sample), their signs of~$\widehat \Delta_j$ can be correctly calculated and used as ``training data''. The analyst can then employ a black-box prediction model, such as a random forest, to predict the signs of $\widehat \Delta_i$ for the subjects whose assignments are masked (hidden). As an example, we propose an automated strategy as follows to select a subject at step~8 in Algorithm~\ref{alg:pei}.

\begin{algorithm}[H]
   \caption{An automated heuristic to select $i_t^*$ in the \peiZero.}
   \label{alg:select_RF}
\begin{algorithmic}
   \STATE {\bfseries Input:} Current rejection set $\mathcal{R}_{t-1}(\mathcal{I})$, and available information for selection $\mathcal{F}_{t-1}(\mathcal{I})$;
   \STATE {\bfseries Procedure:} 
   \STATE 1.~Train a random forest classifier where the label is $\text{sign}(\widehat \Delta_j)$ and the predictors are $Y_j, X_j$ and the residuals $E_j$, using non-candidate subjects $j \notin \mathcal{R}_{t-1}(\mathcal{I})$; 
   \STATE 2.~Estimate the probability of $\widehat \Delta_i$ being positive as $\widehat p(i,t)$ for subjects $i \in \mathcal{R}_{t-1}(\mathcal{I})$;
   \STATE 3.~Select $i_t^* = \argmin\{\widehat p(i,t): i \in \mathcal{R}_{t-1}(\mathcal{I})\}$.
\end{algorithmic}
\end{algorithm}

We remark that in practice, the analyst can interactively change the prediction model, such as exploring parametric models to see which fits the data better. In principle, the analyst can perform any exploratory analysis on data in $\mathcal{F}_{t-1}(\mathcal{I})$ to decide a heuristic or score for selecting subject $i_t^*$; and the FDR control is valid as long as she does not use the assignments~$A_i$ for candidate subjects $i \in \mathcal{R}_{t-1}(\mathcal{I})$. For computation efficiency, we usually update the prediction models (or their parameters) once every 100 iterations~(say).

To summarize, the \peiZero described in Algorithm~\ref{alg:para} involves two rounds of the  \pei (Algorithm~\ref{alg:pei}), where step~8 of selecting a subject is allowed to involve human interaction; alternatively, step~8 can be an automated heuristic as presented in Algorithm~\ref{alg:select_RF}.

\begin{remark}
We contrast our algorithms to identify individual effects with many existing algorithms targeting at alternative goals. As examples, two commonly discussed goals include testing whether \textit{any} individual has any effect~\citep{rosenbaum2002covariance,rosenblum2009using,howard2019uniform} and estimating the averaged treatment effect~\citep{lin2013agnostic,fogarty2018regression,guo2020generalized}. While the above two questions discuss causal inference at an integrated level for the investigated population, we study the problem of making claims for each individual subject, a harder problem by its nature. Therefore, it is expected that a larger effect size is required, compared to the previous two goals, to have reasonably high power for individual effect identification. In the following sections, we demonstrate through repeated numerical experiments and theoretical analysis that the \peiZero has reasonably high power.
\end{remark}

\section{Numerical experiments} \label{sec:sim}
To assess our proposed procedure, we first describe a baseline method, which calculates a $p$-value for each subject under the assumption of linear models, and applies the classical BH method \citep{benjamini1995controlling}. We call this method the linear-BH procedure.

\subsection{Two baselines: the BH procedure (assuming well-specified model) and  Selective SeqStep+}

\textbf{BH procedure under linear assumptions.} For the treated group and control group, we first separately learn a linear model to predict $Y_i$ using $X_i$, denoted as $\widehat l^T$ and $\widehat l^C$. By imputing the unobserved potential outcomes, we get estimators of the potential outcomes $\widetilde Y_i^T = Y_i \one\{A_i = 1\} + \widehat l^T(X_i) \one\{A_i = 0\} $ and $\widetilde Y_i^C = \widehat l^C(X_i) \one\{A_i = 1\} + Y_i \one\{A_i = 0\}$, and the treatment effect for subject $i$ can be estimated as $\widehat{\Delta}_i^{\text{BH}} := \widetilde Y_i^T - \widetilde Y_i^C$. If the potential outcomes are linear functions of covariates with standard Gaussian noises (which we refer to as the linear assumption), the estimated treatment effect asymptotically follows a Gaussian distribution. For each subject $i \in [n]$, we calculate a $p$-value for the zero-effect null~\eqref{eq:hp_twoside} as
\begin{align} \label{eq:p-linear-BH}
    P_i = 1 - \Phi\left(\widehat{\Delta}_i^{\text{BH}} \Big/ \sqrt{\widehat{\text{Var}}(\widehat{\Delta}_i^{\text{BH}})}\right),
\end{align}
where the estimated variance is $\widehat{\text{Var}}(\widehat{\Delta}_i^{\text{BH}}) = \widehat{\text{Var}} (\widetilde Y_i^T) + \widehat{\text{Var}} (\widetilde Y_i^C)$, and $\Phi$ denotes the CDF of a standard Gaussian.
To identify subjects having positive effects with FDR control, we apply the BH procedure to the above $p$-values. Note that, unlike our methods, the error control for BH would not hold when the linearity assumption is violated (see Appendix~\ref{apd:FDR_BH_linear} for the formal FDR control guarantee; see Section~\ref{apd:sim_more} for more numerical experiments when the linear assumption holds or does not). 

\textbf{Selective SeqStep+}. Once we construct the estimated treatment effect $\widehat \Delta_i$ and call out the critical property~\eqref{eq:sign_property_zero} on the probability of the estimated effect sign, we can apply the Selective SeqStep+ by \citet{barber2015controlling}. Specifically, we set the $p$-value for each individual as $p_i = 1 - \frac{1}{2}\one(\widehat \Delta_i > 0)$, which equals $1/2$ when estimated treatment effect $\widehat \Delta_i > 0$ and equals $1$ when $\widehat \Delta_i <= 0$. The hypotheses can be ordered by $|E_i|$ decreasingly, and the constant $c$ is chosen as $1/2$ to maximize the power. We also note that the Selective SeqStep+ can be viewed as an automated version of \pei where the Statistician (S) picks individuals by the ranking of $|E_i|$ in step~8 of Algorithm~\ref{alg:pei}. We expect the power of our proposed \peiZero to be higher than the Selective SeqStep+ (an automated version of \pei), because \peiZero additionally uses information from the revealed treatment assignments when picking/ordering individuals, which are especially helpful to inform the direction of treatment effect.

\subsection{Numerical experiments and power comparison} \label{sec:sim_setup}
We run a simulation with 500 subjects $\smash{(n = 500)}$. Each subject is recorded with two binary attributes (eg. female/male and senior/junior) and one continue attribute (eg. body weight), denoted as a vector ${X_i = (X_i(1), X_i(2), X_i(3)) \in \{0,1\}^2\times \mathbb{R}}$. Among $n$ subjects, the binary attributes are marginally balanced, and the subpopulation with $X_i(1) = 1$ and $X_i(2) = 1$ is of size $30$. The continuous attribute is independent of the binary ones and follows the distribution of a standard Gaussian.

The outcomes are simulated as a function of the covariates $X_i$ and the assignment $A_{i}$ following the generating model~\eqref{eq:work_model}. Recall that we previously used model~\eqref{eq:work_model} as a working model, which is not required to be correctly specified. Here, we generate data from such a model in simulation for a clear evaluation of the considered methods. We specify the noise~$U_i$ as a standard Gaussian, and the expected control outcome as $ f(X_i) = 5(X_i(1) + X_i(2) + X_i(3)),$
and the treatment effect as
\begin{align} \label{eq:bias-sparse}
    \Delta(X_i) =  S_\Delta \cdot [5X_i^3(3) \one\{X_i(3) > 1\} - X_i(1)/2],
\end{align}
where $S_\Delta > 0$ encodes the scale of the treatment effect. In this setup, around 15\% subjects have positive treatment effects with a large scale, and 43\% subjects have a mild negative effect \footnote{R code to fully reproduce all plots in the paper are available at \url{https://github.com/duanby/I-cube}.}.
  
\begin{figure}[htb!]
\centering
\hspace{1cm}
    \begin{subfigure}[t]{0.3\textwidth}
        \centering
        \includegraphics[width=1\linewidth]{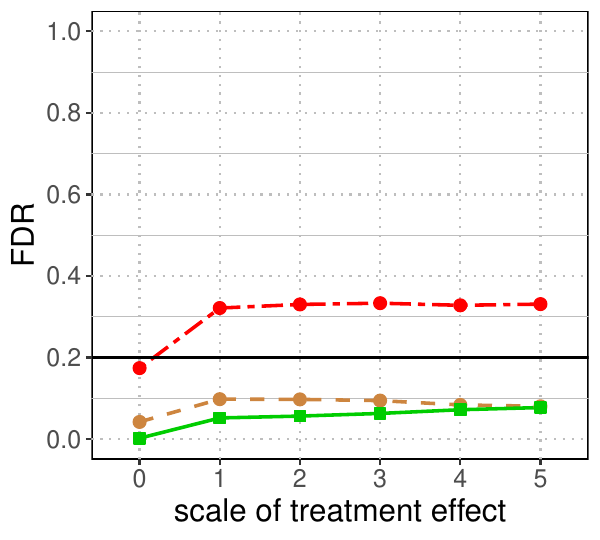}
    \end{subfigure}
    \hfill
    \begin{subfigure}[t]{0.3\textwidth}
    \centering
        \includegraphics[width=1\linewidth]{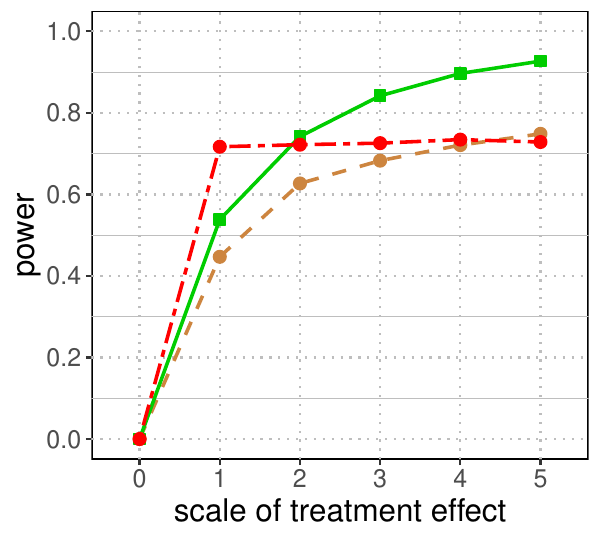}
    \end{subfigure}
\hspace{1cm}
\vfill
    \begin{subfigure}[t]{1\textwidth}
        \centering
        \includegraphics[width=0.7\linewidth]{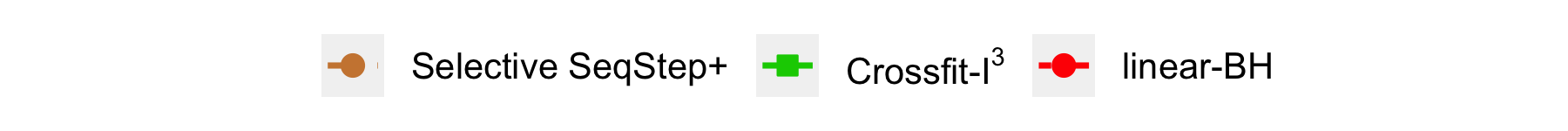}
    \end{subfigure}
    \caption{FDR (left) and power (right) of the \peiZero compared with the linear-BH procedure and the Selective SeqStep+, with the treatment effect specified as model~\eqref{eq:bias-sparse} and the scale $S_\Delta$ varying in $\{0,1,2,3,4,5\}$. The FDR control level is $0.2$, marked by a horizontal line in error control plots. For all plots in this paper, the FDR and power are averaged over 500 repetitions. The \peiZero controls FDR while the linear-BH procedure does not because the treatment effect is nonlinear. Also, the \peiZero can achieve higher power than both the Selective SeqStep+ and the linear-BH procedure.}
    \label{fig:bias-sparse}
\end{figure}

For the \peiZero, we use random forests (with default parameters in \texttt{R}) to compute $\widehat m$, and use the automated selection strategy Algorithm~\ref{alg:select_RF} to select a subject at step~8 in Algorithm~\ref{alg:pei}. The linear-BH procedure results in a substantially higher FDR than desired because the linear assumption does not hold in the underlying truth~\eqref{eq:bias-sparse} (see Figure~\ref{fig:bias-sparse}), whereas the Selective SeqStep+ and our proposed \peiZero control FDR at the target level as expected. At the same time, the \peiZero appears to have higher power than the Selective SeqStep+ and the linear-BH procedure to correctly identify subjects with true positive effects. More experiments that explore different treatment effect setups can be found in Section~\ref{apd:sim_more}.

\section{Asymptotic power analysis in simple settings} \label{sec:power_analysis}
In addition to the numerical experiments, we provide a theoretical power analysis in some simple cases to understand the advantages and limitations of our proposed \peiZero.

First, consider the case without covariates. Our analysis is inspired by the work of \citet{arias2017distribution,rabinovich2020optimal}, who study the power of methods with FDR control under a sparse Gaussian sequence model. Let there be $n$ hypotheses, each associated with a test statistic $V_i$ for $i = 1, \ldots, n$. They consider a class of methods called \textit{threshold procedures} such that the final rejection set $\mathcal{R}$ is in the form $\mathcal{R} = \{i: V_i \geq \tau(V_1, \ldots, V_n)\},$ for some threshold $\tau(V_1, \ldots, V_n)$; they discuss two types of thresholds; see Appendix~\ref{apd:power} for details of their results. An example of the threshold procedure is the BH procedure. Our proposed \pei can also be simplified to a threshold procedure when using an automated selection strategy at step~8 of Algorithm~\ref{alg:pei}: at each iteration, we exclude the subject with the smallest absolute value of the estimated treatment effect $|\widehat \Delta_i|$ (note that this strategy satisfies our requirement of not using assignments since $|\widehat \Delta_i| = |4(A_i - 1/2)(Y_i - \widehat m(X_i))| = 2|Y_i - \widehat m(X_i)|$). The resulting (simplified and automated)~\pei is a threshold procedure where $V_i = \widehat \Delta_i$. 
Note that our original interactive procedure is highly flexible, making the power analysis less obvious, so we discuss the power of \peiZero with the above simplified selection strategy.

To contextualize our power analysis, we paraphrase one of the results in \citet{arias2017distribution,rabinovich2020optimal}. Assume the test statistics $V_i\sim N(\mu_i, 1)$ are independent, with $\mu_i = 0$ under the null and $\mu_i = \mu > 0$ otherwise. Denote the number of non-nulls as $n_1$ and the \textit{sparsity}  of the non-nulls is parameterized by $\beta \in (0,1)$ such that $n_1/n = n^{-\beta}$. Let the signal $\mu$ increase with $n$ as $\mu = \sqrt{2r\log n}$, where the \textit{signal strength} is encoded by $r \in (0,1)$. Their power analysis is characterized by the signal $r$ and sparsity $\beta$, which are also critical parameters to characterize the power in our context as we state later.
These authors effectively prove that
\emph{for any fixed FDR level $\alpha \in (0,1)$, no threshold procedure can have nontrivial power if $r < \beta$, but there exist threshold procedures with asymptotic power one if $r > \beta$}. 

Our analysis differs from theirs in the non-null distribution of the test statistics. Given $n$ subjects, suppose the potential outcomes for subject~$i$ are distributed as: $Y_i^C \sim N(0,1)$ and $Y_i^T \sim N(\mu_i, 1)$, where the alternative mean is $\mu_i = 0$ if subject~$i$ is null, or $\mu_i = \mu > 0$ if~$i$ is non-null. Thus, the observed outcome of a null is $N(0, 1)$, and that of a non-null is a \textit{mixture} of $N(\mu, 1)$ and $N(0, 1)$ (depending on the treatment assignment), instead of a shift of the null distribution as assumed in \citet{arias2017distribution}, and the proof of the following result thus involves some modifications on their proofs (see Appendix~\ref{apd:power_noX}).
\begin{theorem} \label{thm:power_noX}
Given a fixed FDR level $\alpha \in (0,1)$ and let the number of subjects $n$ go to infinity. When there is no covariate, the automated \peiZero and the linear-BH procedure have the same power asymptotically: if $r < \beta$, their power goes to zero; if $r > \beta$, their power goes to~$1/2$. Further, among the treated subjects, their power goes to one.
\end{theorem}
\begin{remark}
Power of both methods cannot converge to a value larger than $1/2$ because without covariates, we cannot differentiate between the subjects with zero effect (whose outcome follows standard Gaussian regardless of treated or not) and the subjects with positive effects that are not treated (which also follows standard Gaussian). And the proportion of untreated subjects among those with positive effects is $1/2$ because of the assumed randomization.
\end{remark}

The above theorem discusses the case where there are no covariates to help guess which untreated subjects have positive effects. Next, we consider the case with an ``ideal'' covariate~$X_i$: its value corresponds to whether a subject is a non-null (having positive effect) or not, $X_i = \one\{\mu_i > 0\}$. Here, we design the selection strategy (for step~8 of Algorithm~\ref{alg:pei}) as a function of the covariates, because we hope that subjects with the similar covariates have similar treatment effects. Specifically, for the \pei implemented on~$\mathcal{I}$, we learn a prediction of~$\widehat \Delta_j$ by $X_j$ using non-candidate subjects $j \in \mathcal{II}$: $\text{Pred}(x) = \tfrac{1}{|\mathcal{II}|}\sum_{i \in \mathcal{II}}\widehat \Delta_j \one\{X_j = x\}$, where $x = \{0,1\}$. Then for candidate subjects $i \in \mathcal{I}$, we exclude the ones whose $\text{Pred}(X_i)$ are lower. As we integrate information among subjects with the same covariate value, all non-null subjects (i.e., those with $X_i = 1$) would excluded after the nulls (with probability tending to one), regardless of whether they are treated or not; hence we achieve power one. 
\begin{theorem} \label{thm:power_oracle_X}
Given a fixed FDR level $\alpha \in (0,1)$ and let the number of subjects $n$ go to infinity. With a covariate $X_i = \one\{\mu_i > 0\}$, 
the power of the automated \peiZero converges to one for any fixed $r \in (0,1)$ and $\beta \in (0,1)$. In contrast, the power of the linear-BH procedure goes to zero if $r < \beta$. (When $r > \beta$, power of both methods converges to one.)  
\end{theorem}
Here is a short informal argument for why our power goes to one (see detailed proof in Appendix~\ref{apd:power_oracle_X}). Since the nulls can be excluded before the non-nulls, we focus on the test statistics of the non-nulls. Let $\overset{d}{\to}$ denote convergence in distribution. The estimated effect $\widehat \Delta_i \overset{d}{\to} N(\mu, 1)$ for each non-null (since in the notation of Algorithm~\ref{alg:pei}, $\widehat m(X_i = 1)$ converges to $\mu/2$ for the non-nulls, and thus, $E_i \overset{d}{\to} N(\mu/2, 1)$ for those with $A_i = 1$, and $E_i \overset{d}{\to} N(-\mu/2, 1)$ for those with $A_i = 0$.) Hence, at the time $t_0$ right after all the nulls are excluded (and all the non-nulls are in~$\mathcal{R}_{t_0}$), the proportion of positive estimated effects $|\mathcal{R}_{t_0}^+|/|\mathcal{R}_{t_0}|$ converges to $\Phi(\mu)$, where $\Phi$ denotes the CDF of a standard Gaussian. We can stop before $t_0$ and identify subjects in $\mathcal{R}_{t_0}^+$ if $\widehat{\mathrm{FDR}}(\mathcal{R}_{t_0})$, as a function of $|\mathcal{R}_{t_0}^+|/|\mathcal{R}_{t_0}|$, is less than $\alpha$, which holds when $\Phi(\mu) > \tfrac{1}{1 + \alpha}$. The power goes to one because $\mu$ grows to infinity for any fixed $r \in (0,1)$, so that for large~$n$, we stop before~$t_0$ and the proportion of rejected non-nulls $|\mathcal{R}_{t_0}^+|/|\mathcal{R}_{t_0}|$ (which converges to~$\Phi(\mu)$ as argued above) also goes to one. In short, the power guarantee does not depend on the sparsity~$\beta$ because of the designed selection strategy that incorporates covariates.

We note that our theoretical power analysis discusses two extreme cases, one with no covariate to assist the testing procedure (Theorem~\ref{thm:power_noX}), and the other with a single ``ideal'' covariate that equals the indicator of non-nulls (Theorem~\ref{thm:power_oracle_X}). The numerical experiments in Section~\ref{sec:sim} consider more practical settings, where the analyst is provided with a mixture of covariates informative about the heterogeneous effect ($X_i(1)$ and $X_i(3)$ in our example) and some uninformative ones; still, the \peiZero tends to have reasonably high power. So far, the paper discusses the setup of a randomized experiment where each subject has 1/2 probability to be treated; in the following, we present the variant of \peiZero for observational studies, where the probabilities of receiving treatment can depend on covariates and unknown.


\section{\peiZero in observational studies}
\label{sec:propensity_unknown}
For clear notation, we denote the true propensity score (probability of receiving treatment) for subject~$i$ as $\pi_i$, and the estimated one as $\widehat \pi_i$ (which we introduce soon). For the setup in observational studies, we introduce an alternative set of assumptions to replace assumption in~\eqref{eq:randomize1}:
\begin{enumerate}
    \item[(iii)] conditional on covariates, treatment assignments are independent:
    \begin{align} \label{eq:randomized3}
        \mathbb{P}[(A_1, \ldots, A_n) = (a_1, \ldots, a_n) \mid X_1, \ldots, X_n] = \prod_{i=1}^n \mathbb{P}(A_i = a_i \mid X_1, \ldots, X_n) =  \prod_{i=1}^n \pi_i,
    \end{align}
    for any $(a_1, \ldots, a_n) \in \{0,1\}^n$ and $\{\pi_i\}_{i=1}^n$ can be unknown. 
    
    \item[(iv)] the propensity scores are bounded away from $0$ and $1$:
    \begin{align} \label{eq:prop_bound}
        0 < \pi_{\min} \leq \pi_i \leq \pi_{\max} < 1 \text{ for all } i \in [n].
    \end{align}
\end{enumerate}

Because the propensity scores $\pi_i$ are unknown, we estimate them using the revealed data --- specifically in the cross-fitting framework using the complete data for subjects in $\mathcal{II}$.  To be exact, we modify the \peiZero in algorithm~\ref{alg:para} in two components:
\begin{itemize}
    \item prior to step~2 of implementing the \pei, we estimate the bounds for the propensity scores as $\widehat{\pi_{\min}}(\mathcal{I})$ and $\widehat{\pi_{\min}}(\mathcal{I})$ by the complete data in $\mathcal{II}$. For example, we can estimate individual propensity scores by a logistic regression on covariates $X_j$ using the complete data from non-candidate subjects $j \in \mathcal{II}$, and take their minimum and maximum as the estimations; and
    \item we modify the implementation of Algorithm~\ref{alg:pei} in the FDR estimator:
    \begin{align} \label{eq:fdr_propoensity_estA}
    \widehat{\mathrm{FDR}}_t^{\widehat \pi}(\mathcal{R}_t(\mathcal{I})) := \left(\frac{1}{1 - \max\{1 - \widehat{\pi_{\min}}(\mathcal{I}), \widehat{\pi_{\max}}(\mathcal{I})\}} - 1\right)\frac{|\mathcal{R}_t^-(\mathcal{I})| + 1}{|\mathcal{R}_t^+(\mathcal{I})| \vee 1};
    \end{align}
\end{itemize}
and similarly modify the procedure on set $\mathcal{II}$. We call the resulting algorithm \peiZeroUnknowP, which have asymptotic FDR control when the propensity scores are well estimated (proof in Appendix \ref{apd:proof-crossfit-unknown}). Note that \peiZeroUnknowP does not involve a modification of the treatment effect estimator $\widehat \Delta_j$, which will not affect the FDR control as stated in the theorem below. As a future direction, it would be an interesting extension to modify $\widehat \Delta_j$ and utilize the information of estimated propensity scores, and potentially improve the identification power.



\begin{theorem} \label{thm:crossfit-unknown}
Suppose there are $n$ samples for identification. In the cross-fitting framework, let $\widehat{\pi_{\min}}(\mathcal{I})$ and $\widehat{\pi_{\max}}(\mathcal{I})$ be the estimated lower and upper bound of the propensity scores based on data information in $\mathcal{F}_{0}(\mathcal{I})$. Denote the estimation error as 
\begin{align*}
\epsilon_n^{\pi}(\mathcal{I}) := \mathbb{E}_{\mathcal{F}_0(\mathcal{I})}\left[\max\{\widehat{\pi_{\min}}(\mathcal{I}) - \pi_{\min}, \pi_{\max} - \widehat{\pi_{\max}}(\mathcal{I}), 0\}\right],
\end{align*}
and similarly define $\epsilon_n^{\pi}(\mathcal{II})$. Let $\epsilon_n^{\pi} = \epsilon_n^{\pi}(\mathcal{I})+\epsilon_n^{\pi}(\mathcal{II})$, and
the FDR of \peiZeroUnknowP is upper bounded:
\begin{align}
    \mathrm{FDR} \leq \alpha \left\{1 + \epsilon_n^{\pi} \left(\frac{4}{ \max\{1 - \pi_{\min}, \pi_{\max}\}\} (1 - \max\{1 - \pi_{\min}, \pi_{\max}\})}\right) \right\},
\end{align}
when $\max\{\epsilon_n^q(\mathcal{I}), \epsilon_n^q(\mathcal{II})\}  \leq  \frac{1}{2}\max\{1 - \pi_{\min}, \pi_{\max}\}$,
under assumptions~\eqref{eq:randomized3}, ~\eqref{eq:prop_bound}, and ~\eqref{eq:randomize2} in observational studies, for the null hypothesis~\eqref{eq:hp_twoside}. 
\end{theorem}

\begin{corollary}
\peiZeroUnknowP has asymptotic FDR control for the zero-effect null~\eqref{eq:hp_twoside} when the estimation of propensity score bounds is consistent in the sense that $\epsilon_n^{\pi} \to 0$ as sample size $n$ goes to infinity.
\end{corollary}

\begin{remark}
The \peiZeroUnknowP would have a larger FDR than the target level if the propensity score estimation is not consistent, and this inflation increases as the true bounds of propensity score get close to $0$ and~$1$. Nonetheless, note that the inflation only depends on the minimum and maximum of the propensity scores (rather than for each individual), and only concerns the \textit{one-sided} error for their estimation. Intuitively, we could have exact FDR control if the estimated minimum propensity score is lower than the true minimum and the estimated maximum larger than the true maximum --- if the estimation is conservative to capture the extreme cases. Such desirable FDR control comes with a risk of having lower power because the conservative propensity score estimations would increase the FDR estimation in \peiZeroUnknowP, and in turn, more subjects need to be excluded before claiming rejections.   
\end{remark}

\begin{remark} \label{rmk:MAY}
Another variation is proposed in Appendix~\ref{sec:nonpositive-effect}, as what we call \peiNegUnknowP, which can have a stronger error control guarantee at the cost of reserving (masking) more information for FDR control and could potentially have lower identification power. Specifically, the variant can have doubly robust asymptotic FDR control: either when the propensity scores are well-estimated as above, or when the conditional outcomes given covariates $m(X_i) \equiv \mathbb{E}(Y_i \mid X_i)$ are well-estimated by $\widehat{m}(X_i)$. In addition, another advantage of the \peiNegUnknowP is that it controls FDR for nonpositive effect, as we detail in Appendix~\ref{sec:nonpositive-effect} and show in numerical experiments in the next section.  
\end{remark}


\subsection{Numerical experiments}
We follow the simulation setting in Section~\ref{sec:sim_setup}, except different propensity scores specified as a function of covariates. Let the 
the treatment effect be
\begin{align} \label{eq:bias-sparse-fixC}
    \Delta(X_i) =  15X_i^3(3) \one\{X_i(3) > 1\} - 3X_i(1)/2,
\end{align}
which is the treatment effect in~\eqref{eq:bias-sparse} with $S_\Delta = 3$. Consider the case where subjects with positive effects coincides with those having higher propensity scores:
\begin{align} \label{eq:propensity}
    \pi_i = \pi(X_i) = (1/2 + \delta) \one\{\Delta(X_i) > 0\} + 1/2 \one\{\Delta(X_i) = 0\} + (1/2 - \delta) \one\{\Delta(X_i) < 0\},
\end{align}
where $\delta \in (0, 0.5)$ denotes the deviation of the propensity score bounds to 1/2. 

In the case of unknown propensity scores, several approaches are explored: estimating the propensity scores as in \peiZeroUnknowP and \peiNegUnknowP; and falsely treating all propensity scores as $1/2$ and implementing the original algorithms, referred to as \peiZero and \peiNeg (detailed in Appendix~\ref{sec:nonpositive-effect} for controlling error of nonpositive effects). They are compared with the oracle algorithms where the true propensity scores are plugged into \peiZeroUnknowP and \peiNegUnknowP, denoted as \peiZeroKnowP and \peiNegKnowP.  We are interested in the sensitivity of \peiZero and \peiNeg because we might assume propensity scores to be $1/2$ while they differ in practice.

\peiZeroUnknowP and \peiNegUnknowP with estimated propensity scores appear to control FDR at the target level for their corresponding null hypotheses, respectively (see Figure~\ref{fig:unknownP}). They have less power compared with the \peiZeroKnowP and \peiNegKnowP, which is expected since the latter two methods make use of the true propensity scores.

\begin{figure}[h!]
\centering
    \begin{subfigure}[t]{0.32\textwidth}
        \centering
        \includegraphics[width=1\linewidth]{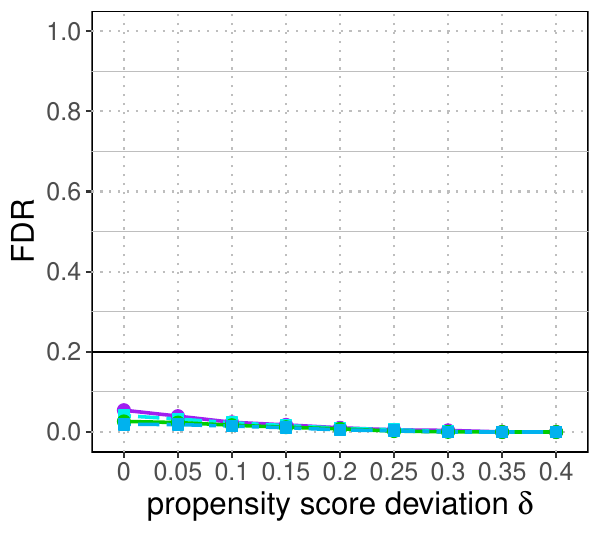}
        \subcaption{FDR for the zero-effect null~\eqref{eq:hp_twoside}.}
        \label{fig:FDR_zero_unknown}
    \end{subfigure}
    \hfill
    \begin{subfigure}[t]{0.32\textwidth}
        \centering
        \includegraphics[width=1\linewidth]{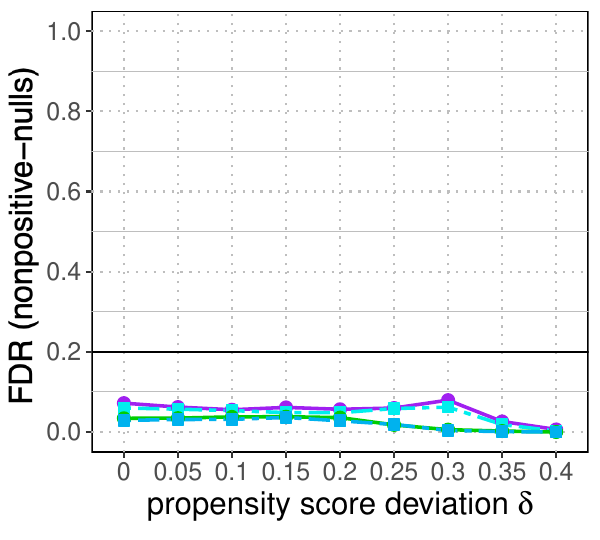}
        \subcaption{FDR for the nonpositive-effect null~\eqref{eq:hp_oneside}.}
        \label{fig:FDR_neg_unknown}
    \end{subfigure}
    \hfill
    \begin{subfigure}[t]{0.32\textwidth}
    \centering
        \includegraphics[width=1\linewidth]{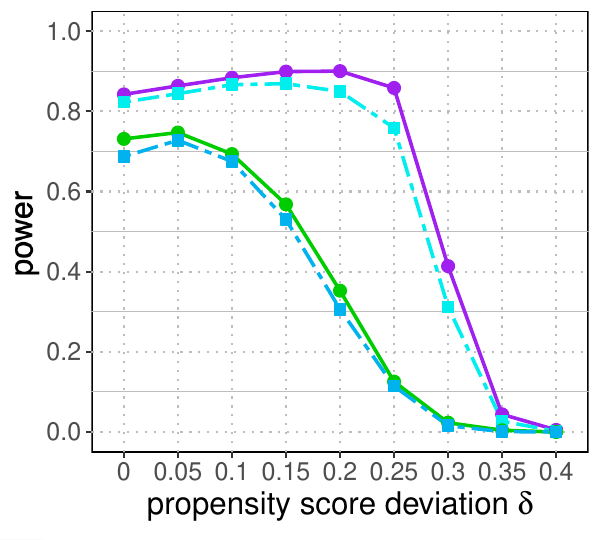}
        \caption{Power of identifying subjects with positive effects.}
        \label{fig:power_pos_unknown}
    \end{subfigure}
    
    \vfill
    \begin{subfigure}[t]{1\textwidth}
        \centering
        \includegraphics[width=0.7\linewidth]{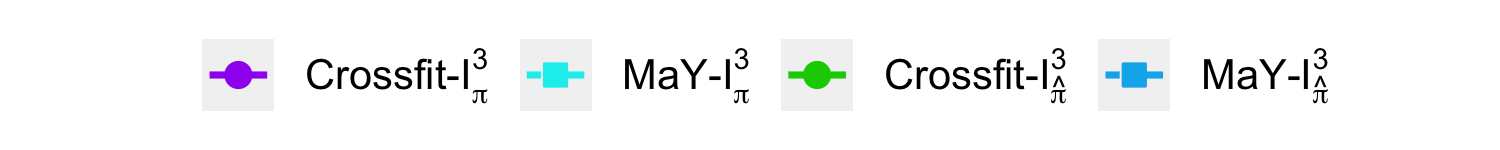}
    \end{subfigure}
    
    \caption{Performance of \peiZeroUnknowP and \peiNegUnknowP, which estimate the propensity scores, compared with \peiZeroKnowP and \peiNegUnknowP, which use the knowledge of the true propensity scores, when the treatment effect specified as model~\eqref{eq:bias-sparse-fixC} and the propensity score deviates from 1/2 by $\delta$ where $\delta$ varies in $\{0, 0.05, 0.1, 0.15, 0.2, 0.25, 0.3, 0.35, 0.4\}$. Both \peiZeroUnknowP and \peiNegUnknowP appears to control FDR, and have similar power. Their power are lower than \peiZeroKnowP and \peiNegUnknowP because the latter additionally use the true propensity scores.}
    \label{fig:unknownP}
\end{figure}

\begin{figure}[h!]
\centering
    \begin{subfigure}[t]{0.32\textwidth}
        \centering
        \includegraphics[width=1\linewidth]{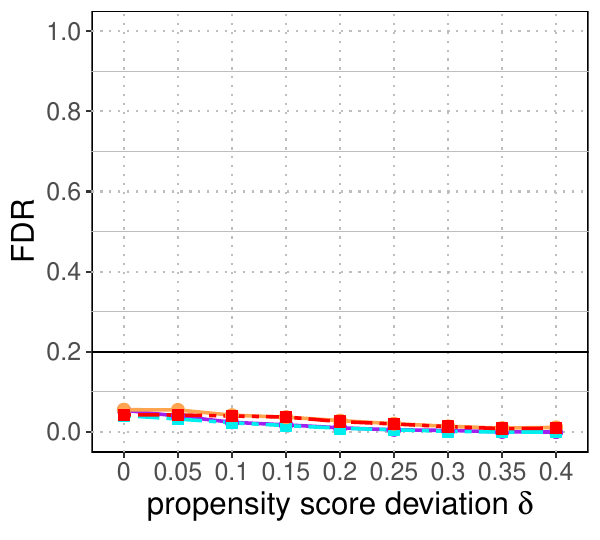}
        \subcaption{FDR for the zero-effect null~\eqref{eq:hp_twoside}.}
        \label{fig:FDR_zero_false}
    \end{subfigure}
    \hfill
    \begin{subfigure}[t]{0.32\textwidth}
        \centering
        \includegraphics[width=1\linewidth]{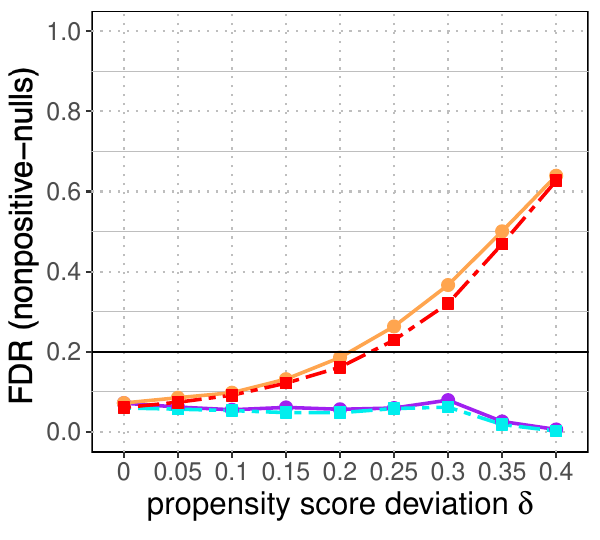}
        \subcaption{FDR for the nonpositive-effect null~\eqref{eq:hp_oneside}.}
        \label{fig:FDR_neg_false}
    \end{subfigure}
    \hfill
    \begin{subfigure}[t]{0.32\textwidth}
    \centering
        \includegraphics[width=1\linewidth]{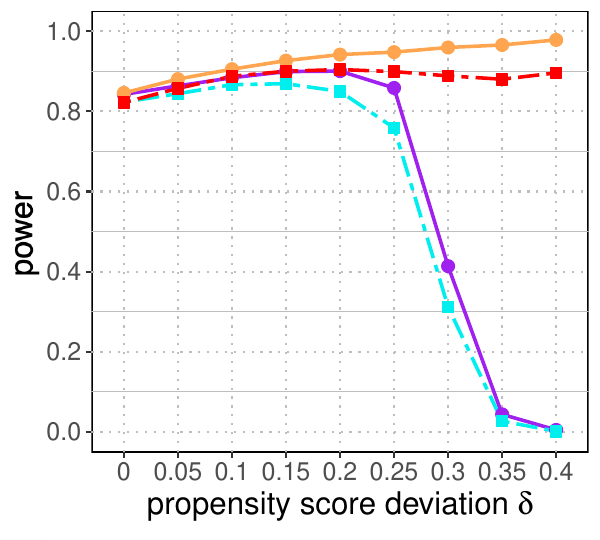}
        \caption{Power of identifying subjects with positive effects.}
        \label{fig:power_pos_false}
    \end{subfigure}
    
    \vfill
    \begin{subfigure}[t]{1\textwidth}
        \centering
        \includegraphics[width=0.7\linewidth]{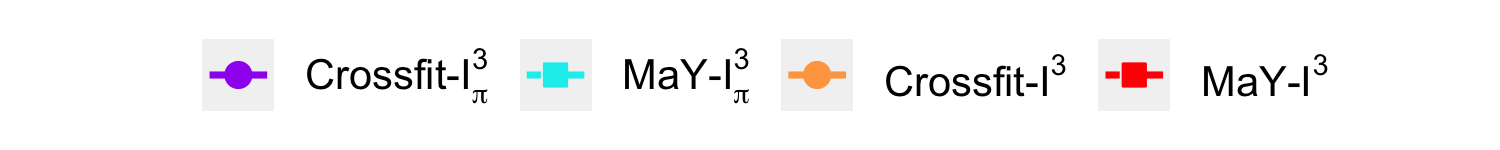}
    \end{subfigure}
    
    \caption{Performance of \peiZero and \peiNeg, which falsely treat all propensity scores as $1/2$, compared with \peiZeroKnowP and \peiNegUnknowP, which use the true propensity scores, when the treatment effect specified as model~\eqref{eq:bias-sparse-fixC} and the propensity score deviates from 1/2 by $\delta$ where $\delta$ varies in $\{0, 0.05, 0.1, 0.15, 0.2, 0.25, 0.3, 0.35, 0.4\}$. Power of the \peiZero and \peiNeg increases because they do not suffer from conservative FDR estimator as $\delta$ increases. Although FDR for the nonpositive-effect null grows to exceed the target level when $\delta$ is larger than 0.2, FDR control for the zero-effect null seems to hold even when the true propensity scores are vastly different from $1/2$.}
    \label{fig:falseP}
\end{figure}

When all propensity scores are falsely treated as $1/2$, we can implement \peiZero and \peiNeg (see Figure~\ref{fig:falseP}). In our experiments, the FDR for the zero-effect null seems to be controlled below the target level even when the true propensity scores are extreme (with $\pi_{\min} = 0.1$ and $\pi_{\max} = 0.9$ when $\delta = 0.4$). It coincides with our claim on doubly robust FDR control once noticing that \peiNeg is equivalent to \peiNegUnknowP when $\widehat \pi_{\min} = \widehat \pi_{\max} = 1/2$. In such a case, the propensity scores are poorly estimated $|\widehat \pi_{\min} - \pi_{\min}| = |\widehat \pi_{\max} - \pi_{\max}| = 0.4$, but FDR can be small when the expected outcome $\mathbb{E}(Y_i \mid \{X_i\}_{i=1}^n)$ is well-estimated by $\widehat m^{-\mathcal{I}}(X_i)$. The FDR for the nonpositive-effect null can exceed the target level when the deviation $\delta$ is large and the propensity score estimation is poor, corresponding to the case where $\pi_{\min} \leq 0.25$ and $\pi_{\max} \geq 0.75$. The power of \peiZero and \peiNeg does not follow the same trend as \peiZeroKnowP and \peiNegKnowP when $\delta$ grows large, because their FDR estimator does not suffer from conservativeness introduced by extreme propensity scores.

\section{FDR control at a subgroup level} 
\label{apd:subgroup}
Our proposed interactive methods control FDR on \textit{individual} level, which means upper bounding the proportion of falsely identified subjects. In this section, we show that the idea of interactive testing can be extended to control FDR on \textit{subgroup} level, where we aim at identifying multiple subgroups with positive effects and upper bounding the proportion of falsely identified subgroups. Recall that FDR control at a subgroup level is studied by \citet{karmakar2018false} as we review in Section~\ref{sec:review}. 

\subsection{Problem setup}
Let there be $G$ non-overlapping subgroups $\mathcal{G}_g$ for $g \in [G] \equiv \{1,\ldots, G\}$. The null hypothesis for each subgroup is defined as zero effect for all subjects within:
\begin{align} \label{hp:subgroup}
    \mathcal{H}_{0g}: H_{0i}^\text{zero} \text{ is true for all } i \in \mathcal{G}_g,
\end{align}
or equivalently, $\mathcal{H}_{0g}: \mathcal{G}_g \subseteq \mathcal{H}_0$ (recall that $\mathcal{H}_0$ is the set of true null subjects). Let $D_g$ be the decision function receiving the values $1$ or $0$ for whether $\mathcal{H}_{0g}$ is rejected or not rejected, respectively, and the FDR at a subgroup level is defined as:
\begin{align*} 
    \mathrm{FDR}^\text{subgroup} := \mathbb{E}\left[ \frac{|\{g \in [G]: \mathcal{G}_g \subseteq \mathcal{H}_0, D_g = 1|}{\max\{|\{g \in [G]: D_g = 1|,1\}}\right].
\end{align*}
Same as the algorithms at an individual level, the algorithms we propose at a subgroup level can be applied to samples that are paired or unpaired. For simple notation, we use $\{Y_i, A_i, X_i\}$ to denote the observed data for subject $i$ when the samples are unpaired, and for pair $i$ when the samples are paired (where $Y_i = \{Y_{i1}, Y_{i2}\}$ and similarly for $A_i$ and $X_i$).

\subsection{An interactive algorithm to identify subgroups}
We first follow the same steps of \citet{karmakar2018false} to define subgroups and generate the $p$-value for each subgroup. Specifically, the subgroups $\mathcal{G}_g$ for $g \in [G]$ is defined using the outcomes and covariates $\{Y_j, X_j\}_{j=1}^n$ (by an arbitrary algorithm or strategy, such as grouping subjects with the same covariates). For each subgroup~$\mathcal{G}_g$, we can compute a $p$-value $P_g$ by the classical Wilcoxon test (or using a permutation test, which obtains the null distribution by permuting the treatment assignment $\{A_i\}_{i=1}^n$). 

The interactive procedure we propose differs from \citet{karmakar2018false} by how we process the $p$-values of the subgroups. We adopt the work of \citet{lei2017star} that proposes an interactive procedure with FDR control for generic multiple testing problems. The key property that allows human interaction while guaranteeing valid FDR control is similar to that in the \pei: the independence between the information used for selection and that used for FDR control. Here with the $p$-values of subgroups, the two independent parts are 
\begin{align*}
    P^1_g := \min\{P_g, 1 - P_g\},
\end{align*}
which is revealed to the analyst for selection and
\begin{align*}
    P^2_g := 2\cdot \one\{P_g < \tfrac{1}{2}\} - 1,
\end{align*}
which is masked (hidden) for FDR control. Notice that for a null subgroup with a uniform $p$-value, $(P^1_g, P^2_g)$ are independent, and we have that 
\begin{align} \label{eq:sign_property_subgroup}
    \mathbb{P}(P^2_g = 1 \mid P^1_g, [Y_i, X_i]_{i \in \mathcal{G}_g}) \leq 1/2,
\end{align}
because the $p$-values obtained by permutating assignments is uniform when conditional on the outcomes and covariates. We remark that the above property is similar to property~\eqref{eq:sign_property_zero} and~\eqref{eq:sign_property_neg} that lead to valid FDR control at an individual level. 

\begin{algorithm}[h]
   \caption{An interactive procedure for subgroup identification.}
   \label{alg:subgroup}
\begin{algorithmic}
   \STATE {\bfseries Initial state:} Explorer (E) knows the covariates, outcomes $\{Y_i, X_i\}_{i=1}^n$.
   \STATE Oracle (O) knows the treatment assignments
   $\{A_i\}_{i=1}^n$.
   \STATE Target FDR level $\alpha$ is public knowledge.
   \STATE {\bfseries Initial exchange:} Set $t=1$.
   \STATE 1.~E defines subgroups $\{\mathcal{G}_g\}_{g=1}^G$ using $\{Y_i, X_i\}_{i = 1}^n$.
   \STATE 2.~Both players initialize $\mathcal{R}_0 = [G]$, and E informs O about the subgroup division.
   \STATE 3.~O compute the $p$-value for each subgroup $\{P_g\}_{g=1}^G$, and decompose each $p$-value as $P^1_g := \min\{P_g, 1 - P_g\}$ and $P^2_g := 2\cdot \one\{P_g < \tfrac{1}{2}\} - 1$. 
   \STATE 4. O then divides $\mathcal{R}_t$ into $\mathcal{R}_t^- := \{g \in \mathcal{R}_t: P^2_g \leq 0\}$ and $\mathcal{R}_t^+ := \{g \in \mathcal{R}_t: P^2_g > 0\}$.
   \STATE 5. O reveals $\{P_g^1\}_{g=1}^G$, $|\mathcal{R}_t^-|$ and $|\mathcal{R}_t^+|$ to E.
   \STATE {\bfseries Repeated interaction:} 6. E checks if {$\widehat{\mathrm{FDR}}^{\text{subgroup}}(\mathcal{R}_t) \equiv \frac{|\mathcal{R}_t^-| + 1}{\max\{|\mathcal{R}_t^+|,1\}} \leq \alpha$}.
   \STATE 7.  If yes, E sets $\tau=t$, reports $\mathcal{R}_\tau^+$ and exits;
   \STATE 8.~Else, E picks any $g_t^* \in \mathcal{R}_{t-1}$ using everything E currently knows.
   \STATE (E tries to pick an $g^*_t$ that they think is null; E hopes that $P_g^2 \leq 0$.)
   \STATE 9.~O reveals $\{A_i\}_{i\in \mathcal{G}_g}$ to E, who also infers $P_g^2$.
   \item 10. E updates ${\mathcal{R}_{t+1} = \mathcal{R}_{t}\backslash \{g_t^*\}}$, and also $|\mathcal{R}^+_{t+1}|$ and $|\mathcal{R}^-_{t+1}|$;
   \STATE 11. Increment $t$ and go back to Step 6.
\end{algorithmic}
\end{algorithm}

Similar to the proposed methods at an individual level, the interactive procedure for subgroups progressively excludes subgroups and recursively estimates the FDR.
Let the candidate rejection set~$\mathcal{R}_t$ be a set of selected subgroups, starting from all subgroups included $\mathcal{R}_0 = [G]$. We interactively shrink~$\mathcal{R}_t$ using the available information:
\begin{align*}
    \mathcal{F}_{t-1}^{\text{subgroup}} = \sigma\left(\{P^1_g, [Y_i, X_i]_{i \in \mathcal{G}_g}\}_{g \in \mathcal{R}_{t-1}}, \{P_g, [Y_j, A_j, X_j]_{j \in \mathcal{G}_g}\}_{g \notin \mathcal{R}_{t-1}}, \sum_{g \in \mathcal{R}_{t-1}} P_g^2\right),
\end{align*}
which masks (hides) the partial $p$-value $P_g^2$ and the treatment assignment $A_i$ for candidate subgroups in $\mathcal{R}_{t-1}$; and the sum $\sum_{g \in \mathcal{R}_{t-1}} P_g^2$ is mainly provided for FDR estimation.
Similar to our previously proposed interactive procedures, the FDR estimator is defined as:
\begin{align} \label{eq:fdr_hat_subgroup}
    \widehat{\text{FDR}}^{\text{subgroup}}(\mathcal{R}_{t}) = \frac{|\mathcal{R}_t^-| + 1}{\max\{|\mathcal{R}_t^+|,1\}},
\end{align}
with $\mathcal{R}_t^+ = \{g \in \mathcal{R}_t: P^2_g = 1\}$ and $\mathcal{R}_t^- = \{g \in \mathcal{R}_t: P^2_g = -1\}$. 
The algorithm shrinks $\mathcal{R}_t$ until time $\tau := \inf\{t: \widehat{\text{FDR}}^{\text{subgroup}}(\mathcal{R}_{t}) \leq \alpha\}$, and identifies only the subgroups in $\mathcal{R}_\tau^+$, as summarized in Algorithm~\ref{alg:subgroup}.
Details of strategies to select subgroups based on the revealed $p$-value and covariates can be found in \citet{lei2017star}. As a comparison, \citet{karmakar2018false} use the same set of $p$-values $\{P_g\}_{g \in [G]}$, and control FDR by the classical BH procedure.

\subsection{Numerical experiments} We compare the performance of our proposed interactive procedure for subgroup identification with the method proposed by \citet{karmakar2018false}, following an experiment in their paper. Suppose each subject is recorded with two discrete covariates $X_i = \{X_i(1), X_i(2)\}$ where $X_i(1) \in \{1, \ldots, 40\}$ takes $40$ levels with equal probability, and $X_i(2)$ is binary with equal probability (for example, $X_i(1)$ could encode the city subject $i$ lives in, and $X_i(2)$ the gender). The treatment effect~$\Delta(X_i)$ is a constant $\delta$ if $X_i(1)$ is even, and we vary $\delta$ in six levels. We conduct the above experiment in two cases: unpaired samples ($n = 2000$) with independent covariates and paired samples ($n = 1000$) whose covariate values are the same for subjects within each pair.

\begin{figure}[h!]
\centering
    \begin{subfigure}[t]{0.3\textwidth}
        \centering
        \includegraphics[width=1\linewidth]{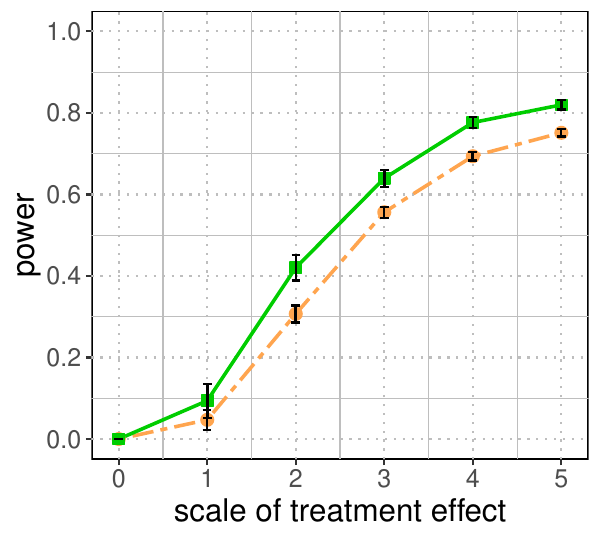}
        \subcaption{Unpaired samples.}
        \label{fig:unpaired_subgroup}
    \end{subfigure}
    \hfill
    \begin{subfigure}[t]{0.3\textwidth}
        \centering
        \includegraphics[width=1\linewidth]{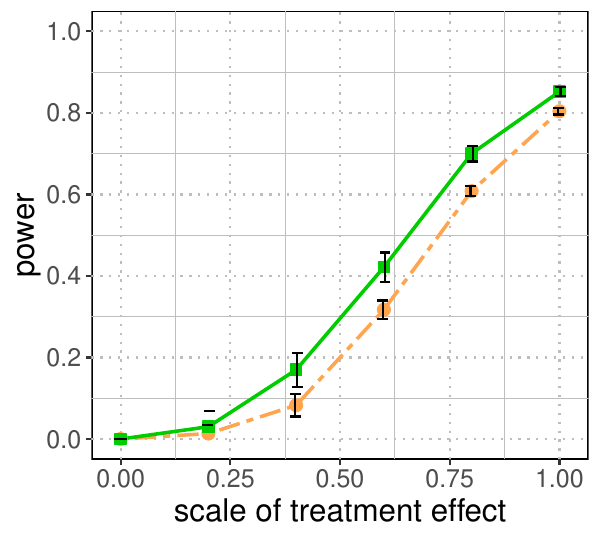}
        \subcaption{Paired samples.}
        \label{fig:paired_subgroup}
    \end{subfigure}
    \hfill
    \begin{subfigure}[t]{0.3\textwidth}
        \centering
        \includegraphics[width=1\linewidth]{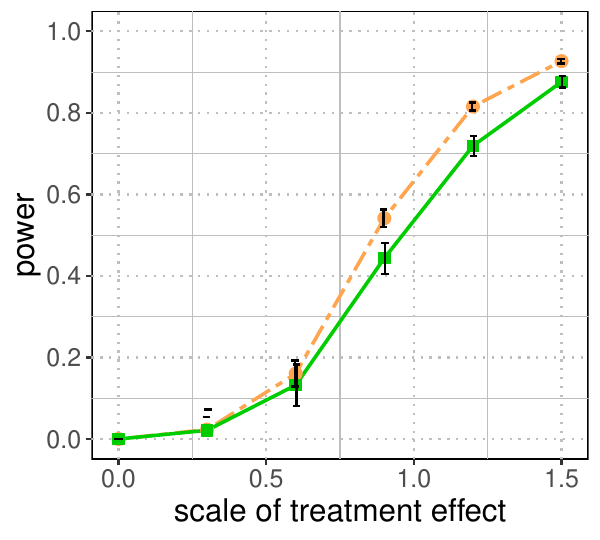}
        \subcaption{Only a few subgroups are non-nulls.}
        \label{fig:few_subgroup}
    \end{subfigure}
    
    \hfill
    \begin{subfigure}[t]{1\textwidth}
        \centering
        \includegraphics[width=0.7\linewidth]{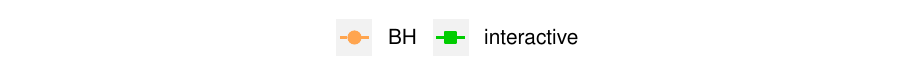}
    \end{subfigure}

    \caption{Performance of methods to identify subgroups with positive effects: the BH procedure and the interactive procedure (for $80$ subgroups defined by the distinct values of covariates). We vary the scale of treatment effect under unpaired or paired samples. In both cases, the interactive procedure can have higher power than the BH procedure. When the number of non-null subgroups is too small (less than 20), the BH procedure can have higher power. The error bar marks two standard deviations from the center.}
    \label{fig:subgroups}
\end{figure}

Recall that the subgroups can be defined by covariates and outcomes. Here, since the covariates are discrete, we define subgroups by different values of $(X_i(1), X_i(2))$, resulting in 80 subgroups. The interactive procedure tends to have higher power than the BH procedure (see Figure~\ref{fig:unpaired_subgroup} and Figure~\ref{fig:paired_subgroup}) because it focuses on the subgroups that are more likely to be the non-nulls using the excluding process, and utilizes the covariates together with the $p$-values to guide the algorithm. Meanwhile, the BH procedure does not account for covariates once the $p$-values are calculated. Nonetheless, the interactive procedure can have lower power when the total number of subgroups that are truly non-null is small. We simulate the case where a subject has a positive effect $\delta$ if $X_i(1)$ a multiplier of $4$ (i.e., $X_i(1)/2$ is even), so that there are 20 non-null subgroups in total (previously 40 non-nulls). The power of the interactive procedure is lower than the BH procedure (see Figure~\ref{fig:few_subgroup}) because the FDR estimator in~\eqref{eq:fdr_hat_subgroup} can be conservative when $|\mathcal{R}^+|$ is small due to a small number of true non-nulls (for example, with FDR control at $\alpha = 0.2$, we need to shrink $\mathcal{R}_t$ until $|\mathcal{R}^-| < 3$ when $|\mathcal{R}^+|$ is around 20). 

A side note is that we define the subgroups by distinct values of the covariates, whereas \citet{karmakar2018false} suggest forming subgroups by regressing the outcomes on covariates using a tree algorithm. In their experiments and several numerical experiments we tried, we find that the number of subgroups defined by the tree algorithm is usually less than ten. However, we think the FDR control is less meaningful when the total number of subgroups is small. To justify our comment, note that an algorithm with valid FDR control at level $\alpha$ can make zero rejection with probability $1 - \alpha$ and reject all subgroups with probability~$\alpha$, which can happen when the total number of subgroups is small. In contrast, with a large number of subgroups, a reasonable algorithm is unlikely to jump between the extremes of making zero rejection and rejecting all $n$ subgroups; and thus, controlling FDR indeed informs that the proportion of false identifications is low for the evaluated algorithm.

\subsection{Explanation of the higher power achieved by the interactive procedure} Although the interactive procedure and the BH procedure define the same set of subgroups and corresponding $p$-values, the interactive procedure has two properties that potentially improve the power from the BH procedure:
(a)~it excludes possible null subgroups so that it can be less sensitive to a large number of nulls, whereas the BH procedure considers all the subgroups at once; (b)~ the interactive procedure additionally uses the covariates. We can separately evaluate the effect of the above two properties by implementing two versions of Algorithm~\ref{alg:subgroup}, which differ in the strategy to select subgroups in step~$\mathfrak{a}$. Specifically, the adaptive procedure selects the subgroup whose revealed (partial) $p$-value $P_g^1$ is the smallest (not using the covariates); and the interactive procedure selects the subgroup by an estimated probability of the $P_g^2$ to be positive (using the revealed $P_g^1$, the covariates, and the outcomes). 

\begin{figure}[!htb]
\centering
\hspace{1cm}
    \begin{subfigure}[t]{0.32\textwidth}
        \centering
        \includegraphics[width=1\linewidth]{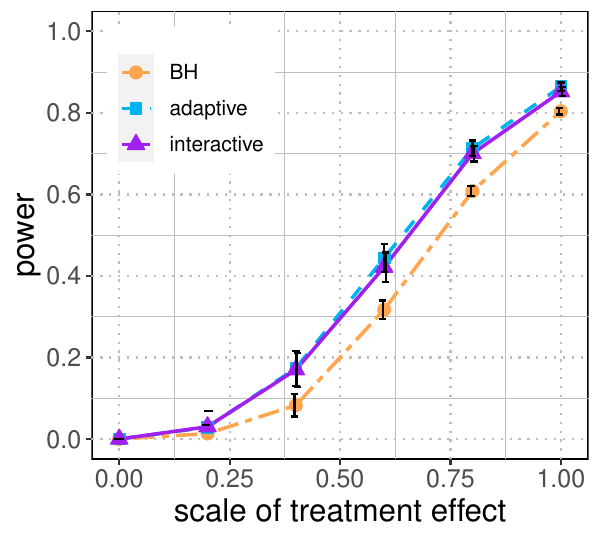}
        \subcaption{Effect as discrete function of covariates.}
        \label{fig:paired_subgroup_discrete}
    \end{subfigure}
    \hfill
    \begin{subfigure}[t]{0.32\textwidth}
        \centering
        \includegraphics[width=1\linewidth]{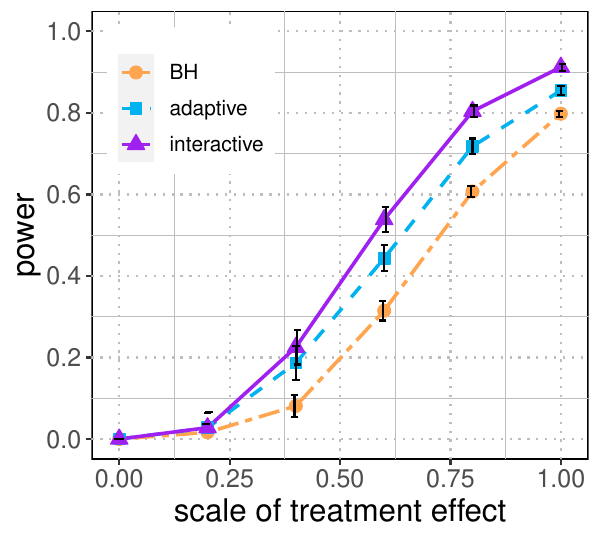}
        \subcaption{Effect as a simpler function of covariates.}
        \label{fig:paired_subgroup_smooth}
    \end{subfigure}
\hspace{1cm}
    
    \caption{Power of two methods for subgroup identification: the BH procedure proposed by~\citet{karmakar2018false}, the adaptive procedure, and the interactive procedure under different types of treatment effect (we define $80$ subgroups by discrete values of the covariates). Our proposed interactive procedure tends to have higher power than the BH procedure because~(1) it excludes possible nulls (shown by higher power of the adaptive procedure than the BH procedure in both plots); and (2)~it additionally uses the covariates (shown when the treatment effect can be well learned as a function of covariates in the right plot).}
    \label{fig:pair-subgroup-bias-sparse}
\end{figure}

To see if both properties of Algorithm~\ref{alg:subgroup} contribute to the improvement of power from the BH procedure, we tested the methods under two simulation settings. Recall that the previous experiment defines a positive treatment effect when the discrete covariate $X_i(1) \in \{1, \ldots, 40\}$ is even. Here, we add another case where the treatment effect is positive when $X_i(1) \leq 20$, so that the density of subgroups with positive effects is the same as previous, but the treatment effect is a simpler function of the covariates. Hence in the latter case, we would expect the interactive procedure to learn this function of covariates rather accurately, and have higher power than the adaptive procedure which does not use the covariates; as confirmed in Figure~\ref{fig:paired_subgroup_smooth}. In the former simulation setting where the treatment effect is not a smooth function of the covariates and hard to be learned, the adaptive procedure and interactive procedure have similar power (Figure~\ref{fig:paired_subgroup_discrete}). Still, they have higher power than the BH procedure because they exclude possible null subgroups.

\section{Additional numerical experiments for variations in treatment effect}
\label{apd:sim_more}
We have seen the numerical results of the proposed methods in previous sections where the treatment effect is defined in~\eqref{eq:bias-sparse} with sparse and strong positive effect, and dense and weak negative effect. This section presents three more examples of the treatment effect.

\subsection{Linear effect}
Let the treatment effect be
\begin{align} \label{eq:linear-both}
    \Delta(X_i) = S_\Delta \cdot [2X_i(1) X_i(2) +  2X_i(3)],
\end{align}
where $S_\Delta > 0$. In this case, all subjects have treatment effects, and the scale correlates with the covariates (with interaction terms) linearly. Thus, the linear-BH procedure has valid error control as shown in Figure~\ref{fig:linear} (unlike other cases with nonlinear treatment effect).

\subsection{Sparse and strong effect that is positive}
Let the treatment effect be
\begin{align} \label{eq:sparse-oneside}
    \Delta(X_i) = S_\Delta \cdot [5X_i^3(3) \one\{X_i(3) > 1\}],
\end{align}
where $S_\Delta > 0$. Here, the subjects with $X_i(3) > 1$ have positive treatment effects. Although linear-BH procedure seems to have high power, its FDR is largely inflated since the assumption of linear correlation does not hold (see Figure~\ref{fig:sparse-oneside}). In contrast, our proposed methods and Selective SeqStep+ have valid FDR control, and our proposed methods have higher power.

\subsection{Sparse and strong effect in both directions}
Let the treatment effect be
\begin{align} \label{eq:sparse-twoside}
    \Delta(X_i) = S_\Delta \cdot [5X_i^3(3) \one\{|X_i(3)| > 1\}],
\end{align}
where $S_\Delta > 0$. Here, the subjects with $X_i(3) > 1$ have positive treatment effects and those with $X_i(3) < -1$ have negative treatment effects; the scale and proportion of effects in both directions are the same.
The power of Selective SeqStep+ is trivial because $|E_i|$ used in ordering cannot inform the direction of treatment effect. In contrast, the power of \peiZero and \peiNeg are slightly lower than the previous setting since there is additionally negative effect in this example (see Figure~\ref{fig:sparse-twoside}). 

\begin{figure}[H]

    \centering
    \begin{subfigure}[t]{1\textwidth}
        \centering
        \includegraphics[width=0.3\linewidth]{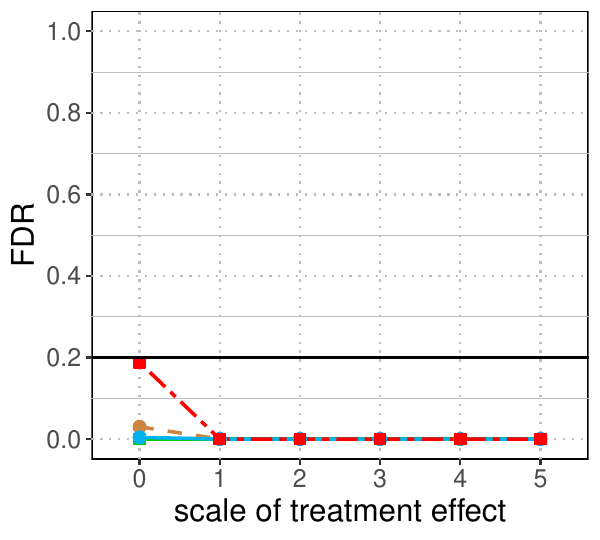}
    \hfill
        \includegraphics[width=0.3\linewidth]{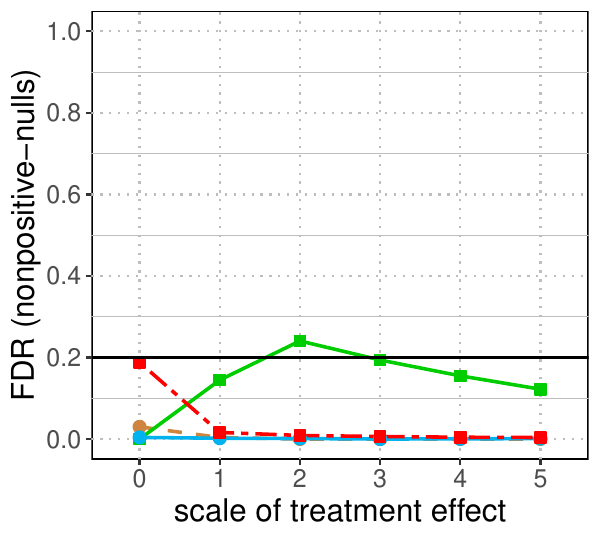}
    \hfill
        \includegraphics[width=0.3\linewidth]{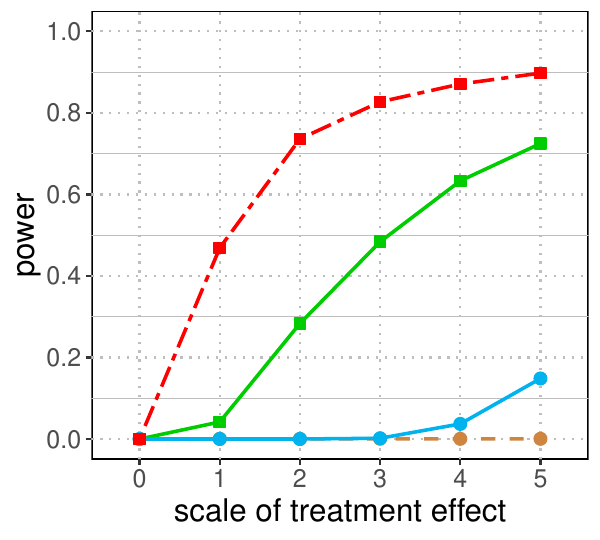}
    \caption{Dense two-sided effect (linear) as in model~\eqref{eq:linear-both}.}
    \label{fig:linear}
    \end{subfigure}
\vfill
     \begin{subfigure}[t]{1\textwidth}
        \centering
        \includegraphics[width=0.3\linewidth]{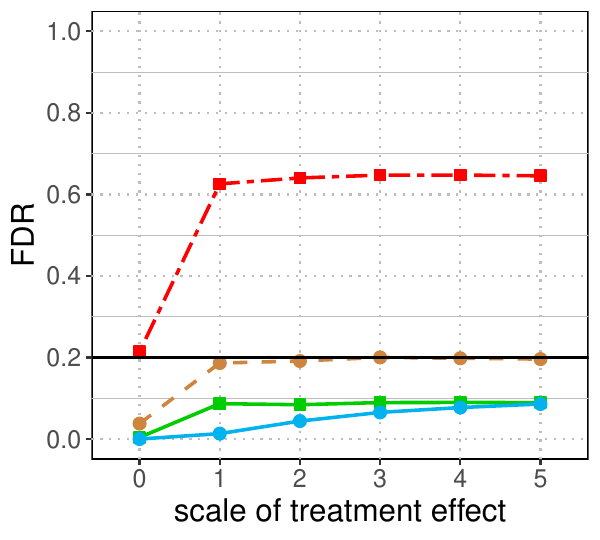}
    \hfill
        \includegraphics[width=0.3\linewidth]{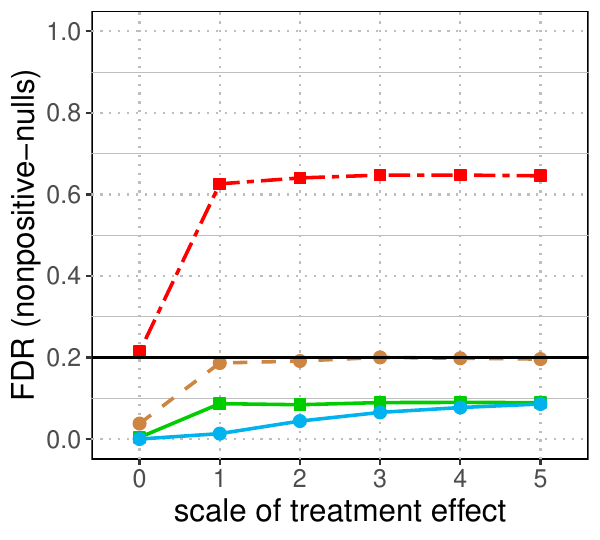}
    \hfill
        \includegraphics[width=0.3\linewidth]{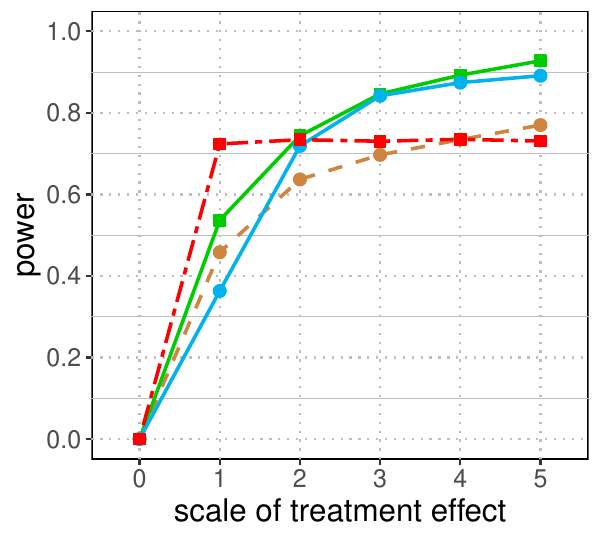}
    \caption{Sparse and strong effect that is positive (nonlinear) in model~\eqref{eq:sparse-oneside}.}
    \label{fig:sparse-oneside}
    \end{subfigure}
\vfill
    \begin{subfigure}[t]{1\textwidth}
        \centering
        \includegraphics[width=0.8\linewidth]{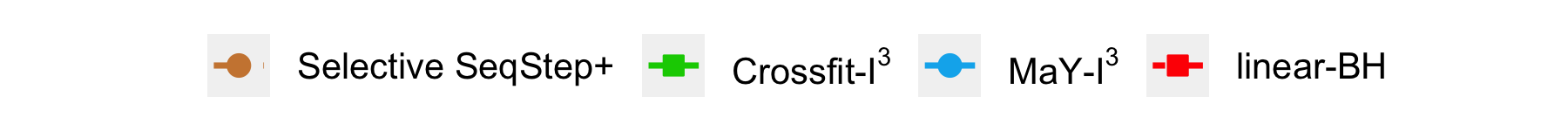}
    \end{subfigure}
\vfill
\end{figure}
\begin{figure}[H]
\ContinuedFloat 
    \begin{subfigure}[t]{1\textwidth}
        \centering
        \includegraphics[width=0.3\linewidth]{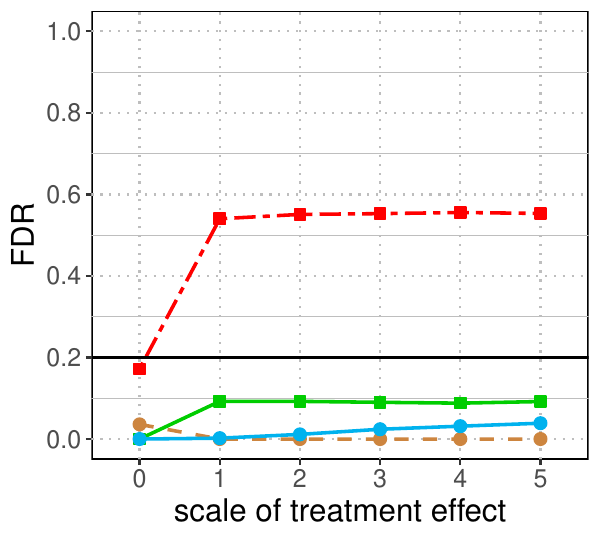}
    \hfill
        \includegraphics[width=0.3\linewidth]{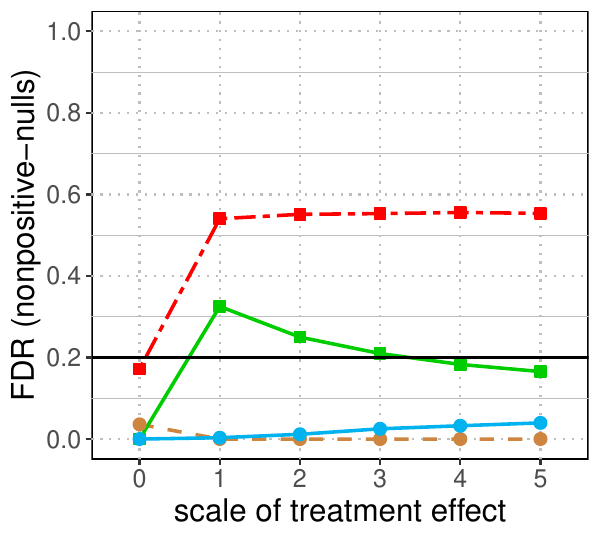}
    \hfill
        \includegraphics[width=0.3\linewidth]{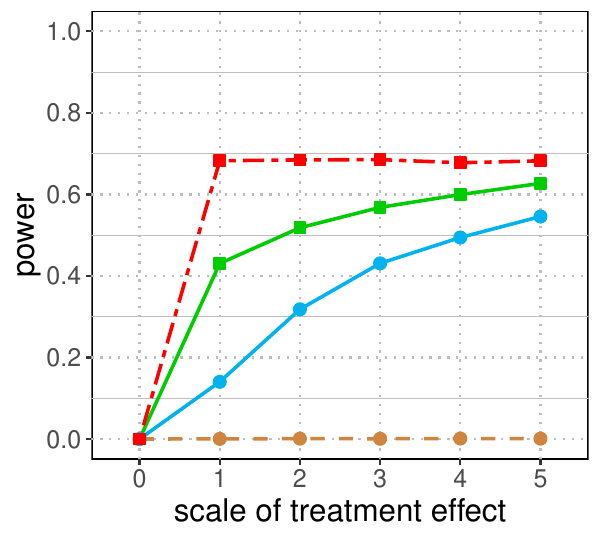}
    \caption{Sparse and strong effect in both directions (nonlinear) in model~\eqref{eq:sparse-twoside}.}
    \label{fig:sparse-twoside}
    \end{subfigure}
   
    \caption{FDR for the zero-effect null~\eqref{eq:hp_twoside} (the first column), and FDR for the nonpositive-effect null~\eqref{eq:hp_oneside} (the second column), and power (the third column) of four methods: linear-BH procedure, Selective SeqStep+, \peiZero, \peiNeg, under three types of treatment effect when varying the scale of treatment effect $S_\Delta$ in $\{0,1,2,3,4,5\}$. When the linear assumption holds as in the first row, the linear-BH procedure has valid FDR control and high power, but its FDR is large when the treatment is a nonlinear function of the covariates as in the last two rows. In contrast, the Selective SeqStep+, \peiZero and \peiNeg have valid FDR control for their target null hypotheses, respectively. Our proposed \peiZero and \peiNeg have higher power than Selective SeqStep+, especially when positive effects have comparable scale as negative effects as in the first and the third rows, because the proposed methods can additionally use the revealed treatment assignments to inform the direction of treatment effects.}
    \label{fig:unpair_varysignal}
\end{figure}

\section{A prototypical application to ACIC challenge dataset} \label{sec:imp}

We implement our proposed methods on datasets generated by \href{https://sites.google.com/view/acic2019datachallenge/home}{Atlantic Causal Inference Conference (ACIC)}, which intend to evaluate methods for average treatment effect(ATE) estimation and uses real data covariates and modified outcomes to simulate cases with heterogeneous treatment effect, heterogeneous propensity scores, etc. We take an example dataset with 500 subjects, each of which is recorded with 22 continuous covariates. The proportion of treated subjects is 0.7, indicating that the propensity scores might not be $1/2$ as in a standard randomized experiment. The actual ATE is $0.1$, rather small compared to the outcomes range $[14, 76]$, but the treatment effect could be positive and large for a subgroup of subjects and our proposed algorithms can be used to identify them.


\begin{figure}[H]
    \begin{subfigure}[t]{1\textwidth}
        \includegraphics[width=0.49\linewidth]{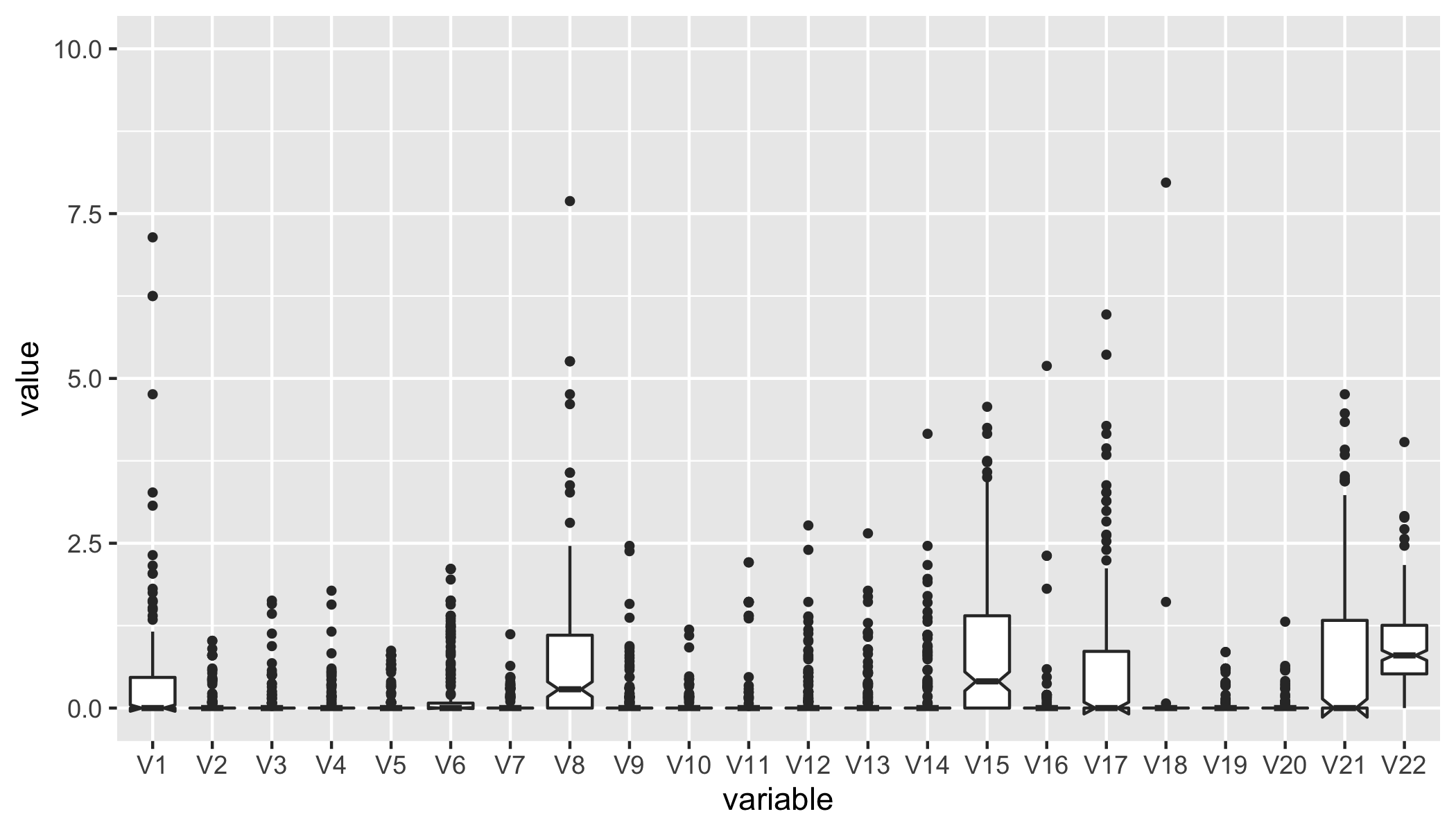}
    \hfill
        \includegraphics[width=0.49\linewidth]{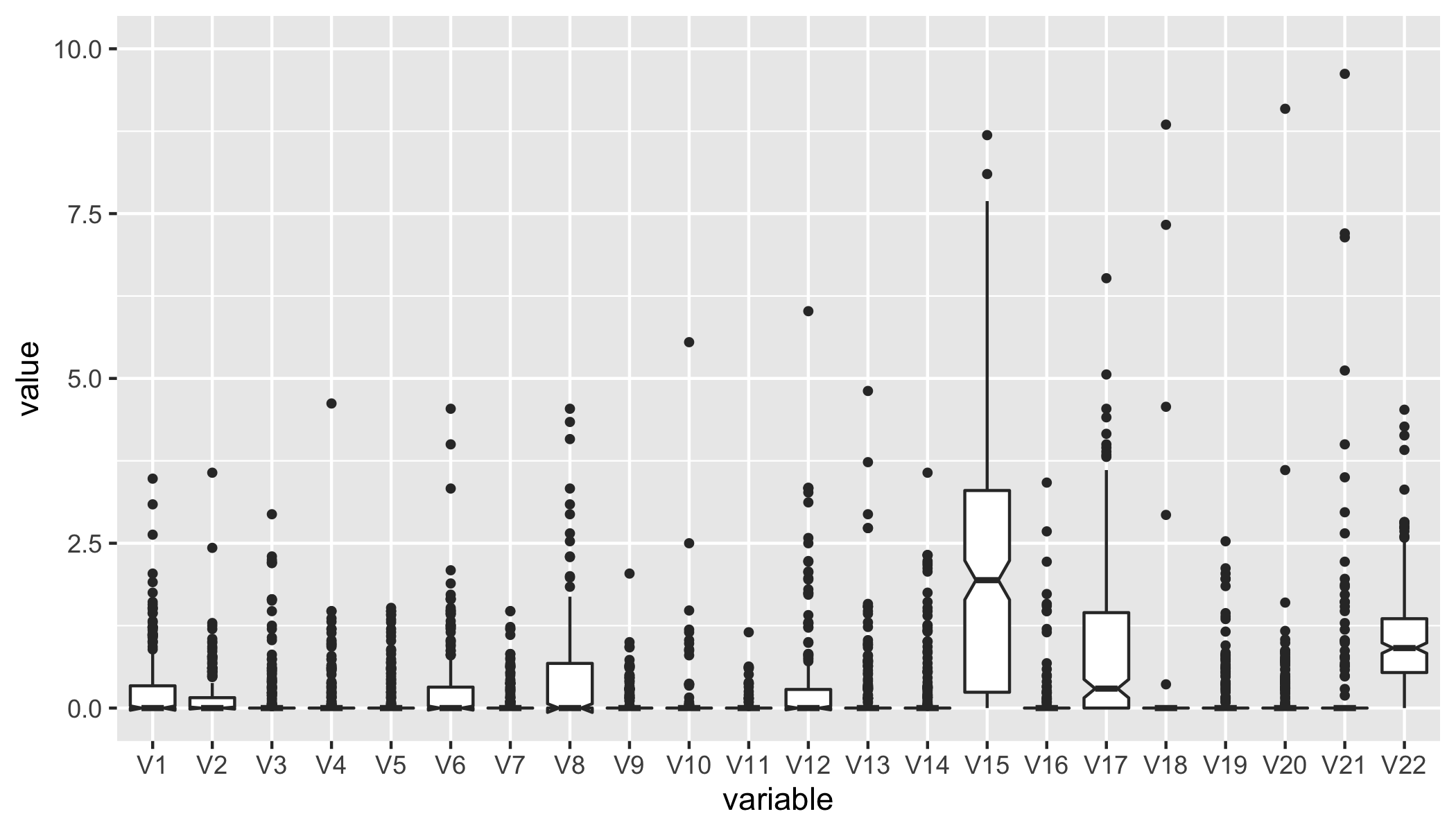}
    \caption{Boxplot of covariates for subjects identified as nonzero effect (left) versus those not being identified (right).}
    \end{subfigure}
\end{figure}
\begin{figure}[H]
\ContinuedFloat
    \begin{subfigure}[t]{1\textwidth}
        \includegraphics[width=0.49\linewidth]{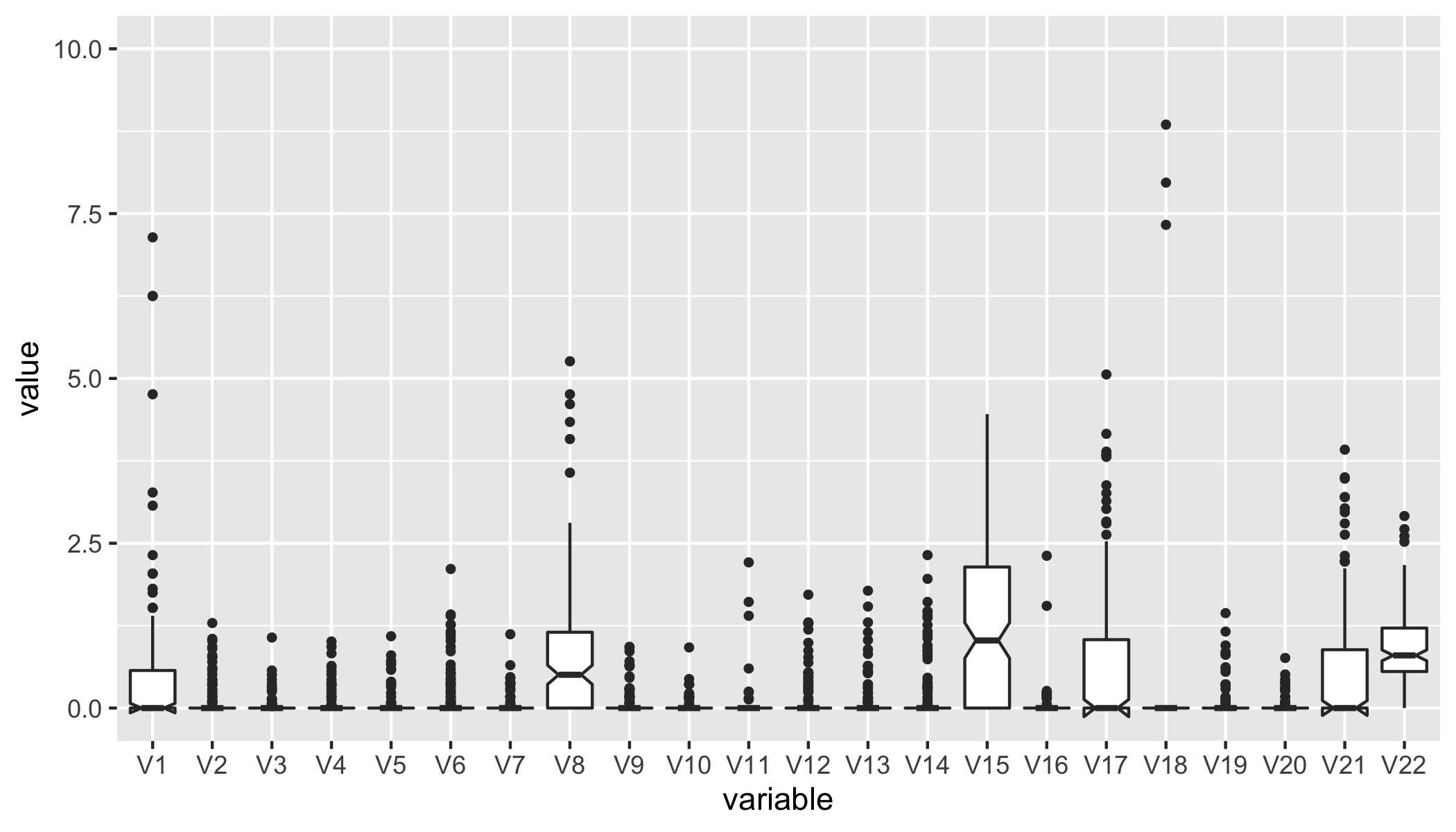}
    \hfill
        \includegraphics[width=0.49\linewidth]{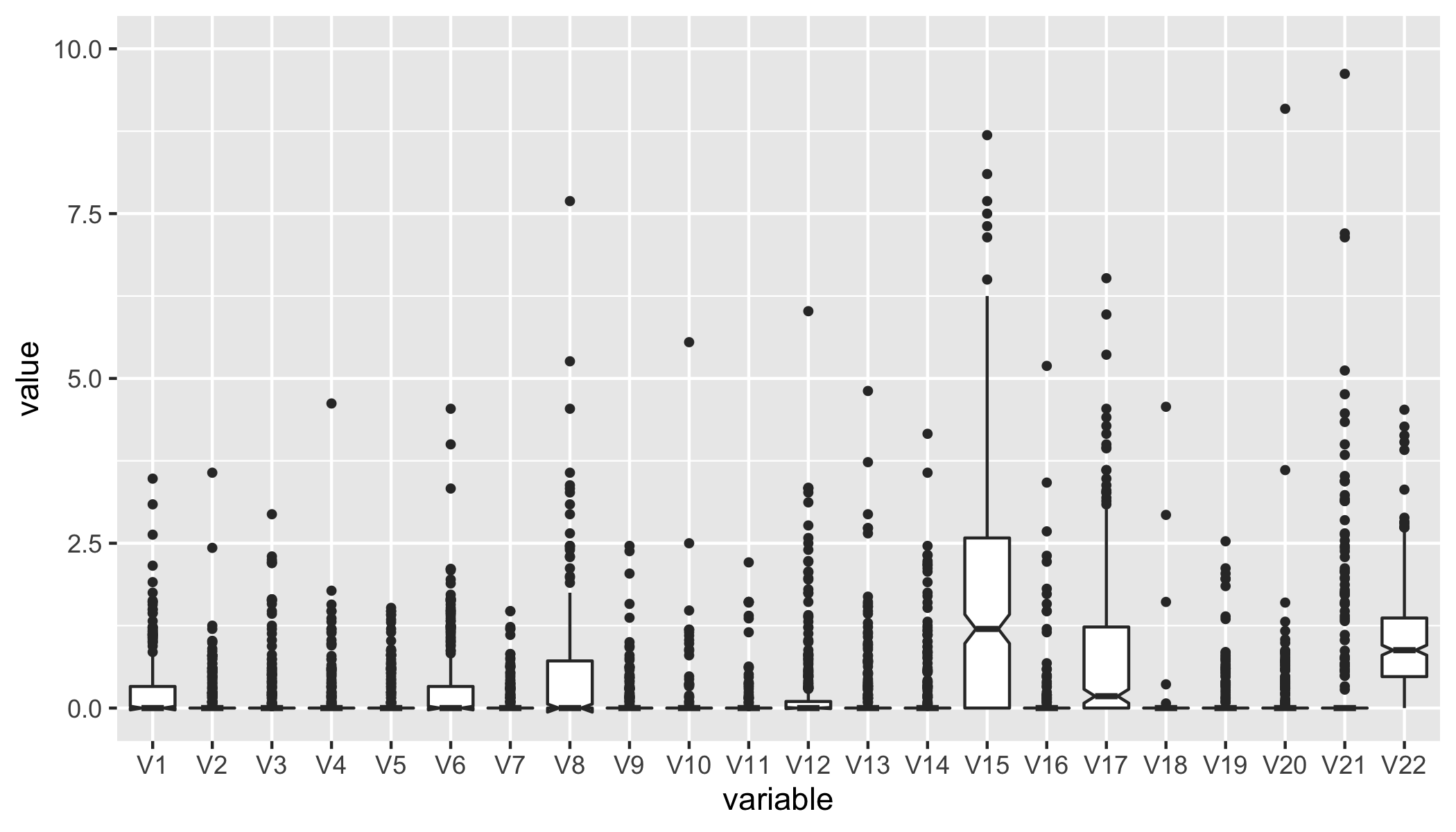}
    \caption{Boxplot of covariates for subjects identified as positive effect (left) versus those not being identified (right).}
    \end{subfigure}
    
    \caption{Characteristics of identified subjects: they tend to have larger value for variable 8, 21 and smaller value for variable 6, 15, 17, compared with not identified subjects.}
    \label{fig:realdat_covar}
\end{figure}

Four of our proposed methods are implemented with FDR control at level $\alpha = 0.2$: the \peiZero and \peiNeg which assume the propensity scores to be $1/2$ for all subjects, and \peiZeroUnknowP and \peiNegUnknowP which estimate the propensity scores. The numbers of identifications by \peiZero, \peiNeg, \peiZeroUnknowP and \peiNegUnknowP are $446, 429, 238, 162$. Among them, 234 subjects are commonly identified by \peiZero and \peiZeroUnknowP, which control the expected proportion of falsely identifying subjects with zero effect (approximately if the propensity scores are not $1/2$); and 158 subjects are commonly identified by \peiNeg and \peiNegUnknowP, which control the expected proportion of falsely identifying subjects with nonpositive effect (approximately if the propensity scores are not $1/2$). Compared with the rest subjects, the ones identified as having positive effect tend to have larger values for covariate 8, 21 and smaller values for covariate 6, 15, 17 (see Figure~\ref{fig:realdat_covar}).

\section{Summary} \label{sec:discuss}
We discuss the problem of identifying subjects with positive effects. Most existing methods identify \textit{subgroups} with positive treatment effects, and they cannot upper bound the proportion of falsely identified \textit{subjects} within an identified subgroup. In contrast, we propose \peiZero with finite-sample FDR control (i.e., the expected proportion of subjects with zero effect is no larger than $\alpha$ among the identified subjects). One advantage of the \peiZero is allowing human interaction --- an analyst (or an algorithm) can incorporate various types of prior knowledge and covariates using any working model; she can also adjust the model at any step, potentially improving the identification power. Despite this flexibility, the \peiZero achieves valid FDR control. Notably, because \peiZero incorporates covariates, it can identify subjects with positive effects, including those not treated. 

Our proposed interactive procedure can extended to various settings:
\begin{itemize}
    \item FDR control of nonpositive effects in randomized experiments (Appendix~\ref{sec:nonpositive-effect});
    \item FDR control of zero/nonpositive effects in observational studies (Section~\ref{sec:propensity_unknown} and Appendix~\ref{sec:extend_obs});
    \item paired samples (Appendix~\ref{sec:paired});
    \item FDR control at a subgroup level (Section~\ref{apd:subgroup}).
\end{itemize}


The error control for our interactive procedures is based on the independence properties between the data used for FDR control and the revealed data for interaction, such as property~\eqref{eq:sign_property_zero} for the zero-effect null and~\eqref{eq:sign_property_neg} for the nonpositive-effect null.

The idea of interactive testing can be generalized to many other problems as long as we can construct two parts of data that are (conditionally) independent. As an example, for alternative definitions of the zero-effect null, such as those involving conditional expectations or conditional quantiles, one can explore specific functions of $\{Y_i, X_i\}$ that are independent of treatment assignment under certain model assumptions, and then plug them into the \peiZero framework. Importantly, our interactive procedures using the idea of ``masking and unmasking'' should be contrasted with data-splitting approaches. We remark that no test, interactive or otherwise, can be run twice from scratch (with a tweak made the second time to boost power) after the entire data has been examined; this amounts to $p$-hacking. We view our interactive tests as one step towards enabling experts (scientists and statisticians) to work together with statistical models and machine learning algorithms in order to discover scientific insights with rigorous guarantees.

\appendix

\section{Extensions of the \pei}
\label{apd:extension}

\subsection{FDR control of nonpositive effects} \label{sec:nonpositive-effect}
The \peiZero controls the false identifications of subjects with zero treatment effect, as defined in the null hypothesis~\eqref{eq:hp_twoside} and its corresponding footnote. In this section, we develop a modification to additionally control the error of falsely identifying subjects with nonpositive treatment effects, by defining a different null hypothesis.

\subsubsection{Problem setup} We define the null hypothesis for subject $i$ as nonpositive effect:
\begin{align} \label{eq:hp_oneside}
    H_{0i}^\text{nonpositive}: (Y_i^T \mid X_i) \preceq (Y_i^C \mid X_i), 
\end{align}
or equivalently, $H_{0i}^\text{nonpositive}: (Y_i \mid A_i = 1, X_i) \preceq (Y_i \mid A_i = 0, X_i).$ As before, our algorithm applies to two alternative definitions of the null hypothesis. In the context of treating the potential outcomes and covariates as fixed, the null hypothesis is
\begin{align} \label{eq:hp_oneside_fixed}
    H_{0i}^\text{nonpositive}: Y_i^T \leq Y_i^C,
\end{align}
and in the hybrid version where the potential outcomes are random with joint distribution $(Y_i^T, Y_i^C) \mid X_i \sim P_i$, the null posits
\begin{align} \label{eq:hp_oneside_hybrid}
    H_{0i}^\text{nonpositive}: Y_i^T \leq Y_i^C \text{ almost surely-} P_i.
\end{align}
Note that the nonpositive-effect null is less strict than the zero-effect null. Thus, an algorithm with FDR control for $H_{0i}^\text{nonpositive}$ must have valid FDR control for $H_{0i}^\text{zero}$, but the reverse needs not be true. Indeed, we observe in numerical experiments (Figure~\ref{fig:FDR_neg_two}) that the \peiZero does not control FDR for the nonpositive-effect null. This section presents a variant of \peiZero that controls false identifications of nonpositive effects, possibly more practical when interpreting the identified subjects. For example, when controlling FDR for the nonpositive-effect null at level $\alpha = 0.2$, we are able to claim that the expected proportion of identified subjects with positive effects is no less than $80\%$.

\subsubsection{An interactive procedure in randomized experiments} Recall that the FDR control of the \peiZero is based on property~\eqref{eq:sign_property_zero} that when the null hypothesis is true for subject~$i$, we have $\mathbb{P}\left(\widehat \Delta_i \mid \{Y_j, X_j, E_j\}_{j=1}^n\right) \leq 1/2$, but this statement no longer holds when the null hypothesis is defined as $H_{0i}^\text{nonpositive}$ in~\eqref{eq:hp_oneside}. Fortunately, this issue can be fixed by making the condition in~\eqref{eq:sign_property_zero} coarser and removing the outcomes:
\begin{align*} 
    \mathbb{P}\left(\widehat \Delta_i \mid \{X_j\}_{j=1}^n\right) \leq 1/2,
\end{align*}
which is reflected in the interactive procedure as reducing the available information for selecting subject~$i_t^*$ (at step~8 of Algorithm~\ref{alg:pei}) --- we additionally mask (hide) the outcome~$Y_i$ of the candidate subjects~$i \in \mathcal{R}_{t-1}(\mathcal{I})$ when implementing the \pei on set $\mathcal{I}$. We call the resulting interactive algorithm \peiNeg, as it masks the outcomes. 

Specifically, the \peiNeg modifies \peiZero where we define the available information to select subjects when implementing Algorithm~\ref{alg:pei} on set $\mathcal{I}$ as
\begin{align} \label{eq:sigma_LHO}
        \mathcal{F}_{t-1}^{-Y}(\mathcal{I}) = \sigma\left(\{X_i\}_{i \in \mathcal{R}_{t-1}(\mathcal{I})}, \left\{Y_j, A_j,  X_j\right\}_{j \notin \mathcal{R}_{t-1}(\mathcal{I})}, \sum_{i \in \mathcal{R}_{t-1}(\mathcal{I})} \one\{\widehat \Delta_i > 0\}\right).
\end{align}
To calculate $\widehat \Delta_i$ at $t = 0$ when $Y_i$ for all $i \in \mathcal{I}$ are masked, let $\widehat m^{-\mathcal{I}}(X_i)$ be an estimator of $\me(Y_i \mid X_i)$ that is learned using data from non-candidate subjects $\{Y_j, X_j\}_{j \notin \mathcal{I}}$, and let the residuals be $E_i^{-\mathcal{I}} := Y_i - \widehat m^{-\mathcal{I}}(X_i)$.  Define $\Delta_i^{-\mathcal{I}} := 4(A_i - 1/2)\cdot E_i^{-\mathcal{I}}$, and similar to property~\eqref{eq:sign_property_zero} for the zero-effect null, we have
\begin{align} \label{eq:sign_property_neg}
    \mathbb{P}\left(\widehat \Delta_i^{-\mathcal{I}} > 0 \mid \{X_j\}_{j \in \mathcal{I}} \cup \{Y_j, X_j, E_j^{-\mathcal{I}}\}_{j \notin \mathcal{I}} \right) \leq 1/2,
\end{align}
under $H_{0i}^\text{nonpositive}$, leading to valid FDR control for nonpositive effects. Overall, the \peiNeg follows Algorithm~\ref{alg:para}, except the estimated treatment effect $\widehat \Delta_i$ replaced by $\widehat \Delta_i^{-\mathcal{I}}$, and the available information for selection $\mathcal{F}_{t-1}(\mathcal{I})$ replaced by~$\mathcal{F}_{t-1}^{-Y}(\mathcal{I})$. See Appendix~\ref{apd:MaY-regular} for the proof of FDR control.

\begin{theorem} \label{thm:LHO-ITE}
Under assumption~\eqref{eq:randomize1} and~\eqref{eq:randomize2} of randomized experiments, the \peiNeg has a valid FDR control at level $\alpha$ for the nonpositive-effect null hypothesis under any of definitions~\eqref{eq:hp_oneside},~\eqref{eq:hp_oneside_fixed} or~\eqref{eq:hp_oneside_hybrid}. For the last definition, FDR control also holds conditional on the covariates and potential outcomes.
\end{theorem}


Similar to Algorithm~\ref{alg:select_RF} for the \peiZero, we can design an automated algorithm for the \peiNeg to select a subject in step~8 of  Algorithm~\ref{alg:pei}, but the available information $\mathcal{F}_{t-1}^{-Y}(\mathcal{I})$ no longer includes the outcomes of candidate subjects. One naive strategy is to follow Algorithm~\ref{alg:select_RF} in the main paper, which is designed for the \peiZero, with the outcomes removed from the predictors; however, it appears to result in less accurate prediction of the effect signs, and in turn rather low power (numerial results are in the next paragraph). Here, we take a different approach by predicting the treatment effect instead of their signs, because the treatment effect might be better predicted as a function of the covariates (without outcomes) than a binary sign, especially when the treatment effect is indeed a smooth and simple function of the covariates. Specifically, we first estimate the treatment effect for the non-candidate subjects $j \notin \mathcal{R}_{t-1}(\mathcal{I})$ using a well-studied doubly-robust estimator (see \citet{kennedy2020optimal} and references therein):
\begin{align} \label{eq:est_effect_dr}
    \Delta_j^{\text{DR}} = 4(A_j - 1/2)\cdot(Y_j - \widehat \mu_A(X_j)) + \widehat \mu_1(X_j) - \widehat \mu_0(X_j),
\end{align}
where $(\widehat \mu_0, \widehat \mu_1)$ are random forests trained to predict the outcomes for the control and treated group, respectively. Using the provided covariates $X_i$, we can predict~$\Delta_i^{\text{DR}}$ for the candidate subjects $i \in \mathcal{R}_{t-1}(\mathcal{I})$. The subject with the smallest prediction of $\Delta_i^{\text{DR}}$ is then excluded. This automated strategy is described in Algorithm~\ref{alg:select_LHO}.

\begin{algorithm}[H]
   \caption{An automated heuristic to select $i_t^*$ in the \peiNeg.}
   \label{alg:select_LHO}
\begin{algorithmic}
   \STATE {\bfseries Input:} Current rejection set $\mathcal{R}_{t-1}(\mathcal{I})$, and available information for selection $\mathcal{F}_{t-1}^{-Y}(\mathcal{I})$;
   \STATE {\bfseries Procedure:} 
   \STATE 1.~Estimate the treatment effect for non-candidate subjects $j \notin \mathcal{R}_{t-1}(\mathcal{I})$ as $\Delta_j^{\text{DR}}$ in~\eqref{eq:est_effect_dr};
   \STATE 2.~Train a random forest where the label is the estimated effect $\Delta_j^{\text{DR}}$ and the predictors are the covariates $X_j$, using non-candidate subjects $j \notin \mathcal{R}_{t-1}(\mathcal{I})$;
   \STATE 3.~Predict $\Delta_i^{\text{DR}}$ for candidate subjects $i \in \mathcal{R}_{t-1}(\mathcal{I})$ via the above random forest, denoted as $\widehat \Delta_i^{\text{DR}}$;
   \STATE 4.~Select $i_t^*$ as $\argmin\{\widehat \Delta_i^{\text{DR}}: i \in \mathcal{R}_{t-1}(\mathcal{I})\}$.
\end{algorithmic}
\end{algorithm}

To summarize, we have presented two types of strategy for selecting subjects: the \peiZero chooses the one with the smallest predicted probability of a positive $\widehat \Delta_i$ (see Algorithm~\ref{alg:select_RF} in the main paper), which we denote here as the \textit{min-prob strategy}; and the \peiNeg chooses the one with the smallest prediction of estimated effect $\Delta_j^{\text{DR}}$ (see Algorithm~\ref{alg:select_LHO}), which we denote here as the \textit{min-effect strategy}. Note that the proposed interactive algorithm can use arbitrary strategy as long as the available information for selection is restricted. That is, the \peiZero can use the same min-effect strategy, and the \peiNeg can use the min-prob strategy (after removing the outcomes from the predictors, which we elaborate in the next paragraph). However, we observe in numerical experiments that both interactive procedures have higher power when using their original strategies, respectively (see Figure~\ref{fig:CATE_RF}).

\begin{figure}[h!]
\centering
    \includegraphics[width=0.4\linewidth]{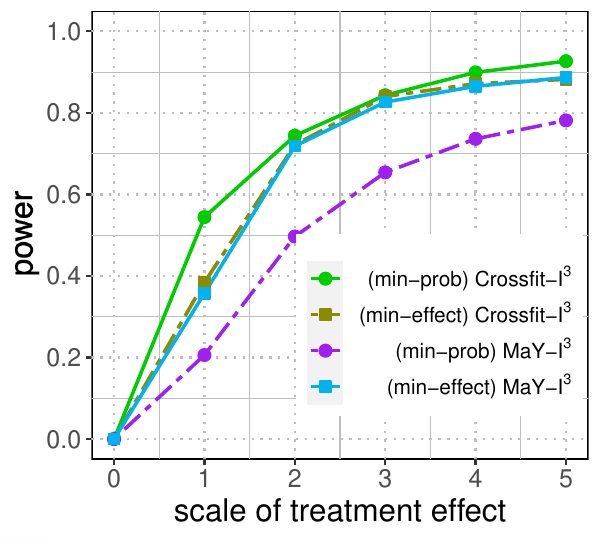}
    \caption{Power of the \peiZero and \peiNeg with two strategies to select subjects: the min-prob strategy and the min-effect strategy, under the treatment effect defined in~\eqref{eq:bias-sparse} of the main paper with the scale $S_\Delta$ varies in $\{0,1,2,3,4,5\}$. The \peiZero tends to have higher power when using the min-prob strategy, and the \peiNeg tends to have higher power when using the min-effect strategy.}
    \label{fig:CATE_RF}
\end{figure}

\subsubsection{Numerical experiments} Before details of the experiment results, we first describe the min-prob strategy for the \peiNeg, where the available information $\mathcal{F}_{t-1}^{-Y}(\mathcal{I})$ does not include the outcomes for candidate subjects. Similar to the min-prob strategy in Algorithm~\ref{alg:select_RF} of the main paper, we hope to use the outcome $Y_i$ and residual $E_i = Y_i - \widehat m(X_i)$ as predictors, and predict the sign of treatment effect for candidate subjects $i \in \mathcal{R}_{t-1}(\mathcal{I})$, but $Y_i$ and $E_i$ for the candidate subjects are not available in~$\mathcal{F}_{t-1}^{-Y}(\mathcal{I})$. Thus, we propose algorithm~\ref{alg:select_LHO_RF}, where we first estimate $Y_i$ and $E_i$ using the covariates (see step~1-2); and step 3-5 are similar to Algorithm~\ref{alg:select_RF}, which obtain the probability of having a positive treatment effect. 

\begin{algorithm}[h]
   \caption{The min-prob strategy to select $i_t^*$ in the \peiNeg.}
   \label{alg:select_LHO_RF}
\begin{algorithmic}
   \STATE {\bfseries Input:} Current rejection set $\mathcal{R}_{t-1}(\mathcal{I})$, and available information for selection $\mathcal{F}_{t-1}^{-Y}(\mathcal{I})$;
   \STATE {\bfseries Procedure:} 
   \STATE 1.~Predict the outcome $Y_k$ of each subject $k \in [n]$ by covariates, denoted as $\widehat Y^{-\mathcal{I}}(X_k)$, where $\widehat Y^{-\mathcal{I}}$ is learned using non-candidate subjects $j \notin \mathcal{R}_{t-1}(\mathcal{I})$;
   \STATE 2.~Predict the residual $E_k = Y_k - \widehat m(X_K)$ of each subject $k \in [n]$ by covariates, denoted as $\widehat E^{-\mathcal{I}}(X_k)$, where $\widehat E^{-\mathcal{I}}$ is learned using non-candidate subjects $j \notin \mathcal{R}_{t-1}(\mathcal{I})$;
   \STATE 3.~Train a random forest classifier where the label is $\text{sign} (\widehat \Delta_j)$ and the predictors are $\left(\widehat Y^{-\mathcal{I}}(X_j), X_j, \widehat E^{-\mathcal{I}}(X_j)\right)$, using non-candidate subjects $j \notin \mathcal{R}_{t-1}(\mathcal{I})$;
   \STATE 4.~Predict the probability of $\widehat \Delta_i$ being positive as~$\widehat p(i,t)$ for candidate subjects $i \in \mathcal{R}_{t-1}(\mathcal{I})$; 
   \STATE 5.~Select $i_t^* = \argmin\{\widehat p(i,t): i \in \mathcal{R}_{t-1}(\mathcal{I})\}$.
\end{algorithmic}
\end{algorithm}

The \peiZero has higher power when using the min-prob strategy than the min-effect strategy because the former additionally uses the outcome as a predictor. For the \peiNeg, the min-effect strategy leads to higher power because the estimated treatment effect $\Delta_j^{\text{DR}}$ in~\eqref{eq:est_effect_dr} can provide reliable evidence of which subjects have a positive effect. If using the min-prob strategy, it could be harder to learn an accurate prediction by Algorithm~\ref{alg:select_LHO_RF} where two of the predictors $\widehat Y^{-\mathcal{I}}(X_j)$ and $ \widehat E^{-\mathcal{I}}(X_j)$ are obtained by estimation, increasing the complexity in modeling. Therefore, we present the \peiZero and \peiNeg with the min-prob and min-effect strategies, respectively, as preferred in numerical experiments. Nonetheless, we remark that our proposed interactive frameworks for the \peiZero and \peiNeg allow arbitrary strategies to select subjects, and an analyst can design her own strategy based on her domain knowledge.


\begin{figure}[h!]
\centering
    \begin{subfigure}[t]{0.32\textwidth}
        \centering
        \includegraphics[width=1\linewidth]{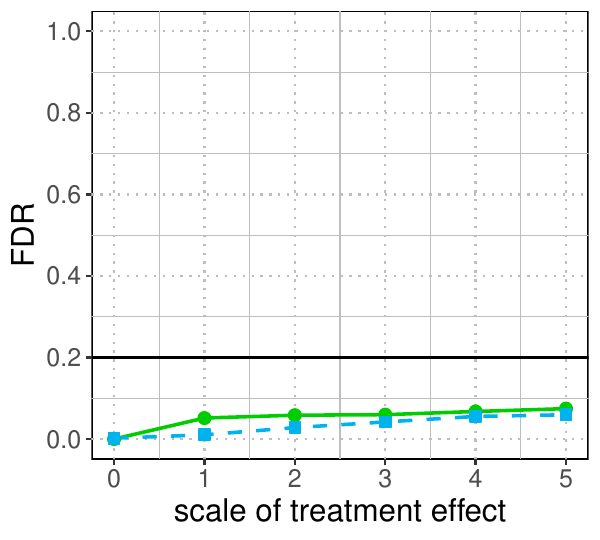}
        \subcaption{FDR for the zero-effect null~\eqref{eq:hp_twoside}.}
        \label{fig:FDR_zero_two}
    \end{subfigure}
    \hfill
    \begin{subfigure}[t]{0.32\textwidth}
        \centering
        \includegraphics[width=1\linewidth]{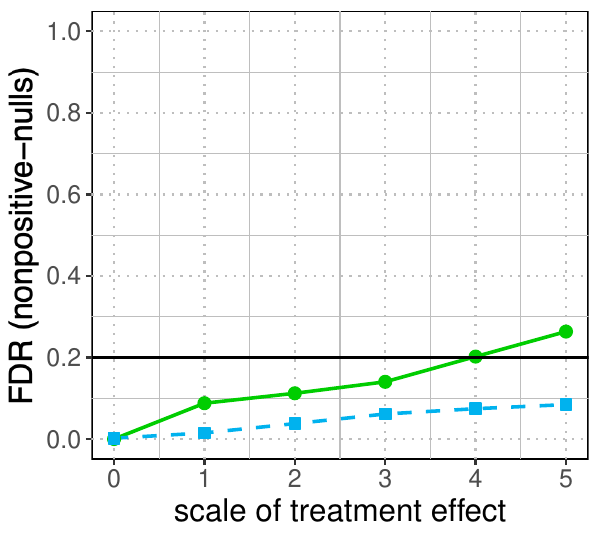}
        \subcaption{FDR for the nonpositive-effect null~\eqref{eq:hp_oneside}.}
        \label{fig:FDR_neg_two}
    \end{subfigure}
    \hfill
    \begin{subfigure}[t]{0.32\textwidth}
    \centering
        \includegraphics[width=1\linewidth]{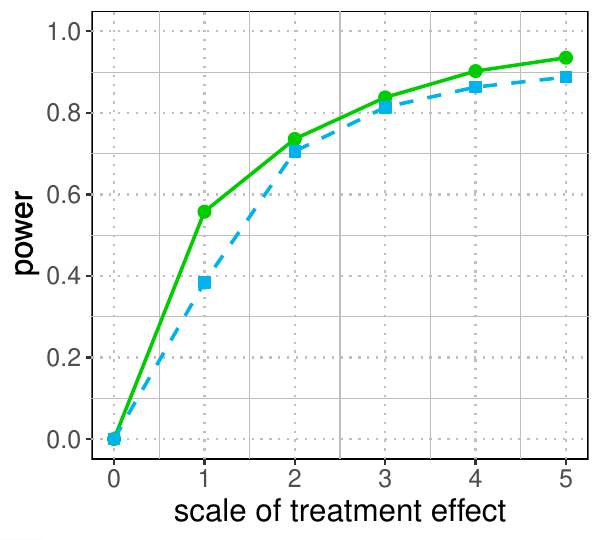}
        \caption{Power of identifying subjects with positive effects.}
        \label{fig:power_pos_two}
    \end{subfigure}
    
    \vfill
    \begin{subfigure}[t]{1\textwidth}
        \centering
        \includegraphics[width=0.7\linewidth]{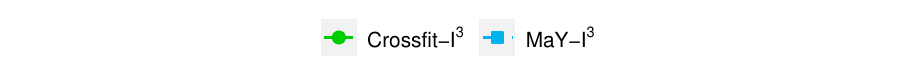}
    \end{subfigure}
    
    \caption{Performance of two interactive methods, \peiZero and \peiNeg, with the treatment effect specified as model~\eqref{eq:bias-sparse} and the scale $S_\Delta$ varying in $\{0,1,2,3,4,5\}$. The \peiNeg controls FDR for a more relaxed null (nonpositive effects) than the \peiZero, while the \peiZero has slightly higher power than the \peiNeg.}
    \label{fig:bias-sparse_two}
\end{figure}

We compare the \peiZero and \peiNeg using the same experiment as Section~\ref{sec:sim_setup}.
In terms of the error control, both the \peiZero and \peiNeg control FDR for the zero-effect null at the target level (Figure~\ref{fig:FDR_zero_two}). 
When the null is defined as having a nonpositive effect, the \peiZero can violate the error control (Figure~\ref{fig:FDR_neg_two}), whereas the \peiNeg preserves valid FDR control. In terms of the power, the \peiZero has slightly higher power since the analyst can select subjects using information defined by~$\mathcal{F}_{t-1}(\mathcal{I})$ in~\eqref{eq:sigma}, which is richer compared to $\mathcal{F}_{t-1}^{-Y}(\mathcal{I})$ in~\eqref{eq:sigma_LHO} for the \peiNeg.

To summarize, the error control of the \peiNeg is more strict than the \peiZero, controlling false identifications of both zero effects and negative effects, while its power is slightly lower. We recommend the \peiZero if one only concerns the error of falsely identifying subjects with zero effect. Alternatively, we recommend the \peiNeg when it is desired to control the error of falsely identifying subjects with nonpositive effects.

\subsubsection{Extensions to observational studies}
\label{sec:extend_obs}

The extension from randomized experiments to observational studies for the nonpositive-effect nulls is similar to that for the zero-efffect nulls from \peiZero to \peiZeroUnknowP~--- we estimate the propensity scores~$\pi_i$ using the revealed data (detailed steps described in Section~\ref{sec:propensity_unknown}). We call the resulting algorithm~\peiNegUnknowP.

FDR control for \peiNegUnknowP holds even when the bounds of true propensity scores reach $0$ or $1$:
\begin{itemize}
    \item[(iv)] the propensity scores are bounded between $0$ and $1$:
\begin{align} \label{eq:prop_bound2}
        0 \leq \pi_{\min} \leq \pi_i \leq \pi_{\max} \leq 1 \text{ for all } i \in [n].
\end{align}
\end{itemize}
which is less stringent than assumption~(\ref{eq:prop_bound}) for \peiZeroUnknowP. 
The error control of the \peiNegUnknowP is doubly robust: when the propensity score estimation is poor, FDR would still be close to target level when the expected outcomes are well-estimated. To characterize the error of expected outcome estimation, we define a ``centered'' CDF $\Phi$ as
\begin{align*}
    \Phi_i(\epsilon) := \mathbb{P}(Y_i - \me(Y_i \mid X_i) \leq \epsilon \mid \{X_i\}_{i=1}^n)
\end{align*}
with ``upper'' and ``lower'' bounds defined as $\Phi_{\max}(\epsilon) := \max_{i \in [n]} \Phi_i(\epsilon)$ and $\Phi_{\min}(\epsilon) := \min_{i \in [n]} \Phi_i(\epsilon)$. If the estimation error $\epsilon$ is small and the outcome distribution is symmetric and continuous, the centered CDF is close to $1/2$, leading to less FDR inflation as we describe later. 
We define several estimation errors when performing the \pei on set $\mathcal{I}$ as follows. 
\begin{itemize}
\item Let the error of propensity score estimation be
$\epsilon_n^\pi(\mathcal{I}) := \max_{i \in \mathcal{I}} |\pi_i - \widehat \pi_i(\mathcal{I})|$, where $\widehat \pi_i(\mathcal{I})$ is the estimated propensity score using data information in~$\mathcal{F}_{0}^{-Y}(\mathcal{I})$.

\item Let the error of expected outcome estimation be
\begin{align*}
    \epsilon_n^Y(\mathcal{I}) := \max_{ i \in \mathcal{I}}\{|\me(Y_i \mid X_i) - \widehat m^{-\mathcal{I}}(X_i)|\},
\end{align*}
where $\widehat m^{-\mathcal{I}}(X_i))$ is the estimated expected outcome learned using data information in~$\mathcal{F}_{0}^{-Y}(\mathcal{I})$, which includes the complete data of non-candidate subjects~$j \in [n] \backslash \mathcal{I}$. Recall that the sign of $\Delta_{\widehat m}(X_i)$ in \peiNeg depends on the sign of $Y_i - \widehat m^{-\mathcal{I}}(X_i)$, whose probability of being positive deviates from 1/2 at most by $\max\left\{\Phi_{\max}\left[\epsilon_n^Y(\mathcal{I})\right], 1 - \Phi_{\min}\left[-\epsilon_n^Y(\mathcal{I})\right]\right\}$.

\item For a candidate subject $i\in \mathcal{I}$ such that the zero-effect null hypothesis~\eqref{eq:hp_twoside} is true, denote the probability of $\Delta_{\widehat m}(X_i)$ being positive as 
\begin{align}
    q_{i}(\mathcal{I}) := \mathbb{P}\left((A_i - 1/2) \cdot (Y_i - \widehat m(X_i)) > 0 \mid \mathcal{F}_{0}^{-Y}(\mathcal{I})\right),
\end{align}
which is upper bounded by
\begin{align} \label{eq:def_q}
    q_{i}(\mathcal{I}) 
    \leq \min\{\max\{\pi_{\max},1-\pi_{\min}\}, \max\left\{\Phi_{\max}\left[\epsilon_n^Y(\mathcal{I})\right], 1 - \Phi_{\min}\left[-\epsilon_n^Y(\mathcal{I})\right]\right\}\} =: q_{\max}(\mathcal{I})
\end{align}
which is close to 1/2 (the ideal case) when \textit{either} the true propensity score is close to $1/2$ \textit{or} the error of the expected outcome estimation~$\epsilon_n^Y(\mathcal{I})$ is small.

\item The estimation error of $q_{i}(\mathcal{I})$ is upper bounded as 
\begin{align}
     \epsilon^q_n(\mathcal{I}) := \epsilon_n^{\pi}(\mathcal{I}) - \max\left\{0, \max\{\pi_{\max},1-\pi_{\min}\} - \max\left\{\Phi_{\max}\left[\epsilon_n^Y(\mathcal{I})\right], 1 - \Phi_{\min}\left[-\epsilon_n^Y(\mathcal{I})\right]\right\}\right\}\}. 
\end{align}
\end{itemize}
Similar error terms can be derived for the procedure on set $\mathcal{II}$.

\begin{theorem} \label{thm:MaY-unknown}
The FDR control of \peiNegUnknowP is upper bounded:
\begin{align*}
     \mathbb{E}\left[\mathrm{FDP}_{\widehat t}^{\widehat \pi} \right] \leq \alpha \left\{1 + \mathbb{E}_{\mathcal{F}_0(\mathcal{I})}\left[\epsilon_n^{q}(\mathcal{I}) \left(\frac{4}{ q_{\max}(\mathcal{I}) (1 - q_{\max}(\mathcal{I}))}\right)\right] +  \mathbb{E}_{\mathcal{F}_0(\mathcal{II})}\left[\epsilon_n^{q}(\mathcal{II}) \left(\frac{4}{ q_{\max}(\mathcal{II}) (1 - q_{\max}(\mathcal{II}))}\right)\right]\right\},
\end{align*}
when $\epsilon_n^{q}(\mathcal{I}) \leq q_{\max}(\mathcal{I}) / 2$ and $\epsilon_n^{q}(\mathcal{II}) \leq q_{\max}(\mathcal{II}) / 2$,
in an observational setting satisfying assumptions~\eqref{eq:randomized3}, ~\eqref{eq:prop_bound2}, and ~\eqref{eq:randomize2} for the zero-effect null~\eqref{eq:hp_twoside}. 
\end{theorem}

\begin{corollary}[Doubly-robust FDR control]
As sample size $n$ goes to infinity, the \peiNegUnknowP has asymptotic FDR control for the zero-effect null~\eqref{eq:hp_twoside} when either
\begin{enumerate}
    \item (a) the propensity score estimation is consistent in the sense that $\mathbb{E}_{\mathcal{F}_0(\mathcal{I})}\left[\epsilon_n^{\pi}(\mathcal{I})\right] \to 0$; and\\ (b)~${\mathbb{E}_{\mathcal{F}_0(\mathcal{II})}\left[\epsilon_n^{\pi}(\mathcal{II})\right] \to 0}$, and the true propensity scores are bounded away from 0 and 1; or
    \item (a)~the expected outcome estimation is consistent in the sense that $\epsilon_n^Y(\mathcal{I}) \to 0$ almost surely over the conditional distribution given $\mathcal{F}_0(\mathcal{I})$; and (b)~same for $\epsilon_n^Y(\mathcal{II})$; and (c)~the difference between bounds on true propensity scores and $1/2$ is larger than its estimation error: $\max\{\pi_{\max},1-\pi_{\min}\} - 1/2 \geq \epsilon_n^{\pi}(\mathcal{I}) \vee \epsilon_n^{\pi}(\mathcal{II})$ almost surely; and (d)~the outcome distribution is symmetric.
\end{enumerate}

\end{corollary}

Note that the above theorem states the FDR guarantee for the zero-effect null~\eqref{eq:hp_twoside}, and the error control for the nonpositive-effect null~\eqref{eq:hp_oneside} is discussed in Appendix~\ref{apd:MaY-unknown}, whose condition to ensure asymptotic error control is similar, but practically it could be hard to have consistent estimation for the expected outcomes.
Also, the above theorem provides an upper bound of the FDR in terms of the maximum estimation error over all subjects, while in practice we expect the FDR to be close to the target level when the the estimation error is small for most subjects. 

\subsection{Paired samples} \label{sec:paired}

\subsubsection{Problem setup} Our discussion has focused on the case where samples are not paired, and the proposed algorithms can be extended to the paired-sample setting. Suppose there are $n$ pairs of subjects. Let outcomes of subjects in the $i$-th pair be~$Y_{ij}$, treatment assignments be indicators~$A_{ij}$, covariates be~$X_{ij}$ for $j = 1,2$ and $i \in [n]$. We deal with randomized experiments without interference, and assume that
\begin{enumerate}
    \item[(i)] conditional on covariates, the treatment assignments are independent coin flips:
    \allowdisplaybreaks
    \begin{align*}
        & \mathbb{P}[(A_{11}, \ldots, A_{n1}) = (a_1, \ldots, a_n) \mid X_1, \ldots, X_n] = \prod_{i=1}^n \mathbb{P}(A_i = a_i) = (1/2)^n, \text{ and } {}\\
        &A_{i1} + A_{i2} = 1 \text{ for all } i \in [n].
    \end{align*}
    
    \item[(ii)] conditional on covariates, the outcome of one subject $Y_{i_1,j_1}$ is independent of the treatment assignment of another subject~$A_{i_2,j_2}$ conditional on~$A_{i_1,j_1}$, for any $(i_1,j_1) \neq (i_2,j_2)$.
\end{enumerate}
As before, we can develop interactive algorithms for two types of error control (only the definitions when treating the potential outcomes as random variables are presented, but the FDR control still applies to all versions of the null):
\begin{align}
    H_{0i}^\text{(zero, paired)}:~& (Y_{ij}^T \mid X_{ij}) \overset{d}{=} (Y_{ij}^C \mid X_{ij}) \text{ for both } j = 1,2; \label{hp:paired-zero}{}\\
    H_{0i}^\text{(nonpositive, paired)}:~& (Y_{ij}^T \mid X_{ij}) \preceq (Y_{ij}^C \mid X_{ij}) \text{ for both } j = 1,2. \label{hp:paired-nonpos}
\end{align}
Next, we present the extensions of \peiZero for FDR control of zero effect and \peiNeg for FDR control of nonpositive effect for the paired-sample setting, by a few modifications in the effect estimator.

\subsubsection{Interactive algorithms for paired samples} With the pairing information, the treatment effect can be estimated without involving $\widehat m$ as in~\eqref{eq:est_effect}:
\begin{align} \label{eq:est_effect_paired}
    \widehat \Delta_i^{\text{paired}} := (A_{i1} - A_{i2})(Y_{i1} - Y_{i2}),
\end{align}
as used by \citet{rosenbaum2002covariance} and \citet{howard2019uniform}, among others. The above estimation satisfies the critical property to guarantee FDR control: for a null pair $i$ of two subjects with zero effects in~\eqref{hp:paired-zero}, we have
\begin{align} \label{eq:obs_pair}
    \mathbb{P}(\widehat \Delta_i^{\text{paired}} > 0 \mid \{Y_{j1}, Y_{j2}, X_{j1}, X_{j2}\}_{j=1}^n) \leq 1/2.
\end{align}
Thus, the \peiZero (Algorithm~\ref{alg:para}) with $\widehat \Delta_i$ replaced by $\widehat \Delta_i^{\text{paired}}$ has valid FDR control for the zero-effect null~\eqref{hp:paired-zero}, where the analyst excludes pairs using the available information, including $\{Y_{i1}, Y_{i2}, X_{i1}, X_{i2}\}$ for candidate subjects $i \in \mathcal{R}_{t-1}(\mathcal{I})$, and $\{Y_{j1}, Y_{j2}, A_{j1}, A_{j2}, X_{j1}, X_{j2}\}$ for non-candidate subjects $j \notin \mathcal{R}_{t-1}(\mathcal{I})$, and the sum $\sum_{i \in \mathcal{R > 
}_{t-1}(\mathcal{I})} \one\{\widehat \Delta_i^{\text{paired}}>0\}$ for FDR estimation. An automated strategy exclude pair $i_t^*$ (at step~8 of Algorithm~\ref{alg:pei}) under paired samples is the same as Algorithm~\ref{alg:select_RF}, except $\widehat \Delta_i$ being replaced by $\widehat \Delta_i^{\text{paired}}$.

Recall the nonpositive-effect null under paired samples:
\begin{align*}
    H_{0i}^\text{(nonpositive, paired)}: (Y_{ij} \mid A_{ij} = 1, X_{ij}) \preceq (Y_{ij} \mid A_{ij} = 0, X_{ij}) \text{ for both } j = 1,2,
\end{align*}
and we observe that 
\begin{align} \label{eq:obs_pair_MaY}
    \mathbb{P}(\widehat \Delta_i^{\text{paired}} > 0 \mid \{X_{j1}, X_{j2}\}_{j=1}^n) \leq 1/2,
\end{align}
where $\widehat \Delta_i^{\text{paired}}$ is defined in~\eqref{eq:est_effect_paired} of the main paper. Thus, the \peiNeg with $\widehat \Delta_i$ replaced by~$\widehat \Delta_i^{\text{paired}}$ has valid FDR control for the nonpositive-effect null, where the analyst progressively excludes pairs using the available information:
\begin{align*}
    \mathcal{F}_{t-1}^{-Y, \text{paired}} = \sigma\left(\{X_{i1}, X_{i2}\}_{i \in \mathcal{R}_{t-1}}, \{Y_{j1}, Y_{j2}, A_{j1}, A_{j2}, X_{j1}, X_{j2}\}_{j \notin \mathcal{R}_{t-1}}, \sum_{i \in \mathcal{R}_{t-1}(\mathcal{I})} \one\{\widehat \Delta_i^{\text{paired}} > 0\}\right).
\end{align*}
We can also implement an automated version of the \peiNeg where the selection of the excluded subject follows a similar procedure as Algorithm~\ref{alg:select_LHO}. The difference is that in step~1, we estimate the treatment effect for non-candidate subjects $j \notin \mathcal{R}(\mathcal{I})$ directly as $\widehat \Delta_i^{\text{paired}} \equiv (A_{i1} - A_{i2})(Y_{i1} - Y_{i2})$ instead of $\Delta_i^{\text{DR}}$ to avoid estimating outcomes in $(\widehat \mu_0, \widehat \mu_1)$.

\subsubsection{Numerical experiments}
We compare the power of the interactive procedures with and without the pairing information, using the same experiments as previous. When the subjects within each pair have the same covariate values, the power under paired samples is higher than treating them as unpaired (see Figure~\ref{fig:exact_pair}), because the noisy variation in the observed outcomes that results from the potential control outcomes can be removed by taking the difference in outcomes within each pair. 


The advantage of procedures under paired samples becomes less evident when the subjects within a pair do not match exactly. We simulate unmatched pairs by introducing a parameter $\epsilon$ such that for each pair $i$, the covariates of the two subjects within satisfy: $\mathbb{P}(X_{i1}(1) \neq X_{i2}(1)) = \epsilon, \mathbb{P}(X_{i1}(2) \neq X_{i2}(2)) = \epsilon, X_{i1}(3) = X_{i2}(3) + U(0, 2\epsilon),$ where $U(0, 2\epsilon)$ is uniformly distributed between $0$ and $2\epsilon$, and a larger $\epsilon$ leads to a larger degree of mismatch. As $\epsilon$ increases, the power of procedures using the pairing information decreases (see Figure~\ref{fig:noexact_pair}), because the estimated treatment effect $\widehat \Delta_i^{\text{paired}}$ becomes less accurate for the mismatching setting. 

\vspace{20pt}
\begin{figure}[H]
\centering
\hspace{1cm}
    \begin{subfigure}[t]{0.3\textwidth}
        \centering
        \includegraphics[width=1\linewidth]{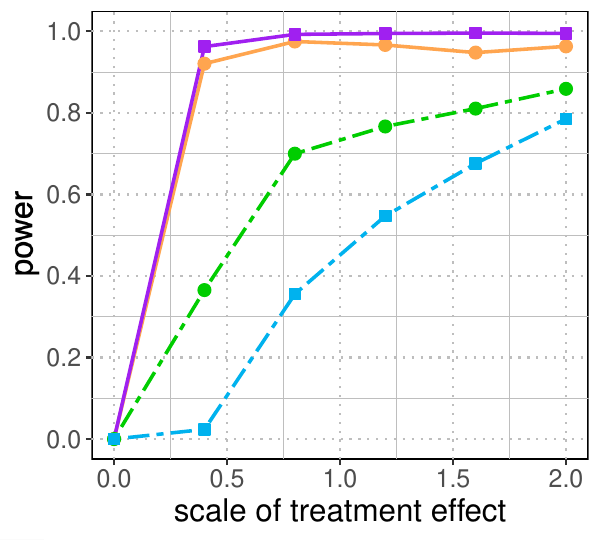}
        \subcaption{Exact pairs.}
        \label{fig:exact_pair}
    \end{subfigure}
    \hfill
    \begin{subfigure}[t]{0.3\textwidth}
        \centering
        \includegraphics[width=1\linewidth]{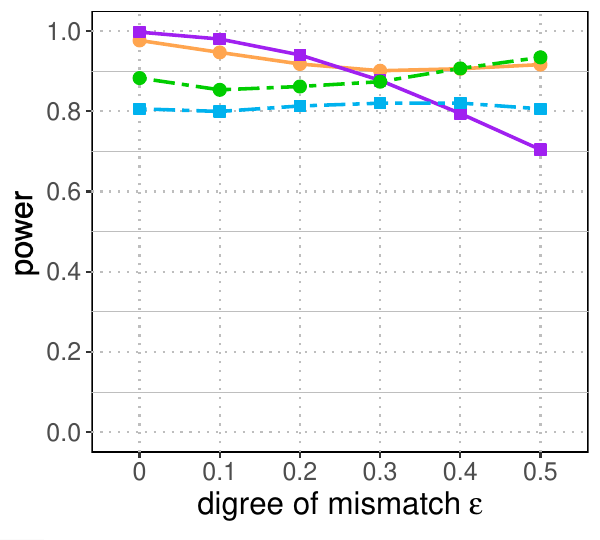}
        \subcaption{Subjects within the pair do not match exactly.}
        \label{fig:noexact_pair}
    \end{subfigure}
\hspace{1cm}

    \vfill
    \begin{subfigure}[t]{1\textwidth}
        \centering
        \includegraphics[width=0.7\linewidth]{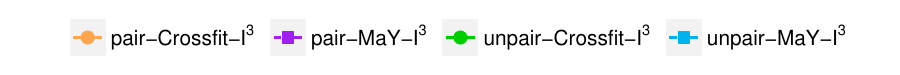}
    \end{subfigure}
    
    \caption{Power under paired samples with treatment effects specified by model~\eqref{eq:bias-sparse} when our proposed algorithms (\peiZero and \peiNeg) utilize the pairing information, which is higher than treating all subjects as unpaired. The advantage is less evident when the subjects within each pair are not exactly matched to have the same covariate values.}
    \label{fig:pair-bias}
\end{figure}
\vspace{20pt}

\begin{figure}[h!]
\centering

    \begin{subfigure}[t]{0.32\textwidth}
        \centering
        \includegraphics[width=1\linewidth]{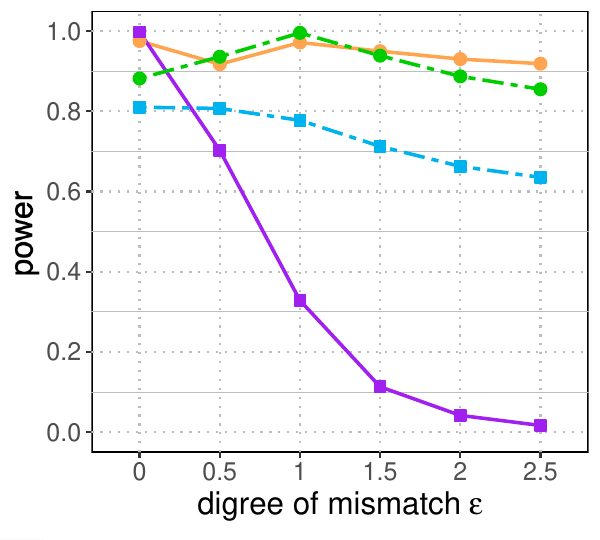}
    \end{subfigure}
    
    \vfill
    \begin{subfigure}[t]{1\textwidth}
        \centering
        \includegraphics[width=0.7\linewidth]{legend_pair.pdf}
    \end{subfigure}
    
    \caption{Power of identifying subjects with positive effects of the proposed algorithms (\peiZero and \peiNeg) with or without pairing information, when the scale of treatment effect is fixed at $2$ and the degree of mismatch~$\epsilon$ varies. The power of algorithms without pairing information first increase and then decrease as $\epsilon$ becomes larger. }
     \label{fig:noexact_pair_large}
\end{figure}

To further investigate the power decrease, we extend the definition of mismatch for $\epsilon \in (0,1)$ to a larger $\epsilon$: $ \mathbb{P}(X_{i1}(1) \neq X_{i2}(1)) = \min\{\epsilon, 1\}$ and $\mathbb{P}(X_{i1}(2) \neq X_{i2}(2)) = \min\{\epsilon, 1\} \text{ and } X_{i1}(3) = X_{i2}(3) + U(0, 2\epsilon)$,
where $U(0, 2\epsilon)$ is uniformly distributed between $0$ and $2\epsilon$, and a larger $\epsilon$ leads to a larger degree of mismatch. As $\epsilon$ increases, the power under the unpaired samples first increases (see Figure~\ref{fig:noexact_pair_large}). It is because the treatment effect is positive when $X_i(3) > 1$, which only takes $15\%$ proportion if without mismatching; thus, the pattern of treatment effect is not easy to learn. In contrast, when there is a positive shift on $X_i(3)$ as designed in the mismatching setting above, more subjects have positive effects so that the algorithm can more easily learn the effect pattern and hence increase the power. The power can slightly decrease when the degree of mismatch is too large ($\epsilon > 1$), because there are fewer subjects without treatment effect, also affecting the estimation of treatment effect.


\section{Proof of error controls}
The proofs are based on an optional stopping argument, as a variant of the ones presented in \citet{lei2018adapt}, \citet{lei2017star}, \citet{li2017accumulation} and \citet{barber2015controlling}.
\begin{lemma}[Lemma~2 of \citet{lei2018adapt}] \label{lm:adapt}
Suppose that, conditionally on the $\sigma$-field~$\mathcal{G}_{-1}$, $b_{1}, \ldots, b_{n}$ are independent Bernoulli random variables with
\begin{align*}
    \mathbb{P}(b_{i} = 1 \mid \mathcal{G}_{-1}) = \rho_{i} \geq \rho > 0, \text{ almost surely.}
\end{align*}
Let $(\mathcal{G}_t)_{t=0}^\infty$ be a filtration with $\mathcal{G}_0 \subset \mathcal{G}_1 \subset \ldots$ and suppose that $[n] \supseteq \mathcal{C}_0 \supseteq \mathcal{C}_1 \supseteq \ldots$, with each subset $\mathcal{C}_{t+1}$ measurable with respect to $\mathcal{G}_t$. If we have
\begin{align} \label{eq:fil_G}
    \mathcal{G}_t = \sigma\left(\mathcal{G}_{-1}, \mathcal{C}_{t}, (b_{i})_{i\notin \mathcal{C}_{t}}, \sum_{i \in \mathcal{C}_{t}} b_i \right),
\end{align}
and $\tau$ is an almost-surely finite stopping time with respect to the filtration $(\mathcal{G}_t)_{t \geq 0}$, then
\begin{align*}
    \mathbb{E}\left[\frac{1 + |\mathcal{C}_\tau|}{1 + \sum_{i \in \mathcal{C}_\tau}b_{i}} \Bigg| \mathcal{G}_{-1}\right] \leq \rho^{-1}.
\end{align*}
\end{lemma}

\subsection{Proof of theorem~\ref{thm:naive-ITE}} \label{apd:proof-naive}
\begin{proof}
We show that the \pei controls FDR by Lemma~\ref{lm:adapt}, where 
\begin{align*}
    b_i := \one\{(A_i - 1/2) \cdot E_i \leq 0\} \text{ and }
    \mathcal{G}_{-1} := \sigma\left(\{Y_j, X_j\}_{j=1}^n\right)
     \text{ and }  \mathcal{C}_t := \mathcal{R}_{t} \cap \mathcal{H}_0, 
\end{align*}
for $t = 0,1,\ldots$. The assumptions in Lemma~\ref{lm:adapt} are satisfied: (a)~$\mathbb{P}(b_i = 1 \mid \mathcal{G}_{-1}) \geq 1/2$ for subjects  with zero effect $i \in \mathcal{H}_0$:
\begin{align*}
    &\mathbb{P}\left((A_i - 1/2) \cdot E_i \leq 0 \mid \mathcal{G}_{-1}\right){}\\
    =~&  \mathbb{P}(A_i = 1) \one(E_i \leq 0 \mid \mathcal{G}_{-1}) + \mathbb{P}(A_i = 0)\one(E_i \geq 0 \mid \mathcal{G}_{-1}),{}\\
    &\text{ because } A_i \text{ is independent of } \mathcal{G}_{-1} {}\\
    =~& 1/2\left[\one(E_i \leq 0 \mid \mathcal{G}_{-1}) + \one(E_i \geq 0 \mid \mathcal{G}_{-1}))\right]
    \geq 1/2;
\end{align*}
and (b)~the filtration in our algorithm satisfies $\mathcal{F}_{t} \subseteq \mathcal{G}_{t}$, so the time of stopping the algorithm $\widehat t := \min\{t:\widehat{\mathrm{FDR}}(\mathcal{R}_t) \leq \alpha\}$ is a stopping time with respect to $\mathcal{G}_t$; and~(c)  $\mathcal{C}_{t+1}$ is measurable with respect to $\mathcal{G}_t$. Thus, by Lemma~\ref{lm:adapt},
expectation, 
we have
\begin{align*}
    \mathbb{E}\left[\frac{1 + |\mathcal{R}_{\widehat t} \cap \mathcal{H}_0|}{1 + |\mathcal{R}_{\widehat t}^- \cap \mathcal{H}_0|}~\Big|~\mathcal{G}_{-1}\right] \leq 2,
\end{align*}
By definition, the FDR conditional on $\mathcal{G}_{-1}$ at the stopping time $\widehat t$ is
\begin{align*}
    &\mathbb{E}\left[\frac{|\mathcal{R}_{\widehat t}^+ \cap \mathcal{H}_0|}{\max\{|\mathcal{R}_{\widehat t}^+| ,1\}}~\Big|~\mathcal{G}_{-1}\right]
    = \mathbb{E}\left[\frac{1 + |\mathcal{R}_{\widehat t}^- \cap \mathcal{H}_0|}{\max\{|\mathcal{R}_{\widehat t}^+| ,1\}} \cdot \frac{|\mathcal{R}_{\widehat t}^+ \cap \mathcal{H}_0|}{1 + |\mathcal{R}_{\widehat t}^-\cap \mathcal{H}_0|}~\Big|~\mathcal{G}_{-1}  \right]{}\\
    \leq~& \mathbb{E}\left[ \widehat{\mathrm{FDR}}(\mathcal{R}_{\widehat t}) \cdot \frac{|\mathcal{R}_{\widehat t}^+ \cap \mathcal{H}_0|}{1 + |\mathcal{R}_{\widehat t}^-\cap \mathcal{H}_0|}~\Big|~\mathcal{G}_{-1} \right] \leq \alpha   \mathbb{E}\left[\frac{|\mathcal{R}_{\widehat t}^+ \cap \mathcal{H}_0|}{1 + |\mathcal{R}_{\widehat t}^-\cap \mathcal{H}_0|}~\Big|~\mathcal{G}_{-1}\right] {}\\
    =~& \alpha   \mathbb{E}\left[\frac{1 + |\mathcal{R}_{\widehat t} \cap \mathcal{H}_0|}{1 + |\mathcal{R}_{\widehat t}^- \cap \mathcal{H}_0|} - 1~\Big|~\mathcal{G}_{-1} \right] \leq \alpha,
\end{align*}
and the proof completes by applying the tower property of conditional expectation. 

Notice that when the potential outcomes are treated as fixed, the same proof applies to the null defined as $Y_j^T = Y_j^C$, because the independence between $A_i$ and $\mathcal{G}_{-1}$ still holds for the nulls. In the hybird version of the null~$H_{0i}^\text{zero}: Y_i^T = Y_i^C \text{ almost surely-}P_i$, the above proof applies with $\mathcal{G}_{-1} := \sigma\left(\{Y_j, Y_j^T, Y_j^C, X_j\}_{j=1}^n\right)$. Thus, FDR is controlled at level $\alpha$ conditional on the potential outcomes and covariates $\{Y_j^T, Y_j^C, X_j\}_{j=1}^n$. 
\end{proof}

\subsection{Proof of theorem~\ref{thm:LNO-ITE}} \label{apd:crossfit}
\begin{proof}
Let the set of false rejections in $\mathcal{R}(\mathcal{I})$ be $\mathcal{V}(\mathcal{I})$. We conclude that the FDR of the \pei implemented on set $\mathcal{I}$ is controlled at level $\alpha/2$:
\begin{align*}
    \mathbb{E}\left[\frac{|\mathcal{V}(\mathcal{I})|}{\max\{|\mathcal{R}(\mathcal{I})| ,1\}}~\Big|~\mathcal{G}_{-1}\right] \leq \alpha/2,
\end{align*}
following the error control of the  \pei in Section~\ref{apd:proof-naive}, where the initial candidate rejection set is $R_0 = \mathcal{I}$, and thus, $C_0 = \mathcal{I} \cap \mathcal{H}_0$. Similarly, the FDR of the \pei implemented on set $\mathcal{II}$ is also controlled at level $\alpha/2$. Therefore, the FDR of the combined set $\mathcal{R}(\mathcal{I}) \cup \mathcal{R}(\mathcal{II})$ is controlled at level $\alpha$ as claimed: 
\begin{align*}
    &\mathbb{E}\left[\frac{|\mathcal{V}(\mathcal{I}) \cup \mathcal{V}(\mathcal{II})|}{|\mathcal{R}(\mathcal{I}) \cup \max\{|\mathcal{R}(\mathcal{II})| ,1\}}~\Big|~\mathcal{G}_{-1}\right]{}\\
    \leq~& \mathbb{E}\left[\frac{|\mathcal{V}(\mathcal{I})|}{|\mathcal{R}(\mathcal{I}) \cup \max\{|\mathcal{R}(\mathcal{II})| ,1\}}~\Big|~\mathcal{G}_{-1}\right] + \mathbb{E}\left[\frac{|\mathcal{V}(\mathcal{II})|}{|\mathcal{R}(\mathcal{I}) \cup \max\{|\mathcal{R}(\mathcal{II})| ,1\}}~\Big|~\mathcal{G}_{-1}\right]{}\\
    \leq~& \mathbb{E}\left[\frac{|\mathcal{V}(\mathcal{I})|}{\max\{|\mathcal{R}(\mathcal{I})| ,1\}}~\Big|~\mathcal{G}_{-1}\right] + \mathbb{E}\left[\frac{|\mathcal{V}(\mathcal{II})|}{| \max\{|\mathcal{R}(\mathcal{II})| ,1\}}~\Big|~\mathcal{G}_{-1}\right] \leq \alpha,
\end{align*}
the proof completes for the null~\eqref{eq:hp_twoside} in the main paper after applying the tower property of conditional expectation. The FDR control also applies to the other two definitions of the null with fixed or hybrid version of the outcomes, following the same arguments as the end of Section~\ref{apd:proof-naive}.
\end{proof}


\subsection{Proof of theorem~\ref{thm:LHO-ITE}}
\label{apd:MaY-regular}
\begin{proof}
We prove that the FDR control holds for the \pei implemented on $\mathcal{I}$, and the same conclusion applies to $\mathcal{II}$, so the overall FDR control is guaranteed following the proof of theorem~\ref{thm:LNO-ITE} in Section~\ref{apd:crossfit}.

We first present the proof when the potential outcomes are viewed as random variables. Define $\mathcal{G}_{-1} := \sigma\left(\left\{X_i\right\}_{i=1}^n, \left\{Y_i, A_i\right\}_{i \notin \mathcal{I}}\right)$, and $\mathcal{G}'_{t} = \sigma\left(\mathcal{G}_{-1}, \mathcal{C}_{t}, (Y_i, A_i)_{i\notin \mathcal{C}_{t}}, \sum_{i \in \mathcal{C}_{t}} b_i \right)$, which contains more information than  $\mathcal{G}_{t}$ as defined in~\eqref{eq:fil_G}. We claim that Lemma~\ref{lm:adapt} holds when we replace $\mathcal{G}_{t}$ by $\mathcal{G}'_{t}$, 
because the distribution of $b_i$ conditional on $\mathcal{G}_{t}$ is the same as on $\mathcal{G}'_{t}$ for any $t = 0, \ldots, n$.
Similar to the proof of Theorem~\ref{thm:naive-ITE} in Section~\ref{apd:proof-naive}, we check that the assumption in Lemma~\ref{lm:adapt} are satisfied: (a)~the filtration in our algorithm satisfies $\mathcal{F}_{t} \subseteq \mathcal{G}'_{t}$, so the time of stopping the algorithm $\widehat t := \min\{t:\widehat{\mathrm{FDR}}(\mathcal{R}_t) \leq \alpha\}$ is a stopping time with respect to $\mathcal{G}'_t$; and (b)~$\mathcal{C}_{t+1}$ is measurable with respect to $\mathcal{G}'_t$; and (c)~for subjects with nonpositive effect $i \in \mathcal{H}_0^\text{nonpositive}$:
\begin{align} \label{clm:MaY-sign}
    \mathbb{P}\left((A_i - 1/2) \cdot E_i^{-\mathcal{I}} \leq 0 \mid \mathcal{G}_{-1} \right) \geq 1/2.
\end{align}
To see that the last assumption holds, notice that
\begin{align*}
    &\mathbb{P}\left((A_i - 1/2) \cdot (Y_i - \widehat m^{-\mathcal{I}}(X_i)) \leq 0 \mid \mathcal{G}_{-1}\right){}\\
    =~& \mathbb{P}(Y_i^C \geq \widehat m^{-\mathcal{I}}(X_i) \mid \mathcal{G}_{-1}) \mathbb{P}(A_i = 0) + \mathbb{P}(Y_i^T \leq \widehat m^{-\mathcal{I}}(X_i) \mid \mathcal{G}_{-1}) \mathbb{P}(A_i = 1); \text{ and} {}\\
    &\mathbb{P}\left((A_i - 1/2) \cdot (Y_i - \widehat m^{-\mathcal{I}}(X_i)) > 0 \mid \mathcal{G}_{-1}\right){}\\ 
    =~& \mathbb{P}(Y_i^C < \widehat m^{-\mathcal{I}}(X_i) \mid \mathcal{G}_{-1}) \mathbb{P}(A_i = 0) + \mathbb{P}(Y_i^C > \widehat m^{-\mathcal{I}}(X_i) \mid \mathcal{G}_{-1}) \mathbb{P}(A_i = 1).
\end{align*}
For any potential outcomes of the nulls such that $(Y_i^T \mid X_i) \preceq (Y_i^C \mid X_i)$, it holds that
\begin{align*} 
    \mathbb{P}(Y_i^C \geq D \mid X_i) \geq \mathbb{P}(Y_i^T > D \mid X_i), \text{ and } 
    \mathbb{P}(Y_i^T \leq D \mid X_i) \geq \mathbb{P}(Y_i^C < D \mid X_i),
\end{align*}
for any constant $D$, so 
\begin{align*}
   &\mathbb{P}(Y_i^C \geq \widehat m^{-\mathcal{I}}(X_i) \mid \mathcal{G}_{-1})  \geq  \mathbb{P}(Y_i^T > \widehat m^{-\mathcal{I}}(X_i) \mid \mathcal{G}_{-1}), \text{ and }{}\\
   &\mathbb{P}(Y_i^T \leq \widehat m^{-\mathcal{I}}(X_i) \mid \mathcal{G}_{-1})  \geq \mathbb{P}(Y_i^C < \widehat m^{-\mathcal{I}}(X_i) \mid \mathcal{G}_{-1}),
\end{align*}
because $m^{-\mathcal{I}}(X_i)$ is fixed given $\mathcal{G}_{-1}$. Because $\mathbb{P}(A_i = 1)$ is $1/2$ for all subjects, we have
\begin{align*}
    \mathbb{P}\left((A_i - 1/2) \cdot (Y_i - \widehat m^{-\mathcal{I}}(X_i)) \leq 0 \mid \mathcal{G}_{-1}\right) \geq \mathbb{P}\left((A_i - 1/2) \cdot (Y_i - \widehat m^{-\mathcal{I}}(X_i)) > 0 \mid \mathcal{G}_{-1}\right),
\end{align*}
which proves Claim~\eqref{clm:MaY-sign} and in turn the FDR control of the \peiNeg.

When the potential outcomes are treated as fixed, the above proof applies to the null defined as $Y_i^T \leq Y_i^C$ in~\eqref{eq:hp_oneside_fixed} of the main paper, in which case $ \mathbb{P}(Y_i^C \geq D \mid \mathcal{G}_{-1})$ is zero or one, and the above arguments still hold. For the hybird version of the null~\eqref{eq:hp_oneside_hybrid} in the main paper, the above proof applies with $\mathcal{G}_{-1} := \sigma\left(\left\{Y_i^T, Y_i^C, X_i\right\}_{i=1}^n, \left\{Y_i, A_i\right\}_{i \notin \mathcal{I}}\right)$. Thus, FDR is controlled at level $\alpha$ conditional on the potential outcomes and covariates $\{Y_j^T, Y_j^C, X_j\}_{j=1}^n$. 
\end{proof}

\subsection{Preliminaries to proof of error controls under observational studies}

\begin{lemma} \label{lm:FDR_hetero}
Let $q_i$ be the conditional probability of a positive estimated sign:
\begin{align*}
     q_{i} := \mathbb{P}\left[(A_i - 1/2) \cdot (Y_i - \widehat m(X_i)) > 0 \mid \mathcal{G}_{-1} \right],
\end{align*}
where $\widehat m$ is an summary statistic (mean or mean) of $Y_i \mid X_i$, learned using $\mathcal{G}_{-1}$. Denote the maximum as $q_{\max}(\mathcal{I}) \equiv \max_{i \in \mathcal{I}} q_{i}$ and let its estimation of using information in $\mathcal{G}_{-1}$ be~$\widehat q_{\max}$, and the (one-sided) estimation error be $\epsilon_n^q(\mathcal{I}) = \max\{q_{\max}(\mathcal{I}) - \widehat q_{\max}(\mathcal{I}),0\}$. Define the FDR estimator as
\begin{align*}
    \widehat{\mathrm{FDR}}_{\widehat t}(\mathcal{I}) \equiv \left(\frac{1}{1 - \widehat{q_{\max}}(\mathcal{I})} - 1\right)\frac{|\mathcal{R}_{\widehat t}^-| + 1}{|\mathcal{R}_{\widehat t}^+| \vee 1},
\end{align*}
then the FDR of \pei run by Analyst I at level $\alpha/2$ is bounded:
\begin{align*}
    \mathbb{E}\left[\mathrm{FDP}_{\widehat t}(\mathcal{I}) \mid  \mathcal{G}_{-1} \right] \leq \alpha/2 \left\{1 + \epsilon_n^q(\mathcal{I}) \cdot \frac{4}{q_{\max}(\mathcal{I}) (1 - q_{\max}(\mathcal{I}))}\right\}.
\end{align*}
when $\epsilon_n^q(\mathcal{I}) \leq q_{\max}(\mathcal{I}) / 2$.
\end{lemma}

\begin{proof}
By Lemma~\ref{lm:adapt} where 
\begin{align*}
    b_i := \one\{(A_i - 1/2) \cdot E_i \leq 0\}  \quad \text{ and } \quad \mathcal{C}_t := \mathcal{R}_{t} \cap \mathcal{H}_0, 
\end{align*}
and the tower property of conditional expectation, we have
\begin{align*}
    \mathbb{E}\left[\frac{|\mathcal{R}_{\widehat t}^+ \cap \mathcal{H}_0|}{1 + |\mathcal{R}_{\widehat t}^-|} \Big| \mathcal{G}_{-1}\right] \leq \left(\frac{1}{1 - q_{\max}(\mathcal{I})} - 1\right),
\end{align*}
 where the stopping time is denoted as $\widehat t$. The FDR at $\widehat t$ is upper bounded:
\begin{align*}
    &\mathbb{E}\left[\frac{|\mathcal{R}_{\widehat t}^+ \cap \mathcal{H}_0|}{|\mathcal{R}_{\widehat t}^+| \vee 1} \Big| \mathcal{G}_{-1}\right]{}\\ 
    =~& \mathbb{E}\left[\frac{1 + |\mathcal{R}_{\widehat t}^-|}{|\mathcal{R}_{\widehat t}^+| \vee 1} \cdot \frac{|\mathcal{R}_{\widehat t}^+ \cap \mathcal{H}_0|}{1 + |\mathcal{R}_{\widehat t}^-|} \Big| \mathcal{G}_{-1}\right]{}\\
    \leq~& \alpha/2 \left(\frac{1}{1 - \widehat{q_{\max}}(\mathcal{I})} - 1\right)^{-1}  \mathbb{E}\left[\frac{|\mathcal{R}_{\widehat t}^+ \cap \mathcal{H}_0|}{1 + |\mathcal{R}_{\widehat t}^-|} \Big| \mathcal{G}_{-1} \right]{}\\
    \leq~& \alpha/2 \left(\frac{1}{1 - \widehat{q_{\max}}(\mathcal{I})} - 1\right)^{-1}  \left(\frac{1}{1 - q_{\max}(\mathcal{I})} - 1\right). 
\end{align*}
By Taylor expansion on $f(x) = \left(\frac{1}{1 - x} - 1\right)^{-1}$ around $x_0 = q_{\max}(\mathcal{I})$, we have 
$$
f(x_0 - \epsilon) \leq f(x_0) + \frac{\epsilon}{(x_0 - \epsilon)^2} \leq f(x_0) + \frac{4\epsilon}{x_0^2}$$ 
when $0 \leq \epsilon \leq \frac{x_0}{2}$. Thus, FDR is close to the target level when $q_{\max}(\mathcal{I}) -\widehat{q_{\max}}(\mathcal{I})$ is small:
\begin{align*}
    \mathbb{E}\left[\mathrm{FDP}_{\widehat t}^{\widehat \pi}(\mathcal{I}) \mid \mathcal{G}_{-1}\right] \leq~&  \alpha/2 \left\{\left(\frac{1}{1 - q_{\max}(\mathcal{I})} - 1\right)^{-1} + 4\epsilon_n^q(\mathcal{I}) \left(\frac{1}{q_{\max}(\mathcal{I})}\right)^2 \right\} \left(\frac{1}{1 - q_{\max}(\mathcal{I})} - 1 \right){}\\
    =~& \alpha/2 \left\{1 + \epsilon_n^q(\mathcal{I}) \frac{4}{q_{\max}(\mathcal{I})(1 - q_{\max}(\mathcal{I}))} \right\},
\end{align*}
when $\epsilon_n^q(\mathcal{I}) \leq q_{\max}(\mathcal{I}) / 2$.
\end{proof}

\subsection{Proof of theorem~\ref{thm:crossfit-unknown}}
\label{apd:proof-crossfit-unknown}
\begin{proof}

By Lemma~\ref{lm:FDR_hetero}, the probability of a positive sign of the estimated treatment effect $q_i$ is
\begin{align*}
    q_{i} :=~&\mathbb{P}\left((A_i - 1/2) \cdot E_i > 0 \mid \mathcal{G}_{-1}\right){}\\
    =~&  \pi_i\mathbb{P}(E_i > 0 \mid \mathcal{G}_{-1}) + (1 - \pi_i)\mathbb{P}(E_i < 0 \mid \mathcal{G}_{-1}) \leq \max\{1 - \pi_{\min}, \pi_{\max}\}.
\end{align*}

\noindent Thus, $q_{\max}(\mathcal{I}) = \max\{1 - \pi_{\min}, \pi_{\max}\}$ and $\widehat q_{\max}(\mathcal{I}) = \max\{1 - \widehat{\pi_{\min}}(\mathcal{I}), \widehat{\pi_{\max}}(\mathcal{I})\}$, and $\epsilon_n^q(\mathcal{I}) = q_{\max}(\mathcal{I}) - \widehat q_{\max}(\mathcal{I})$, we have 
\begin{align*}
    \mathbb{E}\left[\mathrm{FDP}_{\widehat t}^{\widehat \pi}(\mathcal{I}) \right] \leq \alpha/2 \left\{1 + \mathbb{E}_{\mathcal{F}_{0}(\mathcal{I})}\left[\epsilon_n^q(\mathcal{I}) \cdot \frac{4}{q_{\max}(\mathcal{I}) (1 - q_{\max}(\mathcal{I}))}\right] \right\}.
\end{align*}
when $\epsilon_n^q(\mathcal{I}) \leq  \frac{1}{2}\max\{1 - \pi_{\min}, \pi_{\max}\}$.
\end{proof}

\subsection{Proof of theorem~\ref{thm:MaY-unknown}} \label{apd:MaY-unknown}
\begin{proof}
Denote the individual propensity score as $\mathbb{P}(A_i = 1 \mid X_i) = \pi_i$. For the nulls, the probability of a positive sign of the estimated treatment effect as $q_i$:
\begin{align*}
    q_{i}(\mathcal{I}) :=~&\mathbb{P}\left((A_i - 1/2) \cdot (Y_i - \widehat m(X_i)) > 0 \mid \mathcal{F}_{0}^{-Y}(\mathcal{I})\right){}\\
    =~& \pi_i\mathbb{P}(Y_i - \widehat m(x_i) > 0 \mid \mathcal{F}_{0}^{-Y}(\mathcal{I})) + (1 - \pi_i)\mathbb{P}(Y_i - \widehat m(x_i) < 0 \mid \mathcal{F}_{0}^{-Y}(\mathcal{I})){}\\
    \leq~& \min\left\{\max\{\pi_i,1-\pi_i\}, \max\{\Phi_{\max}\left[\epsilon_n^Y(\mathcal{I})\right], 1 - \Phi_{\min}\left[-\epsilon_n^Y(\mathcal{I})\right]\}\right\}
\end{align*}
where $\mathbb{P}(Y_i - \widehat m(x_i) > 0 \mid \mathcal{F}_{0}^{-Y}(\mathcal{I}))$ can be separated from $\mathbb{P}(A_i - 1/2 > 0 \mid \mathcal{F}_{0}^{-Y}(\mathcal{I}))$ because they are independent for zero-effect nulls. Thus, let an upper bound be
\begin{align*}
    q_{\max}(\mathcal{I}) = \min\{\max\{\pi_{\max},1-\pi_{\min}\},\max\{\Phi_{\max}\left[\epsilon_n^Y(\mathcal{I})\right], 1 - \Phi_{\min}\left[-\epsilon_n^Y(\mathcal{I})\right]\}\}.
\end{align*}
Let the estimator of $q_i$ be 
\begin{align*}
    \widehat q_i := \max\{\widehat \pi_i, 1 - \widehat \pi_i\}.
\end{align*}

To describe the resulting estimation error of $q_i$, we define a difference $d_{i}(\mathcal{I}) :=  \max\{\pi_i,1-\pi_i\} - \max\{\Phi_{\max}\left[\epsilon_n^Y(\mathcal{I})\right], 1 - \Phi_{\min}\left[-\epsilon_n^Y(\mathcal{I})\right]\}$, which takes large value if the propensity score deviates from~$1/2$ (smaller value if the outcome probability deviates from $1/2$). The true $q_i$ is upper bounded by estimated $\widehat q_{i}$ plus some estimation error that depends on $d_{i}(\mathcal{I})$: 
\begin{align*}
    q_i -\widehat{q_{i}}  \leq~& \epsilon_{i}^{\pi}(\mathcal{I}) \quad~& \text{ if } d_{i}(\mathcal{I}) \leq 0;{}\\
    q_i -\widehat{q_i}  \leq~&  \epsilon_{i}^{\pi}(\mathcal{I}) - d_{i}(\mathcal{I}) \quad~& \text{ if }  d_{i}(\mathcal{I}) > 0,
\end{align*}
where $\epsilon_{i}^{\pi}(\mathcal{I}) = \pi_{i} - \widehat \pi_{i}(\mathcal{I})$; it can be written in one line as 
\begin{align*}
   q_i -\widehat{q_i} \leq \epsilon_{i}^{\pi}(\mathcal{I}) - \max\left\{0, \max\{\pi_i,1-\pi_i\} - \max\{\Phi_{\max}\left[\epsilon_n^Y(\mathcal{I})\right], 1 - \Phi_{\min}\left[-\epsilon_n^Y(\mathcal{I})\right]\}\right\}.
\end{align*}
Thus, the estimation error for $q_{\max}(\mathcal{I})$ is upper bounded as 
\begin{align}
     &\max_{i \in \mathcal{I}} \{\epsilon_{i}^{\pi}(\mathcal{I}) - \max\left\{0, \max\{\pi_i,1-\pi_i\} - \max\{\Phi_{\max}\left[\epsilon_n^Y(\mathcal{I})\right], 1 - \Phi_{\min}\left[-\epsilon_n^Y(\mathcal{I})\right]\}\right\}\}{}\\ 
     \leq~& \epsilon_n^{\pi}(\mathcal{I}) - \max\left\{0, \max\{\pi_{\max},1-\pi_{\min}\} - \max\{\Phi_{\max}\left[\epsilon_n^Y(\mathcal{I})\right], 1 - \Phi_{\min}\left[-\epsilon_n^Y(\mathcal{I})\right]\}\right\}\} =: \epsilon^q_n(\mathcal{I}). 
\end{align}
By Lemma~\ref{lm:FDR_hetero}, we have 
\begin{align*}
     \mathbb{E}\left[\mathrm{FDP}_{\widehat t}^{\widehat \pi}(\mathcal{I}) \right] \leq \alpha/2 \left\{1 + \mathbb{E}_{\mathcal{F}_0(\mathcal{I})}\left[\epsilon_n^{q}(\mathcal{I}) \left(\frac{4}{ q_{\max}(\mathcal{I}) (1 - q_{\max}(\mathcal{I}))}\right)\right] \right\},
\end{align*}
when $\epsilon_n^{q}(\mathcal{I}) \leq q_{\max}(\mathcal{I}) / 2$.
\end{proof}

\begin{remark}
If we consider nulls as nonpositive effects, we have
\begin{align*}
    q_i(\mathcal{I}) :=~&\mathbb{P}\left((A_i - 1/2) \cdot (Y_i - \widehat m(X_i)) > 0 \mid \mathcal{F}_{0}^{-Y}(\mathcal{I})\right){}\\
    =~& \pi_i\mathbb{P}(Y_i^T - \widehat m(x_i) > 0 \mid \mathcal{F}_{0}^{-Y}(\mathcal{I})) + (1 - \pi_i)\mathbb{P}(Y_i^C - \widehat m(x_i) < 0 \mid \mathcal{F}_{0}^{-Y}(\mathcal{I})){}\\
     \leq~& \pi_i\mathbb{P}(Y_i^C - \widehat m(x_i) > 0 \mid \mathcal{F}_{0}^{-Y}(\mathcal{I})) + (1 - \pi_i)\mathbb{P}(Y_i^C - \widehat m(x_i) < 0 \mid \mathcal{F}_{0}^{-Y}(\mathcal{I})){}\\
    \leq~& \min\left\{\max\{\pi_i,1-\pi_i\}, \max\{\Phi_{\max}\left[\epsilon_n^Y(\mathcal{I})\right], 1 - \Phi_{\min}\left[-\epsilon_n^Y(\mathcal{I})\right]\}\right\},
\end{align*}
where $\Phi_{\max}(c) := \max_{i \in \mathcal{I}} \mathbb{P}(Y_i^C - \me(Y_i^C|X_i) \leq c \mid X_i)$, $\Phi_{\min}(c) := \min_{i \in \mathcal{I}} \mathbb{P}(Y_i^C - \me(Y_i^C|X_i) \leq c \mid X_i)$ and $\epsilon_n^Y(\mathcal{I}) = \max_{i \in \mathcal{I}} |\widehat m(X_i) - \me(Y_i^C|X_i)|$, and then, we can make the same claim as above. However, it is harder to have robust FDR control when the propensity scores are poorly estimated, because $\widehat m(x_i)$ is an estimator for $\me(Y_i \mid X_i)$, which can be very different from the expected \textit{control} outcome for a subject with negative effect. (We can design an algorithm where $\widehat m(x_i)$ is an estimation of the expected control outcome, but when the estimation is well, it would have zero power to detect positive effect if the subject is not treated.) 
\end{remark}

\subsection{Error control guarantee for the linear-BH procedure}
\label{apd:FDR_BH_linear}

\begin{theorem} \label{thm:linear-BH}
Suppose the outcomes follow a linear model: $Y_i = l^{\Delta}(X_i) A_i + l^{f}(X_i) + U_i$, where $l$ denotes a linear function, and $U_i$ is standard Gaussian noise. The linear-BH procedure controls FDR of the nonpositive-effect null in~\eqref{eq:hp_oneside} of the main paper asymptotically as the sample size $n$ goes to infinity.
\end{theorem}
Note that the error control would not hold when the linear assumption is violated. For example, if the expected treatment effect $\mathbb{E}(Y_i^T - Y_i^C \mid X_i)$ is some nonlinear function of the covariates, the estimated treatment effect~$\widehat \Delta_i^\text{BH}$ would not be consistent; in turn, for the null subjects with zero effect, the $p$-values would not be valid (i.e., not stochastically equal or larger than uniform). Hence, the linear-BH procedure would not guarantee the desired FDR control, as we show in the numerical experiments in Section~\ref{sec:sim} of the main paper. 

\begin{proof}
For simplicity, we treat all the covariates as fixed values and denote them as the covariance matrix $\mathbb{X}_a = (X_i: A_i = a)^T$ for $a \in \{T, C\}$, where we temporarily use $A_i = T$ to denote the case of being treated $A_i = 1$. Under the linear assumption, the estimated outcome $\widehat l^a$ asymptotically follows a Gaussian distribution, whose expected value is $l^{\Delta}(X_i) \one\{a = T\} + l^{f}(X_i)$. Its variance can be estimated as
\begin{align*}
    \widehat{\text{Var}}(\widehat l^a(X_i)) = \widehat \sigma_a^2 (X_i^T(\mathbb{X}_a^T\mathbb{X}_a)^{-1}X_i^T),
\end{align*}
where the variance from noise is estimated as 
\begin{align*}
    \widehat \sigma^2_a = \sum_{A_i = a}(Y_i - \widehat l^a(X_i))^2 \Bigg/\left(\sum_i \one\{A_i = a\} - d - 1\right),
\end{align*}
and $d$ is the number of covariates. Note that the observed outcome also follows a Gaussian distribution $N(l^{\Delta}(X_i) \one\{a = T\} + l^{f}(X_i), \sigma^2_a)$. Note that in each estimated effect $\widehat{\Delta}_i^{\text{BH}}$, the observed outcome $Y_i^a$ is independent of the estimated potential outcome $Y_i^{\overline a}$, where $\overline a$ is the complement of $a$: $\overline a \cup a = \{T, C\}$. Thus, the estimated effect asymptotically follow a Gaussian distribution whose expected value is $l^{\Delta}(X_i)$ (nonpositive under the null) and the variance is $\text{Var}(\widehat{\Delta}_i^{\text{BH}}) = \text{Var}(\widetilde Y_i^T) + \text{Var} (\widetilde Y_i^C)$, where an estimation is $\widehat{\text{Var}}(\widetilde Y_i^a) = \widehat \sigma^2_a \one\{A_i = a\} + \widehat{\text{Var}}(\widehat l^a(X_i)) \one\{A_i = \overline a\}$. Therefore, the resulting $p$-value $P_i$ as defined in~\eqref{eq:p-linear-BH} of the main paper is asymptotically valid (uniform or stochastically larger) if subject $i$ is a null, and hence the BH procedure leads to asymptotic FDR control \citep{fan2007many}.
\end{proof}

\section{Proof of power analysis} \label{apd:power}




Our proof of the power analysis mainly uses the results in \citet{arias2017distribution}, who consider the setup with $n$ hypotheses, each associated with a test statistic $V_i$ for $i \in [n]$. Assume the test statistics are independent with the survival function $\mathbb{P}(V_i \geq x) = \Psi_i(x)$, which equals $\Psi(x - \mu_i)$ where $\mu_i = 0$ under the null and $\mu_i > 0$ otherwise. They focus on a class of distribution called asymptotically generalized Gaussian (AGG), whose survial function satisfies:
\begin{align} \label{eq:AGG_original}
    \lim_{x \to \infty} x^{-\gamma} \log \Psi(x) = -1/\gamma,
\end{align}
with a constant $\gamma > 0$. For example, a normal distribution is AGG with $\gamma = 2$. They discuss a class of multiple testing methods called \textit{threshold procedure}: the final rejection set $\mathcal{R}$ is in the form 
\begin{align} \label{eq:threshold_procedure}
    \mathcal{R} = \{i: V_i \geq \tau(V_1, \ldots, V_n)\},
\end{align}
for some threshold $\tau(V_1, \ldots, V_n)$, and separately study two types of thresholds: the BH procedure \citep{benjamini1995controlling} with threshold:
\begin{align}
    \tau_{\text{BH}} = V_{(\iota_\text{BH})}, \quad \iota_\text{BH} := \max\{i: V_{(i)} \geq \Psi^{-1}(i\alpha/n)\},
\end{align}
where $V_{(1)} \geq \ldots \geq V_{(n)}$ are ordered statistics; and the Barber-Cand\`{e}s (BC) procedure \citep{barber2015controlling} with a threshold on the absolute value of $V_i$:
\begin{align}
    \tau_{\text{BC}} = \inf\{\nu \in |\mathbf{V}|: \widehat{\text{FDP}}(\nu) \leq \alpha\},
\end{align}
where $|\mathbf{V}| := \{|V_i|: i \in [n]\}$ is the set of sample absolute values, and
\begin{align*}
    \widehat{\text{FDP}}(\nu) := \frac{ |\{i: V_i \leq -\nu\}| + 1}{\max\{|\{i: V_i \geq \nu\}|,1\}},
\end{align*}
and the final rejection set is those with positive $V_i$ and value larger than $\tau_{\text{BC}}$. The stopping rule for the BC procedure is similar to our proposed algorithms, as detailed next. 



Recall in Section~\ref{sec:power_analysis} of the main paper, we consider a simplified automated version of the \pei that exclude the subject with the smallest absolute value of the estimated treatment effect $|\widehat \Delta_i|$. Thus, the automated \pei is a BC procedure where the test statistic of interest is $V_i = \widehat \Delta_i = 4(A_i - 1/2)(Y_i - \widehat m_n(X_i))$. Following the above notations and let $\Phi$ be the CDF for standard Gaussian, we denote the survival function for the nulls as 
\begin{align*}
    \Psi_n^{\text{null}}(x) = \frac{1}{2}(1 - \Phi(x + \widehat m_n)) + \frac{1}{2}(1 - \Phi(x - \widehat m_n)),
\end{align*}
which is a mixture of two Gaussians, with $\widehat m_n = \tfrac{1}{n}\sum_{i=1}^n Y_i$, and $\widehat m_n \overset{a.s.}{\to} 0$ by the strong law of large numbers. For the non-nulls, the survival function is
\begin{align*}
    \Psi_n^{\text{non-null}}(x) = \frac{1}{2}(1 - \Phi(x + \widehat m_n - \mu)) + \frac{1}{2}(1 - \Phi(x - \widehat m_n)).
\end{align*}
Note that our setting is slightly different from the discussion in \citet{arias2017distribution} because the non-nulls differ from the nulls by a shift on one of the Gaussian component (rather than a shift in the overall survival function $\Psi_n^{\text{null}}$). Similar to the characterization by AGG in~\eqref{eq:AGG_original}, both survival functions $\Psi_n^{\text{null}}$ and $\Psi_n^{\text{non-null}}$ asymptotically satisfy a tail property that for any $x_n \to \infty$ as $n \to \infty$:
\begin{align} \label{eq:AGG}
    \lim_{n\to \infty} x_n^{-\gamma} \log \Psi_n(x_n) = -1/\gamma,
\end{align}
with probability one and $\gamma = 2$, which we later refer to as asymptotic AGG. Conclusions in our paper basically follows the proofs in \citet{arias2017distribution} with the test statistics $V_i$ specified as the estimated treatment effect $\widehat \Delta_i$.

\subsection{Proof of theorem~\ref{thm:power_noX}} \label{apd:power_noX}

We first present the proof for the power of the automated \peiZero, and the power of the linear-BH is proven similarly as shown later.

\begin{proof}

\noindent \textbf{Zero power when $r < \beta$.}
The argument of zero power indeed applies to any threshold procedure as defined in~\eqref{eq:threshold_procedure}: $\mathcal{R} = \{i: \widehat \Delta_i \geq d\}$, for some $d \in \mathbb{R}$. Following the proof of Theorem~1 in \citet{arias2017distribution}, we argue that the FDR control cannot be satisfied for any $\alpha \in (0,1)$ unless the threshold $d$ is large enough such that $d > \mu + \delta_n$ with $\delta_n = \log\log n$; but in this case, most non-nulls cannot be included in the rejection set, and thus the power goes to zero. 

First, we claim that when $d \leq \mu + \delta_n$, the false discovery proportion (FDP) goes to one in probability. By the proof of Theorem~1 in \citet{arias2017distribution}, we have that FDP goes to one in probability if $\frac{(n - n_1)\Psi_n^{\text{null}}(\mu + \delta_n)}{n_1} \to \infty$ with probability one, where $n_1$ is the number of non-nulls. Their proof also verifies that $\frac{(n - n_1)\Psi_n^{\text{null}}(\mu + \delta_n)}{n_1} \to \infty$ because $\Psi_n^{\text{null}}$ satisfies property~\eqref{eq:AGG} with probability one.

Next, we show that when $d > \mu + \delta_n$, the power goes to zero. Notice that power can be equivalently defined as $\mathbb{E}(1 - \text{FNR})$, where FNR (false negative rate) is defined as the proportion of non-nulls not identified. Again by the proof of Theorem~1 in \citet{arias2017distribution}, we have that the FNR converge to one in probability if $\Psi_n^{\text{non-null}}(\mu + \delta_n)$ goes to zero with probability one, which is true because $\delta_n \to \infty$.

Combining the above two arguments, we conclude that for any threshold procedure whose rejection set is in the form of~\eqref{eq:threshold_procedure}, the power goes to zero for any FDR control $\alpha \in (0,1)$ when $r < \beta$.

Note that the above proof assumes that the test statistics $V_i = \widehat \Delta_i$ are mutually independent. For simplicity, we design the \peiZero where $\widehat m_n$ for the \pei implemented on $\mathcal{I}$ is computed using data in $\mathcal{II}$, to ensure the above mutual independence. Thus, the above proof applies to the \pei implemented on each half, $\mathcal{I}$ and  $\mathcal{II}$.
The overall power behaves the same asymptotically, since $\mathcal{I}$ and  $\mathcal{II}$ result from a random split of all subjects $[n]$. For all cases hereafter, we prove the power claim for the \pei implemented on $\mathcal{I}$ conditional on data in~$\mathcal{II}$, and the same claim holds for the overall power as reasoned above. \\

\noindent \textbf{Half power when $r > \beta$.} We first prove the limit inferior of the power is at least $1/2$, and then the limit superior is at most $1/2$, mainly using the proof of Theorem~3 in \citet{arias2017distribution}.

They consider a sequence of thresholds $d_{n}^* = (\gamma r^* \log n)^{1/\gamma}$ for some $r^* \in (\beta, r\wedge 1)$. We first claim that the FDR estimator at $d_n^*$ is less than any $\alpha \in (0,1)$ for large~$n$, or mathematically $\widehat{\text{FDR}}(d_n^*) \leq \alpha$. It can be verified by the proof of Theorem~3 in~\citet{arias2017distribution} where the survival function of $\widehat \Delta_i$ is $G(d_n^*) = (1 - \epsilon) \Psi_n^{\text{null}}(d_n^*) + \epsilon \Psi_n^{\text{non-null}}(d_n^*)$ with $\epsilon = n^{-\beta}$, and the fact that $\Psi_n^{\text{non-null}}(d_n^*) \to 1/2$ and $\Psi_n^{\text{null}}(d_n^*) \to n^{-r^*}$ (by property~\eqref{eq:AGG}) with probability one. It follows that the true stopping threshold $\tau_n$ satisfies $\tau_n \leq d_{n}^*$. Also, by Lemma~1 in \citet{arias2017distribution}, we have that the proportion of correctly identified non-nulls at threshold $d_n^*$ is $\frac{1}{n_1} \sum_{i \notin \mathcal{H}_0} \one\{\widehat \Delta_i \geq d\} = \Psi_n^{\text{non-null}}(d_n^*) + o_{\mathbb{P}}(1)$, where $\Psi_n^{\text{non-null}}(d)$ decreases in $d$ and converges to $1/2$ when $d = d_n^*$. Recall that the true stopping threshold is no larger than $d_{n}^*$, so the limit inferior of the power is at least $1/2$.

The power converges to $1/2$ once we show that the limit superior of the power is at most $1/2$. Consider a positive constant $d^0 \in (0, \infty)$, and we claim that the actual stopping threshold $\tau_n \geq d^0$ for large $n$ because the FDR estimator goes to one, following similar arguments in the proof of Theorem~3 in~\citet{arias2017distribution}. Specifically, 
\begin{align*}
    \widehat{\text{FDR}}(d^0) \equiv \frac{|\{i \in [n]: \widehat \Delta_i \leq -d^0\}| + 1}{\max\{|\{i \in [n]: \widehat \Delta_i \geq d^0\}|,1\}} =
    \frac{1 + n(1 - \widehat G_n(-d^0))}{\max\{n\widehat G_n(d^0),1\}}, 
\end{align*}
where $\widehat G_n(d^0) = \tfrac{1}{n}\sum_{i \in [n]} \one(\widehat\Delta_i \geq d^0)$ denotes the empirical survival function. Use the fact that $G_n(d^0) = (1 - \epsilon) \Psi_n^{\text{null}}(d^0) + \epsilon \Psi_n^{\text{non-null}}(d^0) \to 1 - \Phi(d^0)$ and $G_n(-d^0) \to \Phi(d^0)$ almost surely, we observe that $\widehat{\text{FDR}}(d^0) \to 1$ with probability one. Also, the proportion of correctly identified non-nulls at threshold $d^0$ is $\Psi_n^{\text{non-null}}(d^0) + o_{\mathbb{P}}(1)$ (recall in the previous paragraph), where $\Psi_n^{\text{non-null}}(d^0) \to 1/2 + 1/2(1 - \Phi(d^0))$. Thus, the power for large $n$ is smaller than $1/2 + 1/2(1 - \Phi(d^0))$ for all $d^* \in (0, \infty)$; in other words, the limit superior of the power is smaller than $\inf_{d \in (0, \infty)} 1/2 + 1/2(1 - \Phi(d^*)) = 1/2$.

With the limit inferior and superior of the power bounded by $1/2$, we conclude that the power converges to $1/2$. In fact, the above proof implies that the power of identifying non-null subjects that are treated is one (notice that $\Psi_n^{\text{non-null, treated}}(d_n^*) = 1 - \Phi(d_n^* + \widehat m_n - \mu) \to 1$ with probability one, so the limit inferior of the power for treated non-nulls is at least 1).
\end{proof}

\textbf{Proof for the linear-BH procedure} The power of the linear-BH procedure when there is no covariates can be proved following similar steps as above, and using intermediate results of Theorem~2 in \citet{arias2017distribution}. In their notation, the linear-BH procedure uses $V_i = \widehat \Delta_i^\text{BH}$ as the test statistics, and we separately discuss power among the treated group and the control group in ensure the independence among $V_i$. The survival functions for the nulls and non-nulls in the treated group are
\begin{align*}
    \Psi_n^{\text{null}}(x) = 1 - \Phi\left(\frac{x}{\sqrt{1 + 1/n^C}}\right) \text{ and }  \Psi_n^{\text{non-null}}(x) = 1 - \Phi\left(\frac{x - \mu}{\sqrt{1 + 1/n^C}}\right),
\end{align*}
where $n^C$ is the number of untreated subjects. For the control group, the survival functions are
\begin{align*}
    \Psi_n^{\text{null}}(x) = 1 - \Phi\left(\frac{x + \widehat Y^T}{\sqrt{1 + 1/n^T}}\right) \text{ and }
    \Psi_n^{\text{non-null}}(x) = 1 - \Phi\left(\frac{x + \widehat Y^T}{\sqrt{1 + 1/n^T}}\right),
\end{align*}
where $\widehat Y^T = \sum_{A_i = 1} Y_i \overset{a.s.}{\to} 0$, and $n^T$ is the number of treated subjects. 
Since the above distributions converge to a Gaussian, these
survival functions satisfy the AGG property asymptotically as defined in~\eqref{eq:AGG}.  

\begin{proof}
First, we claim that the power goes to zero when $r < \beta$, following the proof in Section~\ref{apd:power_noX} for any threshold procedure (separately for the treated group conditional on control group). Then, we prove that power converges to $1/2$ when $r > \beta$: the power among untreated subjects is asymptotically zero because the survival functions for the nulls and non-nulls are the same; the power among treated subjects is asymptotically one following the proof of Theorem~2 in \citet{arias2017distribution} as detailed next.

Again consider the sequence of thresholds $d_{n}^* = (\gamma r^* \log n)^{1/\gamma}$ for some $r^* \in (\beta, r\wedge 1)$. We first claim that the FDR estimator at $d_n^*$ is less than any $\alpha \in (0,1)$ for large~$n$, or mathematically $\widehat{\text{FDR}}(d_n^*) \leq \alpha$. It can be verified by the proof of Theorem~2 in~\citet{arias2017distribution} where $G(d_n^*) = (1 - \epsilon) \Psi_n^{\text{null}}(d_n^*) + \epsilon \Psi_n^{\text{non-null}}(d_n^*)$ with $\epsilon = n^{-\beta}$, and the fact that $\Psi_n^{\text{non-null}}(d_n^*) \to 1$ for the treated group and $\Psi_n^{\text{null}}(d_n^*) \to n^{-r^*}$ (by property~\eqref{eq:AGG}) with probability one. It follows that the true stopping threshold $\tau_n$ satisfies $\tau_n \leq d_{n}^*$. Also, by Lemma~1 in \citet{arias2017distribution}, we have that the proportion of correctly identified non-nulls at threshold $d_n^*$ is $\frac{1}{n_1} \sum_{i \notin \mathcal{H}_0} \one\{\widehat \Delta_i^\text{BH} \geq d\} = \Psi_n^{\text{non-null}}(d_n^*) + o_{\mathbb{P}}(1)$, where $\Psi_n^{\text{non-null}}(d)$ for the treated group decreases in~$d$ and converges to $1$ when $d = d_n^*$. Recall that the true stopping threshold is no larger than $d_{n}^*$, so the limit inferior of the power among the treated subjects is at least $1$. Therefore, the overall power converges to $1/2$.

\end{proof}

\subsection{Proof of theorem~\ref{thm:power_oracle_X}} \label{apd:power_oracle_X}

We first consider the power when all subjects are non-nulls ($\beta = 0$). 
\begin{lemma} \label{lm:power_all_nonnull}
Given any fixed FDR control level $\alpha \in (0,1)$ and let the number of subjects $n$ goes to infinity. When all subjects are non-nulls, the stopping time $\tau = 0$ with probability tending to one if $\mu > \Phi^{-1}(\frac{1}{1 + \alpha})$, and in this case the power converges to $\Phi(\mu)$.
\end{lemma}
\noindent E.g., when $\alpha = 0.2$, the asymptotic power of the automated \pei is larger than~$0.8$ if $\mu \geq 1$.
\begin{proof}
The stopping time $\tau = 0$ if and only if the FDR control is satisfied when all the subjects are included: $\widehat{\text{FDR}}_n(\mathcal{R}_0) = \frac{|\mathbb{R}_0^-| + 1}{\max\{|\mathbb{R}_0^+|,1\}} \leq \alpha$, or equivalently, $\frac{|\mathbb{R}_0^+|}{n} \geq \frac{1 + \tfrac{1}{n}}{1 + \alpha}.$
Notice that the proportion of positive $\widehat \Delta_i$ converges to a constant: $\frac{|\mathbb{R}_0^+|}{n} \overset{a.s.}{\to} \Phi(\mu)$, because $\widehat \Delta_i$ of each non-null follows a Gaussian distribution with mean $\mu$ and variance less than $2$.  Thus, if $\Phi(\mu) > \frac{1}{1 + \alpha}$, for any $\epsilon \in (0,1)$, there exists $N$ such that for all ${n \geq N}$, we have that (a)~$\Big|\frac{|\mathbb{R}_0^+|}{n} - \Phi(\mu)\Big| < \epsilon$ with probability at least $1 - \epsilon$; and (b)~ $\widehat{\text{FDR}}_n(\mathcal{R}_0) = \frac{|\mathbb{R}_0^-| + 1}{\max\{|\mathbb{R}_0^+|,1\}} \leq \alpha$ (hence $\tau = 0$) with probability at least $1 - \epsilon$. (Notice that the threshold $N$ can be chosen as not depending on $\mu$, which is useful in the next proof.) In such a case, the power is no less than $(1 - \epsilon)(\Phi(\mu) - \epsilon)$ when ${n\geq N}$; and the power is no larger than $\Phi(\mu) - \epsilon$; so the power converges to $\Phi(\mu)$. The proof completes once notice that the condition $\Phi(\mu) > \frac{1}{1 + \alpha}$ is equivalent to $\mu > \Phi^{-1}(\frac{1}{1 + \alpha})$. 
\end{proof}

\begin{proof}[Proof of Theorem~\ref{thm:power_oracle_X}.]
\noindent \textbf{Power of the \peiZero.}
Recall that the \pei implemented on $\mathcal{I}$ exclude subjects based on the averaged estimated effect on $\mathcal{II}$: $\text{Pred}(x) = \overline{\widehat \Delta_i(X_i = x)}$, which converges to $\mu$ almost surely when $x = 1$ (the non-nulls), and $0$ almost surely when $x = 0$ (the nulls). Thus, no non-nulls in $\mathcal{I}$ would be excluded before excluding all the nulls in $\mathcal{I}$ (with probability going to one) for any fixed $\mu > 0$. Combined with Lemma~\ref{lm:power_all_nonnull}, we have that if $\mu > \Phi^{-1}(\frac{1}{1 + \alpha})$, for any $\epsilon \in (0,1)$, there exists $N(\epsilon, \alpha)$ such that for all $n \geq N$, the power of the \pei is higher than $\Phi(\mu) - \epsilon$. Also, the limit of power increases to one for any $r > 0$ (where the signal $\mu$ increases): there exists $N'(\epsilon)$ such that for all $n \geq N'(\epsilon)$, $\Phi(\mu) \geq 1 - \epsilon$. Therefore, for any $\epsilon \in (0, \tfrac{1}{1+\alpha})$, we have that for all $n \geq \max\{N'(\epsilon), N(\epsilon, \alpha)\}$, the power of \pei implemented on $\mathcal{I}$ is no less than $1 - 2\epsilon$; thus completes the proof.\\ 

\noindent \textbf{Power of the linear-BH procedure.}
As before, we separately argue that the power for the treated group and the control group converges to zero when $r < \beta$, and converges to one when $r > \beta$. For a subject in the treated group with $X_i = x$ where $x \in \{0, 1\}$, the estimated effect is a Gaussian $\widehat \Delta_i^\text{BH} \sim N(0, 1 + \frac{1}{\sum_i \one(X_i = x, A_i = 0)})$ for the nulls and $\widehat \Delta_i^\text{BH} \sim N(\mu, 1 + \frac{1}{\sum_i \one(X_i = x, A_i = 0)})$ for the non-nulls. For a subject in the control group with $X_i = x$ where $x \in \{0, 1\}$, the estimated effect is a Gaussian $\widehat \Delta_i^\text{BH} \sim N(0, 1 + \frac{1}{\sum_i \one(X_i = x, A_i = 1)})$ for the nulls and $\widehat \Delta_i^\text{BH} \sim N(\mu, 1 + \frac{1}{\sum_i \one(X_i = x, A_i = 1)})$ for the non-nulls.

The power of the linear-BH procedure directly results from Theorem~2 in~\citet{arias2017distribution} because in both the treated and control group, (a)~the linear-BH procedure is the BH procedure where the random variable of interest is $V_i = \widehat \Delta_i^\text{BH}$; and (b)~$\widehat \Delta_i^\text{BH}$ of non-nulls and nulls differ by a shift $\mu$; and (c)~the survival function of $\widehat \Delta_i^\text{BH}$ is asymptotically AGG (recall definition in~\eqref{eq:AGG}) since it converges to a Gaussian distribution).
\end{proof}

\section{Experiments when data does not follow working model~\eqref{eq:work_model}} 
\label{apd:working_model}

In previous discussions, we follow the working model~\eqref{eq:work_model}, where the treated outcome is the sum of treatment effect $\Delta(X)$ and the control outcome, in both the modeling of \peiZero algorithm and the simulated data for evaluating performance. Here, to understand the performance of \peiZero when the data does not follow the working model~\eqref{eq:work_model}, we change the simulated data while keeping the same modeling in \peiZero. 

We specify the potential outcomes as 
\begin{align} \label{eq:work_model_interaction}
    Y_i^C = f(X_i) + U_i \text{ and }
    Y_i^T = L \cdot D(X_i) + (5 - L) D(X_i) f(X_i) + f(X_i) + U_i,
\end{align}
where 
\begin{align} \label{eq:bias-sparse}
    D(X_i) =  5X_i^3(3) \one\{X_i(3) > 1\} - X_i(1)/2,
\end{align} 
and $L \in (0, 5)$ encodes the degree of agreement with the additive model~\eqref{eq:work_model} --- $L = 5$ corresponds to the same simulation in Section~3.2 when $S_\Delta = 5$, and the data follows the additive model. Compared with the additive model~\eqref{eq:work_model} in the main paper, the treatment effect becomes a function of the potential control outcome $L \cdot D(X_i) + (5 - L) D(X_i) f(X_i)$. 

As expected, \peiZero has valid FDR control, and the power becomes higher when the simulated data is more aligned with the working model~\eqref{eq:work_model}. Also, the power of \peiZero is consistently higher than the linear-BH procedure when varying $L$, and the power does not decay much when working model~\eqref{eq:work_model} does not reflect the groundtruth data, i.e., $L = 0$ (see Figure~\ref{fig:work_model_interaction}).

\begin{figure}[H]
\centering
\hspace{1cm}
    \begin{subfigure}[t]{0.3\textwidth}
        \centering
        \includegraphics[width=1\linewidth]{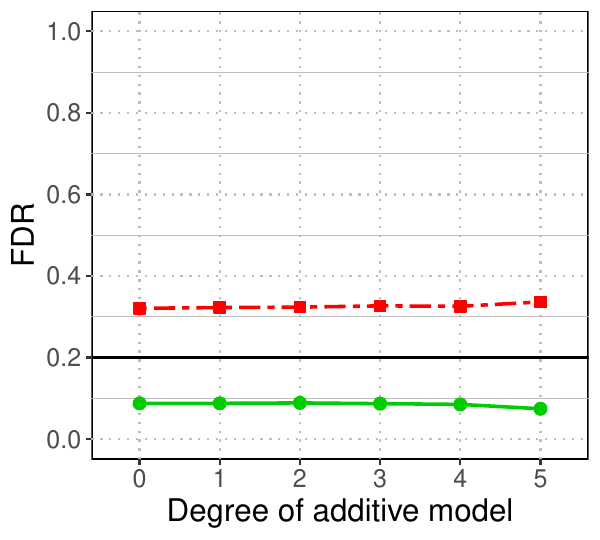}
    \end{subfigure}
    \hfill
    \begin{subfigure}[t]{0.3\textwidth}
    \centering
        \includegraphics[width=1\linewidth]{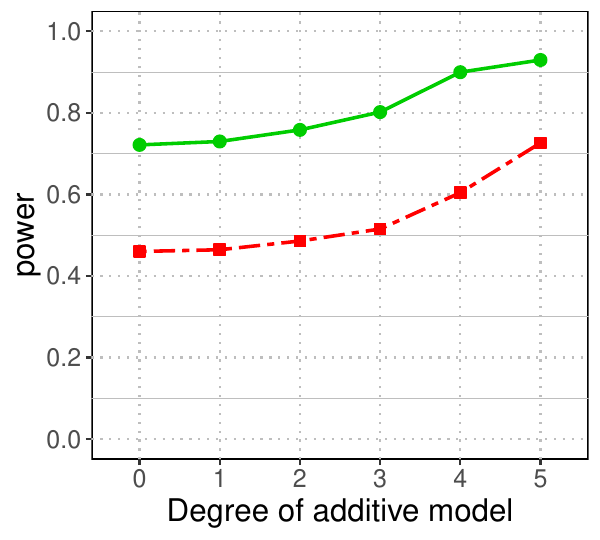}
    \end{subfigure}
\hspace{1cm}
\vfill
    \begin{subfigure}[t]{1\textwidth}
        \centering
        \includegraphics[width=0.7\linewidth]{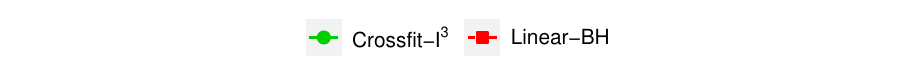}
    \end{subfigure}
    \caption{FDR (left) and power (right) of the \peiZero compared with the linear-BH procedure, with the potential outcomes specified as model~\eqref{eq:work_model_interaction} and the scale $L$ varying in $\{0,1,2,3,4,5\}$. The \peiZero controls FDR and can achieve high power even when the data does not follow the additive model~(7). The power is higher as data agrees more with the additive model, because \peiZero can model the treatment effect more accurately.}
    \label{fig:work_model_interaction}
\end{figure}

\noindent\textbf{Acknowledgements:}
We thank Edward Kennedy and Bikram Karmakar for their comments on an early draft of the paper. AR acknowledges support from NSF CAREER 1916320.

\bigskip
\noindent\textbf{Data Availability Statement:}
The datasets generated during and/or analysed during the current study are available in \url{https://github.com/duanby/I-cube}

\bibliographystyle{unsrtnat}
\bibliography{A_ref}

\end{document}